\documentclass[11pt,a4paper]{scrartcl}

\usepackage{amsmath}
\usepackage{amsfonts}
\usepackage{amsthm}
\usepackage[small,sanserif]{complexity}
\usepackage{mdframed}
\usepackage{fullpage}
\usepackage[utf8]{inputenc}
\usepackage{hyperref}
\hypersetup{
    unicode=true,
    pdftitle={A simpler proof of existence of quantum weak coin flipping with arbitrarily small bias},
    pdfauthor={Dorit Aharonov, Andr\'e Chailloux, Maor Ganz, Iordanis Kerenidis, Lo\"ick Magnin},
    colorlinks=true,
    linkcolor=red!50!black,          
    citecolor=blue!50!black,        
    urlcolor=blue!50!black          
}
\usepackage{microtype}
\usepackage[small,bf,sf]{caption}
\usepackage{pifont}
\usepackage{bm}
\usepackage{pgf,pgfarrows,pgfnodes}
\usepackage{tikz}
\usetikzlibrary{shapes.geometric}
\usetikzlibrary{shapes.misc}
\usetikzlibrary{shapes.symbols}
\usetikzlibrary{decorations.text}
\usetikzlibrary{positioning,backgrounds,shadows}
\usepackage{aliascnt}

\newtheoremstyle{sansthesis}%
{6pt}{6pt}{\itshape}{}{\bfseries\sffamily}{}{.7em}{}
\theoremstyle{sansthesis}
\newtheorem{thm}{Theorem}

\newtheorem{definition}{Definition}

\newaliascnt{lemcount}{definition}
\newtheorem{lem}[lemcount]{Lemma}
\aliascntresetthe{lemcount}

\newaliascnt{propcount}{definition}
\newtheorem{proposition}[propcount]{Proposition}
\aliascntresetthe{propcount}

\newaliascnt{claimcount}{definition}

\aliascntresetthe{claimcount}

\newtheorem{fact}[definition]{Fact}

\newaliascnt{corcount}{definition}
\newtheorem{cor}[corcount]{Corollary}
\aliascntresetthe{corcount}

\newaliascnt{remcount}{definition}

\aliascntresetthe{remcount}

\makeatletter
\def\tagform@#1{\maketag@@@{\ignorespaces#1\unskip\@@italiccorr}}
\let\orgtheequation\theequation
\def\theequation{(\orgtheequation)}
\makeatother

\newcommand{\ket}[1]{| #1 \rangle}
\newcommand{\bra}[1]{\langle #1 |}
\newcommand{\braket}[2]{\langle #1 | #2 \rangle}

\newcommand{\proj}[1]{| #1 \rangle\!\langle #1 |}
\newcommand{\triple}[3]{\langle{#1}|{#2}|{#3}\rangle}
\newcommand{\norm}[1]{\left\| #1 \right\|}
\newcommand{\abs}[1]{\left| #1 \right|}
\renewcommand{\sp}{\mathrm{sp}}
\DeclareMathOperator{\prob}{prob}
\DeclareMathOperator{\supp}{supp}
\newcommand{\ahalf}{\frac{1}{2}}
\newcommand{\eps}{\varepsilon}
\newcommand{\dt}[1]{\emph{#1}}
\renewcommand{\d}{\mathrm{d}}

\DeclareMathOperator{\tr}{Tr}
\DeclareMathOperator{\Tr}{Tr}
\DeclareMathOperator{\cl}{cl}
\DeclareMathOperator{\interior}{int}
\newcommand\I{\mathbb{I}}
\newcommand\N{\mathbb{N}}
\newcommand{\fpp}{f^{\prime\prime}}

\renewcommand{\AA}{\mathcal{A}}
\newcommand{\BB}{\mathcal{B}}

\newcommand{\HH}{\mathcal{H}}

\newcommand{\MM}{\mathcal{M}}

\renewcommand{\PP}{\mathcal{P}}

\renewcommand{\R}{\mathbb{R}}

\newcommand{\COMMENT}[1]{}

\newcommand{\dnote}[1]{\textcolor{blue} {{\textbf{(Dorit:}}
#1\textbf{)}}}

\title{A simpler proof of existence of quantum weak coin flipping with arbitrarily small bias}
\author{
	Dorit Aharonov\thanks{The School of Computer Science and Engineering, The Hebrew University of Jerusalem, Israel},
	Andr\'e Chailloux\thanks{SECRET Project --- INRIA Rocquencourt, 78153 Le Chesnay Cedex, France},
	Maor Ganz$^*$\!\!, \\
	Iordanis Kerenidis\thanks{LIAFA, Universit\'e  Paris Diderot; CNRS} , and 
	Lo\"ick Magnin$^*$}
\begin{document}

\maketitle

\begin{abstract}
Mochon's proof \cite{Moc07} of existence of quantum weak coin flipping with arbitrarily  small bias is a fundamental result in quantum cryptography,  but at the same time one of the least understood. 
Though used several times as a black box in important follow-up results  \cite{Ganz09, CK09, AS10, CK11,KZ13} the result has not been peer-reviewed,  its novel techniques (and in particular Kitaev's point game formalism)  have not been applied anywhere else, and an explicit protocol is missing. We believe that truly understanding the existence proof and the novel techniques it relies on would constitute a major step in quantum information theory, leading to deeper understanding of entanglement and of quantum protocols in general. In this work, we make a first step in this direction. We simplify parts of Mochon's construction considerably, making about $20$ pages of analysis in the original proof superfluous, clarifying some other parts of the proof on the way, and presenting the proof in a way which is conceptually easier to grasp.  We believe the resulting proof of existence is easier to understand, more readable, and certainly verifiable.  Moreover, we analyze the resources needed to achieve a bias $\eps$ and show that the number of qubits is $O(\log \frac{1}{\eps})$, while the number of rounds is $\left(\frac{1}{\eps}\right)^{O(\frac{1}{\eps})}$. A true understanding of the proof,  including Kitaev's point-game techniques and their applicability, as well as completing the task of constructing an explicit (and also simpler and more efficient)  protocol, are left to future work. 
\end{abstract}

\section*{Introduction}
\phantomsection
\addcontentsline{toc}{section}{Introduction}

Coin flipping is a cryptographic primitive that enables two distrustful and far apart parties, Alice and Bob, to create a random bit that remains unbiased even if one of the players tries to force a specific outcome.
It was first proposed by Blum~\cite{Blu83} and has since found numerous applications in two-party secure computation~\cite{Gol09}. 
In the classical world, coin flipping is possible under computational  assumptions, such as the hardness of the factoring or the discrete log problems. However, in the information theoretical setting,  i.e., without any computational assumptions, it has been shown by Cleve~\cite{Cle86} that in any classical protocol, one of the players can always force his or her desired outcome with probability $1$.

Quantum information has given us the opportunity to revisit the notion of information theoretical security in cryptography. 
The first breakthrough result was a protocol of Bennett and Brassard~\cite{BB84} that showed how to securely distribute a secret key between two players in the presence of an omnipotent eavesdropper.
Thenceforth, a long series of work has focused on which other cryptographic primitives are possible with the help of quantum information.
Unfortunately, the subsequent results were not positive.
Mayers \cite{May97} and Lo and Chau \cite{LC97} proved the impossibility of secure quantum bit commitment and oblivious transfer and consequently of any type of two-party secure computation~\cite{May97,LC97,DKSW07}.
However, several weaker variants of these primitives have been shown to be possible~\cite{HK04,BCH+08}.

The case of coin flipping is one of the most intriguing primitives in this respect.
Even though the results of Mayers and of Lo and Chau  show that information theoretically secure perfect coin flipping (i.e. where the resulting coin is perfectly unbiased) is impossible also in the quantum world,  they left open the question of whether one can construct a quantum coin flipping protocol where no player could bias the coin with probability arbitrarily close to $0$.
The subject of this paper is exactly this question; we start with some  historical background. 

\paragraph{Quantum coin flipping}
We begin with a more precise definition of coin flipping. 
Two variants of quantum coin flipping (CF) have been  studied: \emph{strong coin flipping} and \emph{weak coin flipping}. 
A  strong coin  flipping protocol with bias $\eps$ is a protocol in which Alice and Bob exchange messages such that the following holds. 
First, if both  players follow the protocol, then they agree on the outcome and the outcome is $0$ or $1$ with equal probability.
Moreover, it is guaranteed that neither  Alice nor Bob can force the outcome $0$ or $1$ with probability more than $1/2+\eps$, if they try to cheat.
In other words, no dishonest  player (playing against an honest player) can bias the coin towards {\it any} of the outcomes with probability higher than $\eps$.
For weak coin flipping (WCF), Alice and Bob have an \textit{a priori} desired coin outcome. 
In other words the two values of the coin can be thought of as `Alice wins' and `Bob wins'. 
A weak coin flipping protocol with bias $\eps$ guarantees that no dishonest player (playing against an honest player) can bias the coin towards his or her desired outcome with probability 
greater than $\eps$.
The subtle difference between the weak and strong CF versions seems unimportant at first sight; indeed, in the classical setting it does not make a difference.
In the quantum world, however, the two are very different. 
Note that obviously, strong CF implies weak CF with the same bias. 

Aharonov, Ta-Shma, Vazirani, and Yao~\cite{ATVY00}  provided the first quantum strong CF protocol with bias
bounded below $1/2$; strictly speaking, their bias was $\eps< 0.4143$.
Then, Ambainis~\cite{Amb04} described a quantum strong CF protocol with an improved bias of $1/4$. 
Subsequently, a number of different protocols have been  proposed~\cite{SR01,NS03,KN04} that achieved the same bound of $1/4$.

On the other side, in a breakthrough result, Kitaev~\cite{Kit03} proved a lower bound on the best possible bias 
of any strong CF protocol.
Using a formulation of quantum CF protocols as semidefinite programs, and the duality of semidefinite programming, he showed that the bias of any strong CF protocol is bounded from below by $1/\sqrt{2}-1/2$ (for a proof see e.g.~\cite{ABDR04}).
Kitaev's result rules out the existence of strong quantum CF protocols with arbitrarily small bias. 
Historically, this result had a positive, rather than a negative effect; it highlighted the fact that the difference between weak and strong CF is meaningful in the quantum setting, since its proof does not apply to the weak case.
Hence, this negative result in fact gave hope that a quantum WCF protocol with arbitrarily small bias might exist. 

Around that time, a series of works started to provide 
better and better understanding of 
WCF. 
First, Spekkens and Rudolph~\cite{SR02} constructed a WCF protocol with bias $1/\sqrt{2}-1/2$. This is  
a strange coincidence, since Kitaev's lower bound of $1/\sqrt{2}-1/2$ 
applies solely for strong CF protocols and not for weak ones. 
Ambainis~\cite{Amb04} then proved that the number of rounds of 
communication between Alice and Bob for achieving a bias of $\eps$ in a 
(weak or strong) CF protocol is lower-bounded by $\Omega(\log\log 1/\eps)$, 
and thus a WCF protocol with arbitrarily small bias, if exists, 
cannot be achieved with a constant number of rounds.
Mochon then described a WCF 
protocol with bias $0.192$~\cite{Moc04}; 
this was the first result in which Kitaev's lower bound on strong 
CF was broken by an explicit WCF protocol. 
Mochon later showed that this protocol was a member of a 
family of protocols, the best ones achieving a bias of $1/6$~\cite{Moc05}.
Finally, in a breakthrough result, Mochon
resolved the question of the existence of near-perfect 
quantum weak coin flipping to the affirmative, 
and proved the existence of a protocol with bias $\eps$ for 
any $\eps>0$ ~\cite{Moc07}. 

The work of Mochon is a major advance not only because of its result, 
which resolved an intriguing question which was open for a long time, 
but also, and perhaps mainly, because of its techniques. 
The central tool, used also in 
Mochon's earlier work~\cite{Moc05} is a formalism due to Kitaev, of 
\textit{point games}. In this formalism, a protocol which is a sequence 
of unitaries to be applied by Alice and Bob in turns,
is viewed as a semidefinite program; the dual of this program
provides a bound on the 
security of the protocol; and, most importantly, 
the pair of primal and dual together   
are represented by {\it a sequence of sets of points on the $2$-dimensional 
plane}, called a {\it point game}.  
The aim is then to construct such point games whose 
parameters correspond to a protocol with arbitrarily small bias. 
This project is highly complicated and was approached by Mochon 
using further ideas 
related to operator monotone functions; in his work achieving the almost 
perfect protocol \cite{Moc07}, 
Mochon further develops the techniques to time {\it independent} point games, 
(which he attributes to Kitaev as well); he   
applies quite heavy analysis to derive time independent point games 
with bias arbitrarily close to zero. 

There have been several interesting 
applications of this result so far. 
Ganz \cite{Ganz09} and Aharon and Silman \cite{AS10} 
derived a quantum leader election protocol 
of logarithmically many rounds, using Mochon's protocol as a subroutine.  
Leader election is a cryptographic primitive which generalizes CF 
to $n$ players who need to choose a leader among them with equal probability. 
Chailloux and Kerenidis \cite{CK09} used the result of almost perfect WCF to 
derive an optimal strong coin flipping protocol, of the best 
bias possible, namely $1/\sqrt{2}-1/2 + \delta$ for arbitrarily small 
$\delta$. They do this using   
a classical reduction from any weak CF protocol with 
bias $\eps$ to a strong CF protocol 
with bias at most $1/\sqrt{2}-1/2 + 2\eps$. 
This closed the question of strong CF since 
this result is tight, given Kitaev's lower bound on strong CF. 
Mochon's weak CF protocol has 
also beed used in a quantum reduction to derive an optimal 
quantum bit commitment protocol~\cite{CK11}. 
More recently, this protocol has been used to allow two players to
achieve correlated equilibria of games without any computational assumptions and without a mediator~\cite{KZ13}.

Undoubtedly, Mochon's result is a fundamental result in quantum cryptography. 
Nevertheless, it remained so far one of the least understood results. 
In particular, the paper was not peer-reviewed, and no explicit protocol 
is known or had been derived from this existence proof. Moreover,  
all the applications of the result mentioned above use the result 
as a black box, and its novel and beautiful
techniques were never given another interpretation, or applied 
in any other context. 
Perhaps this lack of understanding, almost $7$ 
years after the result was 
derived, can be attributed at least in part to the fact that  
Mochon's paper is $80$-page long and extremely technical. The paper 
delves into semidefinite programming duality, operator monotone functions, perturbation theory,
time dependent and time independent point games and, in the end, 
arrives at an existence proof of a protocol which is non-explicit and 
cannot be described in a simple interactive manner between Alice and Bob.
  
We believe that an understanding of the proof of Mochon's breakthrough result 
and its techniques 
would constitute a major advance in quantum information theory,
quantum cryptography, and 
quantum protocols. Possibly, it will also deepen  
our understanding of the flippant notion of quantum entanglement,   
since in essence, quantum CF is simply a protocol 
to distribute entanglement (more precisely, an EPR state) 
between two parties, where at least one of them is honest. 
Such an understanding of the result 
should probably be demonstrated by an understandable 
and an explicit description of a CF protocol.  
The contribution of this paper is far from achieving these important goals, 
but we believe it constitutes an important step, in that it provides a simpler, 
shorter proof, which is readable and verifiable, and also whose structure 
is conceptually comprehensible.  
We proceed to provide an overview of the proof and subsequently describe where and
in what ways this proof deviates from Mochon's proof.

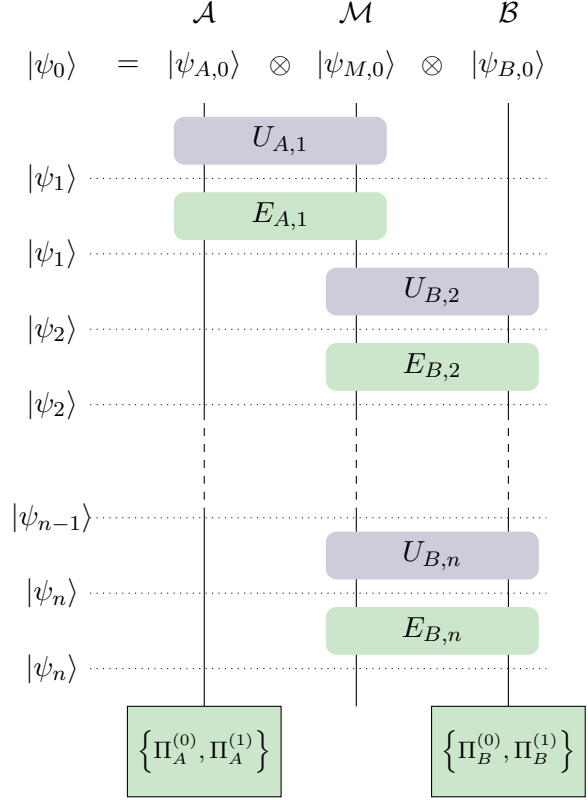
\begin{figure}
\begin{center}
\begin{tikzpicture}[
	unitary/.style={rectangle, rounded corners,  fill=blue!30!black!20, very thick, minimum width=2.8cm},
	proj/.style={rectangle, rounded corners,  fill=green!50!black!20, very thick, minimum width=2.8cm}]
	
	\newcommand{\xa}{2}
	\newcommand{\xm}{4}
	\newcommand{\xb}{6}
	\newcommand{\xam}{3}
	\newcommand{\xmb}{5}

	\draw (\xa,-1) -- (\xa,-5.1); \draw[dashed] (\xa,-5.1) -- (\xa,-6.4); \draw (\xa,-6.4) -- (\xa, -9);
	\draw (\xm,-1) -- (\xm,-5.1); \draw[dashed] (\xm,-5.1) -- (\xm,-6.4); \draw (\xm,-6.4) -- (\xm, -9);
	\draw (\xb,-1) -- (\xb,-5.1); \draw[dashed] (\xb,-5.1) -- (\xb,-6.4); \draw (\xb,-6.4) -- (\xb, -9);

	\node at (\xa,.2) {$\AA$}; \node at (\xm,.2) {$\MM$}; \node at (\xb,.2) {$\BB$};
	
	\node at (0,-.5) {$\ket{\psi_0}$};  \node at (1,-.5) {$=$};
	\node at (\xa,-.5) {$\ket{\psi_{A,0}}$}; \node at (\xam,-.5) {$\otimes$};
	 \node at (\xm,-.5) {$\ket{\psi_{M,0}}$}; \node at (\xmb,-.5) {$\otimes$};
	 \node at (\xb,-.5) {$\ket{\psi_{B,0}}$};

	\node[unitary] at (\xam, -1.5) {$U_{A,1}$};
	\node[proj] at (\xam, -2.5) {$E_{A,1}$};
	\node at (0,-2) {$\ket{\psi_1}$}; \draw[dotted] (.5,-2) -- (6.5,-2);
	\node at (0,-3) {$\ket{\psi_1}$}; \draw[dotted] (.5,-3) -- (6.5,-3);

	\node[unitary] at (\xmb, -3.5) {$U_{B,2}$}; 
	\node[proj] at (\xmb, -4.5) {$E_{B,2}$};
	\node at (0,-4) {$\ket{\psi_2}$}; \draw[dotted] (.5,-4) -- (6.5,-4);
	\node at (0,-5) {$\ket{\psi_2}$}; \draw[dotted] (.5,-5) -- (6.5,-5);
	
	\node[unitary] at (\xmb, -7) {$U_{B,n}$}; 
	\node[proj] at (\xmb, -8) {$E_{B,n}$};
	\node at (0,-6.5) {$\ket{\psi_{n-1}}$}; \draw[dotted] (.5,-6.5) -- (6.5,-6.5);
	\node at (0,-7.5) {$\ket{\psi_n}$}; \draw[dotted] (.5,-7.5) -- (6.5,-7.5);
	\node at (0,-8.5) {$\ket{\psi_n}$}; \draw[dotted] (.5,-8.5) -- (6.5,-8.5);

	\node[rectangle, draw, minimum width=1.8cm,minimum height=1.2cm, fill=green!50!black!20] at (\xa,-9.6) {\footnotesize{$\left\{\Pi_A^{(0)} ,\Pi_A^{(1)}\right\}$}};
	\node[rectangle, draw, minimum width=1.8cm,minimum height=1.2cm, fill=green!50!black!20] at (\xb,-9.6) {\footnotesize{$\left\{\Pi_B^{(0)} ,\Pi_B^{(1)}\right\}$}};
\end{tikzpicture}
\end{center}
\caption{Representation of a coin flipping protocol. Alice and Bob start with a separable state on the spaces $\AA \otimes \MM \otimes \BB$. At every odd round $i$ Alice applies a unitary $U_{A,i}$ and a  projection $E_{A,i}$ on the space $\AA \otimes \MM$ and at every even round $i$ Bob applies a unitary $U_{B,i}$ and a  projection $E_{B,i}$ on the space $\MM \otimes \BB$. At the end, they measure their private registers to obtain their final outcomes.}
\label{fig:cfprotocol}
\end{figure}

\paragraph{Overview of the proof}\ \\

\noindent
\textit{Step 1: SDPs and dual feasible points.}\quad 
A CF protocol is defined as in \autoref{fig:cfprotocol}. The first step involves presenting such an interactive protocol in terms of a semidefinite program. Define Bob's optimal 
cheating probability $P_B^*$ to be the maximum probability (over any possible strategy of Bob) that
Alice outputs ``1'' when she plays honestly, that is, 
she declares Bob the winner of the coin flip. Likewise, define 
$P_A^*$ to be the maximum probability (over any possible strategy of Alice) 
that Bob declares Alice to be the winner and outputs ``0''.  

We can write $P_B^*$ as 
the value of a semidefinite program as follows. For each round $i$, we consider a variable $\rho_{AM,i}$, which is the reduced density matrix of the state at time $i$ for the union of Alice's and the message qubits. The constraints on those variables are of two types. At her turn, Alice honestly applies the operation $U_{A,i}$ (as well as a projection $E_{A,i}$, which we add for technical reasons). This happens at every odd round $i$. 
On even rounds, we only require that Alice's state on the space $\AA$ does not change; 
this follows from the fact that Bob is the one to operate at those steps. We now optimize over all variables that satisfy these constraints; this corresponds to Bob optimizing his operations, in order to maximize the probability that at the end of the protocol Alice gets the outcome ``1'', when she performs her final projection. 
(\autoref{thm:primal}). We can of course write $P_A^*$ similarly as a semindefinite program 
with variables involving Bob's reduced density matrices. 

This primal formulation is, unfortunately, not suitable for proving upper bounds on the cheating probabilities $P_A^*$ and $P_B^*$, since they are defined by a maximization: any set of matrices $\{\rho_i\}$ that satisfies the constraints will lead to a lower bound on the cheating probabilities, rather than an upper bound. 

We therefore consider the dual of these semidefinite programs (\autoref{thm:dual}). For each primal constraint, i.e. for each protocol round, we define a dual variable $Z_{A,i}$, which is a positive semidefinite matrix that satisfies certain constraints (that arise from dualizing the primal constraints). 
Since the constraints are on matrices, so are the variables $Z_{A,i}$ of this dual program.
Any solution of the dual program is referred to as
a  \emph{dual feasible point}, and can be shown to provide an 
\textit{upper bound} on Bob's 
optimal cheating probability, by the 
duality of semidefinite programming. In fact, it turns out that 
$P_B^*\leq \triple{\psi_0}{Z_{A,0} \otimes \I_\MM \otimes \I_\BB}{\psi_0}$, where $\ket{\psi_0}$ is the initially shared separable state.   
In other words, a dual feasible point, which is a set of matrices 
$\{Z_{A, i}\}$ satisfying certain constraints, can be seen as a witness to the 
security of the protocol against a cheating Bob.
Similarly, we can write a dual SDP for Alice's cheating probability $P_A^*$ and upper bound $P_A^*$ by a dual feasible point $\{Z_{B,i}\}$ as $P_A^*\leq \triple{\psi_0}{\I_\AA \otimes \I_\MM \otimes Z_{B,0}}{\psi_0}$.

\COMMENT{
Note that taking the dual creates a ``reversal of time''. The constraints
in the primal express what states $\rho_{AM,i+1}$ are allowed given the state $\rho_{AM,i}$ and starting from the fixed state $\rho_{AM,0}$ that corresponds to the joint state at the start of the protocol. On the other hand, the constraints of the dual express what positive matrices $Z_{A,i-1}$ are allowed given the matrix $Z_{A,i}$ 
\dnote{I don't think this is explained here at all. I didn't understand it in any case from the above explanation. Can we explain it or omit this paragraph from here and move it to where it can be explained? It is an important point and we want it to be clear}
and starting from the fixed final matrix $Z_{A,n}=\Pi_{A}^{(1)}$ that corresponds to the final projection of the protocol where Alice output ``1''. 
}

Our goal is to find a WCF protocol, 
together with two dual feasible points (one for a cheating Alice 
and one for a cheating Bob) which would
provide a certificate for the security of the protocol. 
Ideally, we would like to optimize over all such protocols to 
find the best CF protocol, and find the best possible bias. 
However, even though given a protocol one can find an upper bound on the cheating probability via the dual semidefinite program, this formulation of coin flipping does not provide any intuition on what type of protocols one should be looking for.
The solution proposed by Mochon, following Kitaev, consists of finding a different but 
equivalent representation of the problem. 


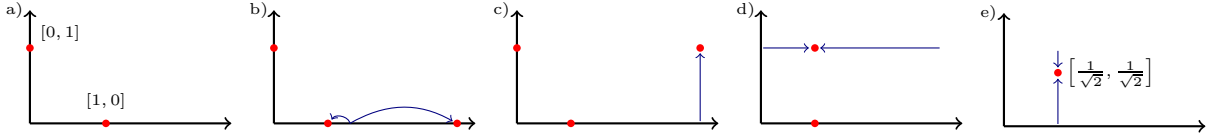
\begin{figure}
\begin{tikzpicture}
	\node at (-.2, 1.5) {\tiny{a)}};
	\draw[->,thick] (0,0) -- (0,1.5);
	\draw[->,thick] (0,0) -- (2.65,0);
	\fill[color=red] (0,1) circle (.05);
	\fill[color=red] (1,0) circle (.05);
	\node at (.4,1.2) {\tiny{$[0,1]$}};
	\node at (1,.3)  {\tiny{$[1,0]$}};
\end{tikzpicture}
\begin{tikzpicture}
	\node at (-.2, 1.5) {\tiny{b)}};
	\draw[->,thick] (0,0) -- (0,1.5);
	\draw[->,thick] (0,0) -- (2.65,0);
	\fill[color=red] (0,1) circle (.05);
	\fill[color=red] (.71,0) circle (.05);
	\fill[color=red] (2.41,0) circle (.05);
	\draw[->,blue!50!black] (1,0) to [out=100,in=60] (.76,0.05);
	\draw[->,blue!50!black] (1,0) to [out=30, in = 150] (2.35,0.03);
\end{tikzpicture}
\begin{tikzpicture}
	\node at (-.2, 1.5) {\tiny{c)}};
	\draw[->,thick] (0,0) -- (0,1.5);
	\draw[->,thick] (0,0) -- (2.65,0);
	\fill[color=red] (0,1) circle (.05);
	\fill[color=red] (.71,0) circle (.05);
	\fill[color=red] (2.41,1) circle (.05);
	\draw[->,blue!50!black] (2.41,0.03) -- (2.41, .93);
\end{tikzpicture}
\begin{tikzpicture}
	\node at (-.2, 1.5) {\tiny{d)}};
	\draw[->,thick] (0,0) -- (0,1.5);
	\draw[->,thick] (0,0) -- (2.65,0);
	\fill[color=red] (.71,0) circle (.05);
	\fill[color=red] (.71,1) circle (.05);
	\draw[->,blue!50!black] (0.03,1) -- (.63,1);
	\draw[->,blue!50!black] (2.35,1) -- (.78,1);
\end{tikzpicture}
\begin{tikzpicture}
	\node at (-.2, 1.5) {\tiny{e)}};
	\draw[->,thick] (0,0) -- (0,1.5);
	\draw[->,thick] (0,0) -- (2.65,0);
	\fill[color=red] (.71,.71) circle (.05);
	\draw[->,blue!50!black] (.71,.03) -- (.71,.63);
	\draw[->,blue!50!black] (.71,1) -- (.71,.78);
	\node at (1.4,.71) {\tiny{$\left[\frac{1}{\sqrt2},\frac{1}{\sqrt2}\right]$}};
\end{tikzpicture}
\caption{A ``simple'' point game for the \cite{SR02} protocol with bias $1/\sqrt{2}-1/2$. The game starts with the uniform distribution over two points. The first transition, between a) and b), is a horizontal transition as the point on the x-axis is split into two points. The second transition, between b) and c) is a vertical transition as one point is raised vertically. The last two transitions are two merges, respectively horizontally and vertically. We omitted the weights of the points in the distributions to simplify the drawings.}
\label{fig:examplePG}
\end{figure}

{~}

\noindent
\textit{Step 2: Point games with EBM transitions.}\quad 
Kitaev~\cite{Moc07} defined a  
graphical way of representing a WCF protocol accompanied with two dual feasible points, one for Alice and one for Bob. He called this representation 
\textit{point games}. At a high level, 
a point game is an ordered sequence of (possibly unnormalized) probability distributions 
supported each on some finite set of points on 
the $2$-dimensional plane.
(See for example \autoref{fig:examplePG}). 



How are point games connected to WCF protocols and their two dual feasible points? 
Recall that we have bounded $P_B^*$ by the quantity $\triple{\psi_0}{Z_{A,0} \otimes \I_\MM \otimes \I_\BB}{\psi_0}$ and likewise $P_A^*\leq \triple{\psi_0}{\I_\AA \otimes \I_\MM \otimes Z_{B,0}}{\psi_0}.$
The point game we associate with the protocol and its dual feasible points is in fact a 2-dimensional representation of the evolution of the above quantities during the protocol. 
More precisely, we consider the expression  
$\bra{\psi_{n-i}}Z_{A,n-i} \otimes\I_\MM\otimes Z_{B,n-i}\ket{\psi_{n-i}}$ for 
$i$ going from $0$ to $n$.  
For each $i$, we associate to the above expression 
a distribution over points on the plane as follows. 
Consider the measurement of the honest state $\ket{\psi_{n-i}}$ by the
observable $Z_{A,n-i} \otimes \I_\MM \otimes
Z_{B,n-i}$. The possible outcomes can be identified with pairs
of eigenvalues of the form $[z_A,z_B]$ (where $z_A$ is an eigenvalue of
$Z_{A,n-i}$ and likewise $z_B$ is an
eigenvalue of $Z_{B,n-i}$); the weight assigned to such a pair is the
projection of the honest state at
time $(n-i)$ onto the eigenspace corresponding to these eigenvalues.
This associates with any protocol of $n$ rounds and its dual feasible points 
a sequence of size $n$ of probability distributions over points in the plane. 

A key point is 
to classify what kind of distributions and transitions can originate from WCF protocols and two dual feasible points.  
As a start, it turns out  that the point game that one derives
this way must start with an initial uniform distribution over the
 two points $[0,1]$ and $[1,0]$; and a final set of points that consists of 
a single point $[\beta, \alpha]$. The initial uniform distribution corresponds to the fact that for an honest protocol both players agree that Alice wins with probability $1/2$ (the $[1,0]$ point) and they both agree that Bob wins with probability $1/2$ (the $[0,1]$ point). Moreover, the coordinates of the final point provide upper bounds on the 
optimal cheating probabilities, $P_A^* \leq \alpha$ and $P_B^* \leq \beta$. 

Our goal is now to find the exact rules that the distributions and transitions between them must satisfy if they arise from a WCF and two dual feasible points. 
The idea is that if we have such a point game which satisfies these rules, we can 
(at least in principle) derive a protocol and a security guranteee from it.  
\COMMENT{ define formally a class of point games, which we will call {\em expressible by matrices} or EBM point games, that include the above mentioned point games and more importantly are equivalent to them. In other words, we would like that given an EBM point game we can construct a protocol and its dual feasible points. }

\COMMENT{
\textit{From EBM point games back to protocols.}\quad

The above definition of EBM point games
suggests that to design a WCF protocol, one 
may design an EBM point game. For that, we need to understand what kinds 
of point games emerge as a result of WCF protocols with dual feasible 
points, and also how to find such a corresponding WCF 
protocol given an EBM point game. 
Let us start with the first question.

We would, first,  like to understand what kind of transitions should be allowed, in other words, }

To understand what transitions may occur in a point game that arises from a WCF protocol and its dual feasible points, consider two rules.  
First, the fact that at each round of the protocol only Alice or Bob act non-trivially 
on the state, is translated to the dual constraints $Z_{A,i-1}=Z_{A,i}$ for even $i$ and  $Z_{B,i-1}=Z_{B,i}$ for odd $i$. This implies that at even steps the points are redistributed along 
vertical lines whereas at odd steps they are redistributed along horizontal lines (See \autoref{fig:examplePG}).  

The second rule describes how a set of points may move along the same vertical or horizontal line during one step. This is solely derived from the requirement that the operations of the cheating players must be quantum operations.
More specifically, let us describe how a set of points $S$ with first coordinates $\{x_j\}_{j \in S}$, the same second coordinate $y$ and weights $\{w_j\}_{j \in S}$ can transition to a set of points $S'$ with first coordinates $\{x'_k\}_{k \in S'}$, the same second coordinate $y$ 
and weights $\{w'_k\}_{k \in S'}$, i.e. a horizontal transition. 
Let us represent the first set of points and its distribution by the function $l$ with finite support such that $l(x_j) = w_j$, and $l(x)=0$ everywhere else; and the second set of point and its distribution by the function $r$ with finite support such that $r(x'_k) = w'_k$, and $r(x) = 0$ everywhere else. 

\begin{definition}\label{def:ebmtransition} 
We say a transition from $l$ to $r$ is a horizontal EBM (Expressible by Matrix) line transition if there
exist two semidefinite positive matrices $0 \preceq X \preceq Y$ 
and a (not necessarily unit) vector $\ket{\psi}$ such that
\begin{align*}
	l(x) = \left\{ \begin{array}{cl}
		\bra{\psi} \Pi_X^{[x]} \ket{\psi} & \text{if}\ x \in \sp(X) \\
		0 & \text{otherwise}
	\end{array} \right.
	\quad \text{and}\quad 
	 r(x) = \left\{ \begin{array}{cl}
		\bra{\psi} \Pi_Y^{[x]} \ket{\psi} & \text{if}\ x \in\sp(Y) \\
		0 & \text{otherwise}
	\end{array} \right.
\end{align*}
where $\Pi_X^{[x]}$ is the projector 
onto the eigenspace of $X$ of eigenvalue $x$ and $\sp()$ denotes the spectrum of the matrix.

A transition from $p$ to $q$ is a horizontal EBM transition if it is a horizontal EBM line transition on every horizontal line. Vertical EBM transitions are defined by symmetry. An EBM point game consists of a sequence of EBM transitions.
\end{definition}

It is rather easy to prove that the point games that arise from a WCF protocol and its dual feasible points are in fact EBM point games. The matrices $X$ and $Y$ will be defined through the dual feasible points  $\{Z_{A,i}\}$ and $\{Z_{B,i}\}$ and the constraint $X \preceq Y$ will immediately follow from the dual constraints.

More importantly, the reverse implication is also true (\autoref{thm:EBMReverse}): if there exists an EBM point game with initial uniform distribution $\frac{1}{2}[0,1]+\frac{1}{2}[1,0]$ and a final distribution concentrated on the point $[\beta,\alpha]$, then there exists a WCF protocol with optimal cheating probabilities $P_A^* \leq \alpha$ and $P_B^* \leq \beta$ and dual feasible points witnessing these upper bounds. 

We note that the proof of this theorem is non-constructive. An EBM point game implies that for every line transition there exist matrices $X$ and $Y$ and a vector $\ket{\psi}$ that would witness the fact that it is an EBM line transition. However, we do not know yet of any algorithm better than brute-force for finding these matrices and vector, even when we know that they exist. Note that once these matrices and vectors are known, then we can efficiently construct a WCF protocol.  

The equivalence between EBM point games and WCF protocols together with their dual feasible points, is still a bit of a mystery. This can in fact be tracked back to Kitaev's proof of the lower bound on strong coin flipping, which is far simpler but still contains the same ``magic''. 

Nevertheless, we have reduced the question of existence of a WCF protocol with bias $\eps$ to that of the existence of an EBM point game with final point $[1/2+\eps,1/2+\eps]$.

{~}

\textit{Step 3: Point games with valid transitions.}\quad 
Unfortunately, finding an EBM point game does not seem to be an easy task; 
yet another reduction is required. We will use a different characterisation of EBM transitions which is easier to work with.
Let us first make some small technical detour and 
shift the view from \textit{transitions} to \textit{functions}. 
We denote the EBM transition from $l$ to $r$ (who are non-negative functions with finite support)
by the \textit{EBM function} $(r-l)$, which is also a function with finite support, but this time it can
have positive as well as negative values.
EBM functions have an interesting geometrical  
property: they form a convex cone. We would like to now represent EBM functions using a different language, and for that we will use duality of convex cones. 
Importantly, the dual of the cone of EBM functions is a well known object: 
it turns out to be the set of {\it operator monotone functions} 
\cite{Bha97}. We recall that a function
$f:\R\mapsto \R$ is said to be operator monotone if it  
preserves the order of positive definite matrices, namely, if $X \preceq Y$ as PSDs, 
then $f(X) \preceq f(Y)$ (also as PSDs; a function is applied on a PSD 
in the usual way of diagonalizing and applying the function on the 
eigenvalues). Consider now yet again the dual of this set. 
We call the dual of the set of operator monotone functions  
\dt{set of valid functions}. We now use the following basic fact from convex geometry: 
$C^{**}$ is the closure of the smallest convex cone containing $C$. Since in our case $C$, the set 
of EBM functions, is itself a cone, 
this means that the sets of valid functions and EBM functions are the 
same, up to closures.  Hence, we can show that from any point game with valid transitions, we can construct an EBM point game whose final point is arbitrarily close to the original one (\autoref{thm:ValidReverse}). 
The map from EBM functions to valid functions, being essentially equal sets up to closures, would have been very easy had we been considering a finite dimensional space. Unfortunately, the space of functions we are handling is infinite dimensional and moreover, somewhat pathological, and hence the proof of \autoref{thm:ValidReverse} requires some technical effort. Still, in spirit, the idea is that we move to the dual of the dual and by that we gain the advantage of considering the dual of the well studied object of operator monotone functions; this turns out useful later on in the actual construction of the final point game.    

The advantage is that valid  functions have a very simple analytical
characterization that follows from the characterisation of operator monotone functions (see \autoref{lem:DecompositionOMF}), and hence, checking that a given function $h$ is valid 
corresponds to checking two simple mathematical statements :
\begin{align*}
	\sum_x h(x) = 0
	\quad\text{and}\quad
	 \forall\lambda > 0,\ \sum_x \frac{h(x)}{\lambda+x} \leq 0.
\end{align*} 

Even though verifying a valid point game is rather straightforward, finding valid point games remains again a bit of a mystery. Despite the fact that there are strong tools that enable doing this (mainly \autoref{thm:WeightPolynomial}) there still seems to be missing an intuitive interpretation of what valid transitions are and how to construct them.

Nevertheless, finding a WCF protocol with arbitrarily small bias has been reduced
to the problem of finding a point game with \textit{valid} transitions, which
ends at the point $[\beta,\alpha]$ with both $\alpha, \beta$ arbitrarily
close to $1/2$. 

{~}

\textit{Step 4: Time independent point game.} \quad The introduction of valid 
transitions makes it easier to check that a point game is valid or not. 
But such verification will become very tedious for point games with 
a bias $\eps$ which is arbitrarily small,  
as the number of transitions tends to infinity when $\eps$ tends to $0$. 
More importantly, we need a tool that will help us find constructively 
such valid point games, rather than verifying that a given game is valid. 

The last model in our sequence of reductions, also introduced by Kitaev,  
is called \textit{time independent point games} (TIPG), and essentially addresses this problem. 
The key observation  is that if $f_1$ and $f_2$ are valid functions, either both horizontal or 
both vertical,  then $f_1+f_2$ is also a valid function. Hence, given a valid point game
 with valid 
horizontal functions $\{h_1,h_2,\dots,h_n\}$ and valid vertical functions $\{v_1,v_2,
\dots,v_n\}$ 
and final point $[\beta,\alpha]$, we can define $h = \sum_i h_i$ and $v =\sum_i v_i$, which are valid horizontal and vertical functions respectively. Moreover, we will see that if we sum these two functions then everything cancels apart form the initial and final points. In other words, we have
\begin{align*}
h + v = [\beta,\alpha] - \frac{1}{2}[1,0] - \frac{1}{2}[0,1].
\end{align*}
Hence, a time independent point game is defined
 as one valid horizontal function $h$, one valid vertical function $v$, and one point $[\beta,\alpha]$, that satisfy together the above equation. Verifying that a TIPG is correct involves checking that only two functions are valid, this is in sharp contrast to point games with valid transitions.

It turns out that the reverse direction is also ``approximately'' true. 
This is the direction we will be interested in since we are about to design a TIPG 
and argue that it implies a valid point game. 
The approximate claim is that for all $\eps>0$, a TIPG with final point $[\beta,\alpha]$ and valid 
functions $h$ and $v$ can be turned into a point game with valid transitions and final point $[\beta+\eps, \alpha+\eps]$ (\autoref{prop:TipgToTdpg}), whose number of transitions depends on $\eps$. 
Ideally we would like to start with the TIPG with final point $[\beta,\alpha]$ and exhibit a point game with valid transitions, with initial set of points $\frac{1}{2}[1,0] + \frac{1}{2}[0,1]$ and final point $[\beta,\alpha]$. Unfortunately, we do not know how to do this.  However, there is a very nice trick that makes it possible: one can add a set of points that acts as a ``catalyst'', meaning that the points remain unchanged through the transitions and their weight can be made arbitrarily small, yet, they enable transitions that were not possible without them.  
To make this statement more precise, denote by $h^-, h^+, v^-$, and $v^+$ the negative and positive parts of the functions $h$ and $v$. One can show that there exists a point game that, for any $\gamma>0$, achieves
$\frac{1}{2}[1,0] + \frac{1}{2}[0,1] + \gamma v^- \rightarrow [\beta,\alpha] + \gamma v^-$.
Last, it is possible to remove completely these low-weight catalyst points in the expense of moving the final point to $[\beta+\eps,\alpha+\eps]$.


Hence, in the end, proving the existence of WCF protocol with arbitrarily small bias consists of designing two valid functions $h$ and $v$,
 which define a TIPG with final point $[1/2+\eps, 1/2+\eps]$. 
This is the last model that is involved in this proof. See \autoref{fig:models} for a summary of the models. 

{~}

\textit{Step 5: Construction.}\quad 
In this last part, we describe a family of point games with parameter $k$, whose final point is $[\alpha,\alpha]=[\frac{1}{2}+\frac{c}{k},\frac{1}{2}+\frac{c}{k}]$ for some constant $c$; $k$ can be made to be an arbitrarily large integer, making the bias arbitrarily small. 

The game consists of three main steps (see \autoref{fig:overview}). First, the initial points are split into a large number of points along the axes. We assume that all points are on a grid of step $\omega$ (which will be inverse polynomial in $k$), and that $\Gamma \omega$ is the largest coordinate of a point ($\Gamma$ will be a polynomial in $k$). Second, all the axes points are moved to two points $[\alpha-k\omega,\alpha]$ and $[\alpha, \alpha-k\omega]$, where $k\omega$ will be inverse polynomial in $k$. This is the main part of the construction. It involves a large number of other points on a grid, and all those points together form a {\em ladder} of height $\Gamma$: for each different height, there is one point on the axis and $2k$ points symmetrically put close to the diagonal (see \autoref{fig:ladder}). Third, the last two points are raised to the final point $[\alpha,\alpha]$.  

The main difficulty is to find weights for the points on the axes and the points in the ladder, such that the initial splits are valid and in addition the transitions that involve the points of the ladder in the second step, are also valid. We consider the second part of the point game as a TIPG and hence, we will need to prove the validity of only two functions. 
We show that for any $k$, we can find $\omega$ and $\Gamma$ as functions of $k$ such that the point game is valid and the final point is $\alpha = \frac{1}{2} + \frac{c}{k}$ for some constant $c$.

\COMMENT{
A key technical lemma \dnote{can we ref it to where it is inside the paper} shows that, in high level \dnote{is high related to large k?}, 
for any set of $(2k+1)$ points and any polynomial of degree  at most $(2k-1)$, there exists a valid function involving these points \dnote{this is completely vague, either we say it 
in a more comperhensible way or not at all.. how is the conclusion related to the assumption? the conclusion says nothing about a polynomial. It should be "such that something about the polynomial is satisfied..."}. Using this lemma for the $2k+1$ points at each different height, we find weights for the points of the ladder such that the ladder is valid \dnote{again, this is very vague. What does it mean for the ladder to be valid? }. This has the effect of also fixing the weight of the points of the axes. The proof is concluded by showing  that the initial split is also valid for these weights, as long as $\alpha > \frac{1}{2} + \frac{c}{k}$ for some constant $c$ and appropriate small value of $\omega$ and large value of $\Gamma$ as a function of $k$.
}


Finally, we analyse the resources of the protocols and prove that the number of qubits used in a protocol with bias $\eps$ is only $O(\log \frac{1}{\eps})$, while the number of rounds is $\left(\frac{1}{\eps}\right)^{O(\frac{1}{\eps})}$ (\autoref{thm:resources}). 

We still do not have much intuition about this point game. The main problem being that we have no intuitive understanding of the large number of transitions that involve the points in the ladder during the second step of the protocol. It remains an open question to find a simpler construction that possibly also uses a smaller number of rounds.

\begin{figure}
\definecolor{theyellow}{rgb}{1,1,0.55}
\begin{tikzpicture}[model/.style={rectangle, rounded corners,  fill=green!50!black!20, very thick, minimum width=2.8cm, minimum height= 3cm},
proto/.style={rectangle, rounded corners,  fill=blue!70!black!20, very thick, minimum height= 1cm},
dual/.style={rectangle, rounded corners,  fill=theyellow, very thick, minimum height= 1cm},
linet/.style={very thick,draw=green!30!black!70, shorten >=2pt, shorten <=2pt},
xscale=1.09
]

\node[model,text width=2cm,text centered] at (0,0) (ubp) {};
\node[proto, text width=2cm,text centered] at(0,-.67) (protocol) {\small Protocol};
\node[dual,  text width=2cm,text centered] at (0,.67) (dual) {\small Dual feasible points};
\node[model,text width=2cm,text centered] at (4,0) (ebm) {\small Point game with EBM transitions};
\node[model,text width=2cm,text centered] at (8,0) (valid) {\small Point game with valid transitions};
\node[model,text width=2cm,text centered] at (12,0) (tipg) {\small Time independent point game};

\draw[->,linet] (tipg) -- (valid) node[midway, above ]{\hyperref[prop:TipgToTdpg]{Th.~\ref*{prop:TipgToTdpg}}};
\draw[->,linet] (valid) -- (ebm) node[midway,above] {\hyperref[thm:ValidReverse]{Th.~\ref*{thm:ValidReverse}}};
\draw[->,linet] (ebm) -- (ubp) node[midway, above] {\hyperref[thm:EBMReverse]{Th.~\ref*{thm:EBMReverse}}} ;
\end{tikzpicture}
\caption{The succession of models we will consider. An arrow from model A to model B means that proving the existence of an $\eps$ biased protocol in A implies the existence of an $\eps+\eps'$ biased protocol in B (for all $\eps'>0$).}
\label{fig:models}
\end{figure}
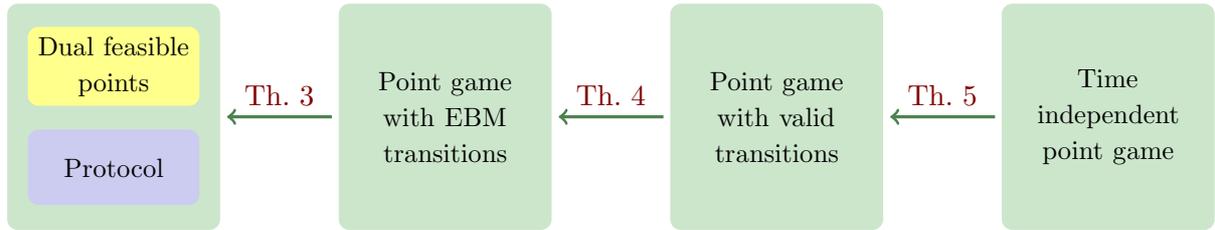

\paragraph{Our contributions} 

Our initial goal was 
to verify the correctness of Mochon's proof, 
and make it more easily verifiable to others, 
since the paper had never been formally peer-reviewed. 
During the process, we managed to simplify and shorten the construction. 
Our contribution is three-fold. 
\begin{enumerate}
\item Our main technical contribution is to give a 
significantly simpler proof to go from valid transitions to EBM 
transitions (\autoref{thm:ValidReverse}). This theorem makes 
the use of operator monotone functions in the paper clearer, 
since now they appear as the dual of the set of EBM functions, whereas in Mochon's paper they were not clearly motivated. 
Then, the definition of valid functions comes naturally as the bidual 
of EBM functions. We conclude the proof of the theorem by using
a topological argument. This replaces about 20 pages of difficult analysis 
using integrals and perturbation theory in Mochon's original manuscript \cite[Appendix C]{Moc07}.

\item On the conceptual side, we have reorganized Mochon's paper 
throughout steps $2$ to $4$. We emphasized the key steps of the proof 
and clarified some arguments which were only implicit in Mochon's paper; 
the proof is structurally much simpler and cleaner.  

\item The point game with arbitrarily small bias that we present in \autoref{sec:construction} is the same as in Mochon's paper. However, in his proof he looked at the limit of $\omega\to 0$ and $\Gamma \to \infty$. While this is enough for a proof of existence, it does not give any bound on the resources needed for the protocol. Here, we make the analysis more precise and present for the first time a bound on the resources necessary to achieve bias $\eps$. The number of qubits is $O(\log \frac{1}{\eps})$ and the number of rounds is at most $\left(\frac{1}{\eps}\right)^{O(\frac{1}{\eps})}$ (\autoref{thm:resources}).  
\end{enumerate}
Hopefully
this work will make Mochon's proof more understandable to 
the community.
Our work is but the beginning rather than the end, on Mochon's important 
result. 
In particular, it remains to further understand how Kitaev's formalism 
can be applied in a way which is both more intuitive and more general, to this as well as 
to other tasks. This also includes an improvement of 
the construction of a TIPG (step 5), which
is still very heavy and should be simplified.
Moreover, much of the proof is still non-explicit, and we believe this 
should also be attributed to our lack of understanding of it. 
It remains to find an explicit  description of a
WCF protocol that achieves bias $\eps$ for any $\eps > 0$. 
This might have of course practical implications; moreover, 
understanding the quantum mechanisms behind this WCF protocol 
could have applications to other quantum communication protocols, 
and in particular, to protocols that involve distribution of entanglement in an adversarial environment.  \\

\noindent
The organization of the paper follows the five steps of the proof described above. 

\section{SDPs and dual feasible points}
\label{sec:wcf}


\subsection{Definitions}

We formally define a quantum weak coin flipping protocol with bias $\eps$.

\begin{definition}[Weak coin flipping protocol (WCF) with bias $\eps$]
\label{def:balancedwcf}
	For $n$ even, an $n$-message weak coin flipping protocol between two players, Alice and Bob, is described by:
\begin{itemize}
	\item Three Hilbert spaces $\AA,\BB$ corresponding to Alice and Bob private workspaces (Bob does not have any access to $\AA$ and Alice to $\BB$), and a message space $\MM$;
	\item An initial product state $\ket{\psi_0} = \ket{\psi_{A,0}} \otimes \ket{\psi_{M,0}} \otimes \ket{\psi_{B,0}} \in \AA \otimes \MM \otimes \BB$;
	\item A set of $n$ unitaries $\{U_1,\dots,U_n\}$ acting on $\AA\otimes\MM\otimes\BB$, with $U_i = U_{A,i}\otimes\I_\BB$ for $i$ odd, and $U_i = \I_\AA \otimes U_{B,i}$ for $i$ even;
	\item A set of \dt{honest states} $\{\ket{\psi_i},\ i\in [n] \}$ defined by $\ket{\psi_i} = U_iU_{i-1}\cdots U_1 \ket{\psi_0}$;
	\item A set of $n$ projectors $\{E_1,\dots,E_n\}$ acting on $\AA\otimes\MM\otimes\BB$, with $E_i = E_{A,i}\otimes\I_\BB$ for $i$ odd, and $E_i = \I_\AA \otimes E_{B,i}$ for $i$ even, such that $E_i\ket{\psi_i} = \ket{\psi_i}$;
\item Two final POVM $\left\{\Pi_{A}^{(0)},\Pi_{A}^{(1)}\right\}$ acting on $\AA$ and  $\left\{\Pi_{B}^{(0)},\Pi_{B}^{(1)}\right\}$ acting on $\BB$.
\end{itemize}
The WCF protocol proceeds as follows:
\begin{enumerate}
	\item In the beginning, Alice holds $\ket{\psi_{A,0}}\ket{\psi_{M,0}}$ and Bob $\ket{\psi_{B,0}}$.
	\item For $i=1$ to $n$: \\
		\ --- If $i$ is odd, Alice applies $U_i$ and measures the resulting state with the POVM $\{E_i, \I - E_i\}$. On the first outcome, Alice sends the message qubits to Bob; on the second outcome, she ends the protocol by outputting ``0'', i.e. Alice declares herself winner. \\
		\ --- If $i$ is even, Bob applies $U_i$ and measures the resulting state with the POVM $\{E_i, \I - E_i\}$. On the first outcome, Bob sends the message qubits to Alice; on the second outcome, he ends the protocol by outputting ``1'', i.e. Bob declares himself winner.
	\item Alice and Bob measure their part of the state with the final POVM and output the outcome of their measurements. Alice wins on outcome ``0'' and Bob on outcome ``1''.
\end{enumerate}

The WCF protocol has the following properties:
\begin{itemize}
	\item \textbf{Correctness}: When both players are honest, Alice and Bob's outcomes are always the same:  $\Pi_{A}^{(0)}\otimes\I_\MM\otimes\Pi_{B}^{(1)}\ket{\psi_n} = \Pi_{A}^{(1)}\otimes\I_\MM\otimes\Pi_{B}^{(0)}\ket{\psi_n} = 0$.
	\item \textbf{Balanced:} When both players are honest, they win with probability 1/2: \\$P_A = \norm{\Pi_{A}^{(0)}\otimes\I_\MM\otimes\Pi_{B}^{(0)}\ket{\psi_n}}^2 = \frac{1}{2}$ and $P_B = \norm{\Pi_{A}^{(1)}\otimes\I_\MM\otimes\Pi_{B}^{(1)}\ket{\psi_n}}^2 = \frac{1}{2}$.
	\item \textbf{$\bm\eps$ biased:} When Alice is honest, the probability that both players agree on Bob winning is $P_B^* \leq 1/2 + \eps$. And conversely, if Bob is honest, the probability that both players agree on Alice winning is $P_A^* \leq 1/2 + \eps$.
\end{itemize}
\end{definition}

This definition of a weak coin flipping protocol differs from the usual one in the sense that we added the projections $\{E_i\}$. The goal of these projections is to catch a cheating player, since they do not change the honest states. Intuitively they can only decrease the bias compared to a protocol without them. This can be proved, but it is not necessary in our case since we will directly prove upper bounds on the cheating probabilities for this specific type of protocols.

\subsection{Cheating probabilities as SDPs}
The cheating probabilities $P_A^*$ and $P_B^*$ cannot be easily computed from the definition above. Kitaev showed that they can be expressed as semidefinite programs (SDP) \cite{Kit03} and a written proof can be found in \cite{ABDR04}. 

Fix a weak coin flipping protocol, and assume that Alice is honest. We describe a semidefinite program with variables the states $\rho_{AM,i}$, i.e. the states after round $i$ once Bob's workspace is traced out. The probability that Bob wins is the probability that Alice outputs ``1'' when applying the POVM $\left\{\Pi_{A}^{(0)},\Pi_{A}^{(1)}\right\}$ to her part of the final state, or equivalently $ \Tr((\Pi_{A}^{(1)}\otimes\I_\MM)\rho_{AM,n})$.  Since Alice is honest, the state in her workspace is not arbitrary, but rather satisfies some constraints. In the beginning of the protocol, Alice held the state $\tr_\MM(\rho_{AM,0}) = \proj{\psi_{A,0}}$. Moreover, the evolution of Alice's state is only due to her own actions, namely $\tr_\MM(\rho_{AM,i}) = \tr_\MM(\rho_{AM,i-1})$  if $i$ is even and 
$\tr_\MM(\rho_{AM,i}) = \tr_\MM(E_iU_i \rho_{AM,i-1} U_i^\dagger E_i)$ if $i$ is odd. Bob's cheating probability is the maximum over all his strategies, i.e. over all states $\{ \rho_{AM,i}\}$ that satisfy these constraints.

The evolution of the states $\rho_{AM,i}$ is not unitary due to the presence of the projections, so they are not necessarily normalized. However, $\tr((\Pi_{A}^{(1)}\otimes\I)\rho_{AM,n})$ represents the probability that Alice and Bob agree on Bob winning when Alice is honest. If Bob got caught cheating by the projections, Alice already declared herself the winner. The non-normalization of the states $\rho_{AM,i}$ reflects the probability that the protocol ended prematurely by one of the players declaring oneself winner, because the other player was caught cheating.

This reasoning leads to the following two semidefinite programs:
\begin{thm}[Primal] \label{thm:primal} \ \newline
$P_B^* = \max \Tr((\Pi_{A}^{(1)}\otimes\I_\MM)\rho_{AM,n})$ over all $\rho_{AM,i}$ satisfying the constraints:
\begin{itemize}
	\item $\tr_\MM(\rho_{AM,0}) = \tr_{\MM\BB}(\proj{\psi_0}) = \proj{\psi_{A,0}}$;
	\item for $i$ odd, $\tr_\MM(\rho_{AM,i}) = \tr_\MM(E_iU_i \rho_{AM,i-1} U_i^\dagger E_i)$;
   \item for $i$ even, $\Tr_\MM(\rho_{AM,i}) = \Tr_\MM(\rho_{AM,i-1})$.
\end{itemize}

\noindent $P_A^* = \max \Tr((\I_\MM\otimes \Pi_{B}^{(0)}) \rho_{MB,n})$ over all $\rho_{BM,i}$ satisfying the constraints:
\begin{itemize}
	\item $\tr_\MM(\rho_{MB,0}) = \tr_{\AA\MM}(\proj{\psi_0}) = \proj{\psi_{B,0}}$;
	\item for $i$ even, $\tr_\MM(\rho_{MB,i}) = \tr_\MM(E_iU_i \rho_{MB,i-1} U_i^\dagger E_i)$;
   \item for $i$ odd, $\tr_\MM(\rho_{MB,i}) = \Tr_\MM(\rho_{MB,i-1})$.
\end{itemize}
\end{thm}

\subsection{Upper bounds on the cheating probabilities via the dual feasible points}
In order to prove upper bounds on the cheating probabilities $P_A^*$ and $P_B^*$ we consider the dual SDPs. The following theorem provides the dual, as well as 
a statement that the maximum is indeed achieved and is equal to the 
optimal value of the primal, namely, strong duality holds. A complete proof of this theorem can be found in \cite{Kit03,ABDR04}. 
\begin{thm}[Dual]\label{thm:dual}\ \newline
	$P_B^* = \min \Tr(Z_{A,0}\proj{\psi_{A,0}})$  over all $Z_{A,i}$ under the constraints:
	\begin{enumerate}
		\item[\ding{172}] $\forall i, \ Z_{A,i} \succeq 0$;
		\item[\ding{173}] for $i$ odd, $Z_{A,i-1}\otimes\I_\MM \succeq U_{A,i}^\dagger E_{A,i}(Z_{A,i}\otimes\I_\MM) E_{A,i}U_{A,i}$;
      		\item[\ding{174}] for $i$ even, $Z_{A,i-1} = Z_{A,i}$;
		\item[\ding{175}] $Z_{A,n} = \Pi_{A}^{(1)}$.
	\end{enumerate}
\noindent $P_A^* = \min \Tr(Z_{B,0}\proj{\psi_{B,0}})$ over all $Z_{B,i}$ under the constraints:
	\begin{enumerate}
		\item[\ding{172}] $\forall i, \ Z_{B,i} \succeq 0$;
		\item[\ding{173}] for $i$ even, $\I_\MM \otimes Z_{B,i-1} \succeq U_{B,i}^\dagger E_{B,i}(\I_\MM \otimes Z_{B,i}) E_{B,i}U_{B,i}$;
      		\item[\ding{174}] for $i$ odd, $Z_{B,i-1} = Z_{B,i}$;
		\item[\ding{175}] $Z_{B,n} = \Pi_{B}^{(0)}$.
	\end{enumerate}
\end{thm}
\noindent
We add one more constraint to the above dual SDPs:
	\begin{enumerate}
\item[\ding{176}] $\ket{\psi_{A,0}}$ is an eigenvector of $Z_{A,0}$, i.e. there exists $\beta>0$ such that $Z_{A,0}\ket{\psi_{A,0}}= \beta \ket{\psi_{A,0}}$, 
	\item[\ding{176}] $\ket{\psi_{B,0}}$ is an eigenvector of $Z_{B,0}$, i.e. there exists $\alpha>0$ such that $Z_{B,0}\ket{\psi_{B,0}}= \alpha \ket{\psi_{B,0}}$.
\end{enumerate}
The reason why we are adding this constraint will become clear a bit later. Notice that this constraint is not positive semidefinite, and thus the following definition is a slight abuse:
\begin{definition}[Dual feasible points]\label{dfp}
We call \emph{dual feasible points} any two sets of matrices $\{Z_{A,0},\dots,Z_{A,n}\}$ and $\{Z_{B,0},\dots,Z_{B,n}\}$ that satisfy the corresponding conditions \ding{172} to \ding{176}.
\end{definition}
However, this additional constraint does not change the value of the dual SDPs:
\begin{proposition}\label{prop:AlphaAndBeta}
	$P^*_A = \inf \alpha$ and $P_B^* = \inf \beta$ where the infimum is over all dual feasible points and $\alpha,\beta$ are defined in constraint \ding{176} of the definition of the dual feasible points (\autoref{dfp}).
\end{proposition}
\begin{proof}
Fix $\{Z_{A,0},\dots,Z_{A,n}\}$ a set of matrices that satisfies the constraints \ding{172} to \ding{175} and $\eps>0$. Let us construct a matrix $Z'_{A,0}$ such that the set $\{Z'_{A,0},Z_{A,1},\dots,Z_{A,n} \}$ satisfies the constraints \ding{172} to \ding{176} with $\alpha = \tr(Z'_{A,0}\proj{\psi_{A,0}}) = \tr(Z_{A,0}\proj{\psi_{A,0}}) + \eps$. 

The proof relies on the following fact: there exists $\Lambda>0$ such that:
\begin{align*}
	Z'_{A,0} = (\bra{\psi_{A,0}}Z_{A,0}\ket{\psi_{A,0}} + \eps)\proj{\psi_{A,0}} + \Lambda(\I - \proj{\psi_{A,0}}) \succeq Z_{A,0}.
\end{align*}
As a consequence $Z'_{A,0}\otimes \I_\MM \succeq Z_{A,0} \otimes \I_\MM$, hence $\{Z'_{A,0},Z_{A,1},\dots,Z_{A,n} \}$ is a dual feasible point and also $\tr(Z'_{A,0}\proj{\psi_{A,0}}) = \tr(Z_{A,0}\proj{\psi_{A,0}}) + \eps$. 
%

We now prove that, indeed, there exists $\Lambda>0$, for which $Z'_{A,0} \succeq Z_{A,0}$. Let $\ket{\phi}$ be a vector in $\AA$ and decompose it as $\ket{\phi}  = a\ket{\psi_{A,0}} + b \ket{\psi_{A,0}^\bot}$ where $\braket{\psi_{A,0}}{\psi_{A,0}^\bot} = 0$. We can restrict ourselves to $b\in\R$ and $\abs{a}^2 + \abs{b}^2 = 1$, thus we have:
\begin{align}
\label{eqn:za0}
\bra{\phi}(Z_{A,0}'-Z_{A,0})\ket{\phi} = \abs{a}^2\eps+b^2(\Lambda- \bra{\psi_{A,0}^\bot}Z_{A,0}\ket{\psi_{A,0}^\bot}) -2b \Re(a\bra{\psi_{A,0}^\bot}Z_{A,0}\ket{\psi_{A,0}}).
\end{align}
We show how to pick a $\Lambda$ independent of $\ket{\phi}$, i.e. of $a,b$ and $\ket{\psi_{A,0}^\bot}$, so that the above expression is always non-negative:
\begin{description}
\item[Case $\bm{a=0}$] We want $\Lambda \geq \bra{\psi_{A,0}^\bot}Z_{A,0}\ket{\psi_{A,0}^\bot}$ for all $\ket{\psi_{A,0}^\bot}$. Hence, we need to choose $\Lambda \geq \norm{Z_{A,0}}$. 

\item[Case $\bm{a\neq0}$] See \autoref{eqn:za0} as a polynomial in $b$. The leading coefficient being non negative, we need to pick $\Lambda$ so that the discriminant is negative. The discriminant reads $4\Re(a\bra{\psi_{A,0}^\bot}Z_{A,0}\ket{\psi_{A,0}})^2 - 4\abs{a}^2\eps(\Lambda - \bra{\psi_{A,0}^\bot}Z_{A,0}\ket{\psi_{A,0}^\bot})$. Since for any complex number $x$, $\Re(x) \leq \abs{x}$, it is sufficient to have, $\abs{a}^2\abs{\bra{\psi_{A,0}^\bot}Z_{A,0}\ket{\psi_{A,0}}}^2 - \abs{a}^2 \eps (\Lambda - \bra{\psi_{A,0}^\bot}Z_{A,0}\ket{\psi_{A,0}^\bot}) \leq 0$. Since $a\neq0$, we want $\Lambda \geq \frac{1}{\eps} (\abs{\bra{\psi_{A,0}^\bot}Z_{A,0}\ket{\psi_{A,0}}}^2 - \bra{\psi_{A,0}^\bot}Z_{A,0}\ket{\psi_{A,0}^\bot})$. Hence, we need to choose $\Lambda \geq \norm{Z_{A,0}}^2 / \eps$.
\end{description}

Choosing $\Lambda \geq \max \{ \norm{Z_{A,0}}, \norm{Z_{A,0}}^2 / \eps \}$ concludes the proof.

\end{proof}





Let us note that the primal formulation sets constraints on the evolution of $\rho_{i-1}$ to $\rho_{i}$, starting from a fixed state $\rho_0$, and the optimization quantity that depends on $\rho_n$. In the dual formulation, the constraints are on the evolution of $Z_{i}$ to $Z_{i-1}$, starting form a fixed $Z_n$ and the optimisation quantity depends on $Z_0$. For this reason we will reverse the time evolution and consider point games that run backwards in time.


\section{Point games with EBM transitions}
\label{sec:EBM}
\label{sec:essence}

In the previous section we saw how to upper-bound the cheating probability of any weak coin flipping protocol by looking at the dual SDP formulation of the protocol and by providing dual feasible points, namely  sets of matrices $\{Z_{A,0},\dots,Z_{A,n}\}$ and  $\{Z_{B,0},\dots,Z_{B,n}\}$ that satisfy a number of conditions. One can think of these matrices as a security witness of the protocol. 

As we have explained, in order to keep track of the two cheating probabilities together, we will be interested in the quantity $\triple{\psi_i}{Z_{A,i} \otimes \I_\MM \otimes Z_{B,i}}{\psi_i}.$ 
For $i=0$, this quantity indeed provides an upper bound on the product of the two cheating probabilities (this only holds for $i=0$, since the state $\ket{\psi_0}$ is separable). Note also that this is the same quantity used by Kitaev in the proof of the lower bound on quantum strong coin flipping~\cite{Kit03}. 

This motivates the definition of EBM point games as graphical representations of the above quantity. This section is devoted to first formally defining EBM point games, and then showing that they are equivalent to WCF protocols and their dual feasible points. We first define the following function $\emph{prob}$.

\begin{definition}[prob]
Let $Z$ be a positive semidefinite matrix and denote by $\Pi^{[z]}$ the projector on the eigenspace of eigenvalue $z\in\sp(Z)$. We have $Z = \sum_z z \Pi^{[z]}$. Let $\ket{\psi}$ be a (not necessarily unit) vector. We define the function with finite support $\prob[Z,\psi]:[0,\infty) \to [0,\infty)$ as:
\begin{align*}
	\prob[Z,\psi](z)  &= \left\{ \begin{array}{cl}
		\bra{\psi} \Pi^{[z]} \ket{\psi} & \text{if}\ z\in\sp(Z) \\
		0 & \text{otherwise.}
	\end{array} \right.
\end{align*}
If $Z= Z_A \otimes \I_\MM \otimes Z_B$, using the same notation, we define the 2-variate function with finite support  $\prob[Z_A,Z_B,\psi]:[0,\infty) \times [0,\infty) \to [0,\infty)$ as:
\begin{align*}
	\prob[Z_A,Z_B,\psi](z_A,z_B)  &= \left\{ \begin{array}{cl}
		\bra{\psi} \Pi^{[z_A]} \otimes \I_\MM \otimes \Pi^{[z_B]} \ket{\psi} & \text{if}\ (z_A,z_B) \in\left(\mathrm{sp}(Z_A),\mathrm{sp}(Z_B)\right) \\
		0 & \text{otherwise.}
	\end{array} \right.
\end{align*}
\end{definition}

We now define EBM transitions and EBM point games.

\begin{definition}[EBM line transition]\label{def:EBM_Line_Transition}
	\label{def:EBMtransition}
	Let $l,r:[0,\infty) \to [0,\infty)$ be two functions with finite supports. The line transition $l\to r$ is \dt{expressible by matrices} (EBM) if there exist two positive semidefinite matrices $0 \preceq X \preceq Y$ and a (not necessarily unit) vector $\ket\psi$ such that $l = \prob[X,\psi] \quad\textrm{and}\quad r=\prob[Y,\psi]$.\end{definition}
	
\begin{definition}[EBM transition]\label{def:EBM_Transition}
	Let $p,q:[0,\infty)\times[0,\infty)\to[0,\infty)$ be two functions 
	with finite supports. The transition $p \to q$ is an \dt{EBM horizontal transition} if for all $y\in[0,\infty)$, $p(\cdot,y) \to q(\cdot,y)$ is an EBM line transition, and an \dt{EBM vertical transition} if for all $x\in[0,\infty)$, $p(x,\cdot) \to q(x,\cdot)$ is an EBM line transition.
\end{definition}

\begin{definition}[EBM point game]\label{def:EBM_Point_Game}
An \emph{EBM point game} is a sequence of functions $\{ p_0,  p_1 , \cdots , p_n \}$ with finite support such that:
\begin{itemize}
	\item $p_0 = 1/2 [0,1] + 1/2 [1,0]$;
	\item For all even $i$, $p_{i} \to p_{i+1}$ is an EBM vertical transition;
	\item For all odd $i$, $p_{i} \to p_{i+1}$ is an EBM horizontal transition;
	\item $p_n = 1 [\beta,\alpha]$.
\end{itemize}
\end{definition}

For completeness, we first show how to go from a given WCF protocol accompanied with two dual feasible points to an EBM point game.  
\begin{proposition}\label{Prop:Forward}
Given a WCF protocol with cheating probabilities $P^*_A$ and $P^*_B$, then, 
for any $\delta > 0$, there exists an EBM point game with final point $[P^*_B+\delta,P^*_A+\delta]$.
\end{proposition}
\begin{proof}
From \autoref{prop:AlphaAndBeta}, we have that $P_A^* = \inf \alpha$ and $P_B^* = \inf \beta$, where the infimum is taken over all dual feasible points, $\alpha$ is the eigenvalue of $Z_{A,0}$ that corresponds to the eigenvector $\ket{\psi_{A,0}}$ and similarly for $\beta$. Hence, for any $\delta >  0$ there exist two dual feasible points $\left\{Z_{A,i}\right\}$ and $\left\{Z_{B,i}\right\}$ with $\alpha = P_A^*+\delta$  and $\beta = P_B^*+\delta$.

Let 
\begin{align}
	p_{n-i} &= \prob[Z_{A,i},Z_{B,i},\psi_i] \nonumber\\
	&= \sum_{(z_{A,i},z_{B,i}) \in\left(\sp(Z_{A,i}),\sp(Z_{B,i})\right)} \bra{\psi_i} \Pi^{[z_{A,i}]} \otimes \I_\MM \otimes \Pi^{[z_{B,i}]} \ket{\psi_i} \ [z_{A,i},z_{B,i}]. \label{eqn:pnMinusi}
\end{align}
We show that the sequence of functions $\{p_0,\dots,p_n\}$ is indeed an EBM point game with $p_n = [P^*_B+\delta,P^*_A+\delta]$.
We first get the initial and final condition:
\begin{align*}
p_0 & = \prob[Z_{A,n},Z_{B,n},\psi_n] = \prob[\Pi_{A}^{(1)},\Pi_{B}^{(0)},\psi_n] = \ahalf [1,0] + \ahalf [0,1],\\
p_n & = \prob[Z_{A,0},Z_{B,0},\psi_0] = 1 [\beta,\alpha] = [P_B^*+\delta, P_A^* + \delta]. 
\end{align*}
Recall that we added an extra condition in \autoref{thm:dual}, namely that $\ket{\psi_{A,0}}$ and $\ket{\psi_{B,0}}$ are eigenstates of $Z_{A,0}$ and $Z_{B,0}$ respectively. This condition ensures us that the game ends with one final point, and not several points.

We now show that $p_i \rightarrow p_{i+1}$ are EBM transitions.
Let us assume that $i$ is odd. According to \autoref{eqn:pnMinusi}, the function $p_{n-i}$ (resp. $p_{n-i+1}$) corresponds to the matrix $Z_{i}$ (resp. $Z_{i-1}$) and the state $\ket{\psi_{i}}$ (resp. $\ket{\psi_{i-1}}$). Since $i$ is odd, the conditions of the dual SDP state that $Z_{B,i} = Z_{B,i-1}$ and also $\ket{\psi_{i}} = E_{A,i}U_{A,i}  \otimes \I_\BB \ket{\psi_{i-1}} $.
Using this we will now prove that the points only move horizontally and moreover, the total weight on every horizontal line remains unchanged (thus so does the total weight).
To see this, write: 

\begin{align*}
p_{n-i+1} = \sum_{(z_{A,i-1},z_{B,i-1})} \bra{\psi_{i-1}} \Pi^{[z_{A,i-1}]} \otimes \I_\MM \otimes \Pi^{[z_{B,i-1}]} \ket{\psi_{i-1}} \ [z_{A,i-1},z_{B,i-1}].
\end{align*}
\begin{align*}
p_{n-i} & =  \sum_{(z_{A,i},z_{B,i-1})}  \bra{\psi_{i}} \Pi^{[z_{A,i}]} \otimes \I_\MM \otimes \Pi^{[z_{B,i-1}]} \ket{\psi_i} \ [z_{A,i},z_{B,i-1}] ,\\
& =   \sum_{(z_{A,i},z_{B,i-1})} \bra{\psi_{i-1}} U_{A,i}^\dagger E_{A,i} (\Pi^{[z_{A,i}]} \otimes \I_\MM ) E_{A,i}U_{A,i} \otimes \Pi^{[z_{B,i-1}]} \ket{\psi_{i-1}} \ [z_{A,i},z_{B,i-1}].
\end{align*}

First, notice that $\sp(Z_{B,i}) = \sp(Z_{B,i-1})$ and hence the possible values for the second coordinate of the points remain the same. Second, the sum of the weights of the points in each horizontal line with second coordinate $z_{B,i-1}$ remains the same and equal to $\bra{\psi_{i-1}}  \I_\AA \otimes \I_\MM \otimes \Pi^{[z_{B,i-1}]}  \ket{\psi_{i-1}}$. Note that the projections $E_i$ leave the honest states unchanged. Note also that for every $z_{B,i-1}$, i.e. for every horizontal line, we can define the functions 
\begin{align*}
 p_{n-i+1}(\cdot,z_{B,i-1}) & =  \prob[  Z_{A,i-1} \otimes \I_\MM \otimes \Pi^{[z_{B,i-1}]} ,{\psi_{i-1}}],  \\
p_{n-i}(\cdot,z_{B,i-1}) & =  \prob[
U_{A,i}^\dagger E_{A,i} (Z_{A,i} \otimes \I_\MM ) E_{A,i}U_{A,i} \otimes \Pi^{[z_{B,i-1}]},{\psi_{i-1}}],
\end{align*}
and from the dual SDP condition \ding{173} in \autoref{thm:dual}, we have
$ Z_{A,i-1} \otimes \I_\MM \otimes \Pi^{[z_{B,i-1}]} \succeq U_{A,i}^\dagger E_{A,i} (Z_{A,i} \otimes \I_\MM ) E_{A,i}U_{A,i} \otimes \Pi^{[z_{B,i-1}]}$. From there, we conclude that $p_{n-i} \rightarrow p_{n-i+1}$ is a horizontal EBM transition.
Similarly, for $i$ even, the points move only vertically and the total weight on every vertical line remains unchanged. 
\end{proof}

Quite surprisingly, the reverse implication is also true.

\begin{thm}[EBM to protocol]
	\label{thm:EBMReverse}
	Given an EBM point game with final point $[\beta,\alpha]$, there exists a WCF protocol and two dual feasible points proving that the optimal cheating probabilities are $P_A^* \leq \alpha$ and $P_B^*\leq\beta$.
\end{thm}

The proof of this fact is somewhat lengthy, though rather simple. In \autoref{thm:SpectraXY} we show that for every EBM line transition, there exist matrices $X$ and $Y$ and a vector $\ket{\psi}$ that certify that the transition is EBM and they have some extra nice properties. This is the non-constructive part of the proof, since we do not know how to find these matrices in an efficient way. Given these matrices and vectors it is not hard to construct the WCF protocol, namely define the honest states, the unitaries, and the projections. Alice and Bob will share a superposition of all points that appear in the EBM game and they will take turns changing the amplitudes according to the distributions over these points specified by the EBM transitions.

\begin{proof}
Consider an EBM point game with transitions $p_0 \to p_1 \to \cdots \to p_n = [\beta,\alpha]$ and let us define the sets of all possible first and second coordinates $z_A$ and $z_B$ of all the points that appear in the game:
\begin{align*}
	S_A &= \left\{ z_A \geq0\ |\ \exists i\in\{0,\ldots,n\},\ \exists z_B \geq0,\ p_i(z_A,z_B) >0 \right\}, \\
	S_B &= \left\{ z_B \geq 0\ |\ \exists i\in \{0,\ldots,n\},\ \exists z_A \geq0,\ p_i(z_A,z_B) >0 \right\}.
\end{align*}
We wish to find a protocol (honest states, unitaries, projections) and dual feasible points that guarantee that in this protocol Alice's and Bob's cheating probabilities are upper-bounded by $\alpha$ and $\beta$ respectively. 
The idea is the following: every point $[z_A,z_B]$ of the game will be represented as an orthogonal state $\ket{0,z_A}\ket{z_A,z_B}\ket{z_B,0} \in \AA \otimes \MM \otimes \BB$, where
\begin{align*}
	\AA  &= \mathrm{span} \{ \ket{b,z_A},\ b\in\{0,1\},\ z_A\in S_A \}, \\
	\MM &= \AA' \otimes \BB' =  \mathrm{span} \{ \ket{z_A,z_B},\  z_A\in S_A,\ z_B \in S_B \}, \\
	\BB  &= \mathrm{span} \{ \ket{z_B,b},\ b\in\{0,1\},\ z_B\in S_B \}.
\end{align*}

The honest state $\ket{\psi_i}$ of the protocol at round $i$ will be
\begin{align*}
	\ket{\psi_i} = \sum_{z_A\in S_A,z_B\in S_B} \sqrt{p_{n-i}(z_A,z_B)} \ket{0,z_A}\ket{z_A,z_B}\ket{z_B,0}.
\end{align*}

The message space $\MM$ contains the states that correspond to all the points of the game so that both players have alternate access to them and can manipulate their amplitudes by applying a unitary operation. The role of the unitary $U_i$ is to transform the state $\ket{\psi_{i-1}}$ to the state $\ket{\psi_i}$, in other words $U_i \ket{\psi_{i-1}} = \ket{\psi_i}$. Moreover, we need to ensure that when Alice (resp. Bob) applies a unitary, then it corresponds to a horizontal (resp. vertical) line transition; in other words that the sum of the squares of the amplitudes of the states with a fixed second (resp. first) coordinate remains unchanged. 
The way to achieve this, is to force Alice to perform a unitary on the message space and her workspace which only uses the second coordinate of the points as a control (similarly we need to force Bob to perform a unitary which only uses the first coordinate of the points as a control). For this reason, Alice (resp. Bob) keeps a copy of the first (resp. second) coordinate of the points and each has a qubit that becomes $1$ (via the unitary operation) when they catch the other player cheating, ie. not using the coordinate as control. We define the cheating detection projections $E_{A,i} = E_A=  \sum_{z_A} \proj{0,z_A,z_A}\otimes \I_{\BB'}$ and $E_{B,i} =  E_B = \sum_{z_B} \I_{\AA'} \otimes \proj{z_B,z_B,0}$ that allow Alice and Bob to prematurely end the protocol and declare themselves winner. Note these projections leave the honest states invariant.    

It remains to find the unitaries $U_i$ and the matrices $Z_i$. Let us assume that $i$ is odd (similarly for $i$ even), hence the transition $p_{n-i} \to p_{n-i+1}$ is horizontal; that is Alice applies the unitary $U_i=U_{A,i}\otimes\I_\BB$. Since we want this unitary to use the second coordinate only as control we have $U_{A,i} =  \sum_{z_B } U_{A,i}^{(z_B)} \otimes \proj{z_B}$.  Define the (non-normalized) states $\ket{\psi^{(z_B)}_{i-1}} = \sum_{z_A\in S_A} \sqrt{p_{n-i+1}(z_A,z_B)} \ket{0,z_A,z_A}$. Then in order to have $U_i \ket{\psi_{i-1}} = \ket{\psi_i}$, we need that
$ U_{A,i}^{(z_B)} \ket{\psi^{(z_B)}_{i-1}} = \ket{\psi^{(z_B)}_{i}}.$

We will show, in fact, that the state $\ket{\psi^{(z_B)}_{i-1}}$ can be rewritten in a different basis, so that the amplitudes are now given by the function $p_{n-i}$. Then, the unitaries $U_{A,i}^{(z_B)}$ will just perform this basis change. More precisely, we find the unitaries $U_{A,i}^{(z_B)}$ for $i=1,\ldots, n$ and a single matrix $Z_A (= Z_{A,0}=\cdots = Z_{A,n-1})$ by expressing each EBM line transition $p_{n-i}(\cdot,z_B)\to p_{n-i+1}(\cdot,z_B)$ as $\prob[X,\psi] \to \prob[Y,\psi]$, where the matrices $X,Y$ and the state $\ket{\psi}$ satisfy the properties of the following lemma:


\begin{lem}
\label{thm:SpectraXY}
	Let $l \to r$ be an EBM line transition and denote by $\supp(l)$ and $\supp(r)$ the supports of $l$ and $r$ respectively. Let $S$ be a set such that $\supp(l) \cup \supp(r) \subseteq S$ and $\Lambda > \max \{z: z\in S\}$. Given a set of orthonormal vectors $\{\ket{z}, z\in S\}$, there exists a family of $|S|$ orthonormal vectors $\{\ket{\varphi(z)},\ z\in S\}$ in the $2|S|^2$-dimensional space $\mathrm{span}\{\ket{b,z,z'},\ b\in\{0,1\},\ z,z'\in S\}$ such that
	\begin{itemize}
		\item the state $\ket{\psi} = \sum_z \sqrt{r(z)} \ket{0,z,z}$ can be expressed as $\ket{\psi} = \sum_z \sqrt{l(z)}\ket{\varphi(z)}$,
		\item $l = \prob[X,\psi]$ and $r = \prob[Y,\psi]$, with $0\preceq X \preceq Y$, and
\begin{align*}
 Y =  \sum_{z\in S} z \proj{0,z,z} + \Lambda \sum_{z\in S} \proj{1,z,z} 
 \quad\text{and}\quad
  X = \sum_{z\in S} z \proj{\varphi(z)}.
\end{align*}
\end{itemize}
\end{lem}
We defer the proof of this lemma to the end of the section and continue with the proof of the theorem.

For each $z_B \in S_B$ and each EBM line transition $p_{n-i}(\cdot,z_B) \to p_{n-i+1}(\cdot,z_B)$, we apply \autoref{thm:SpectraXY} with $S = S_A$. This defines
\begin{align*}
	X_i^{(z_B)} & = \sum_{z_A \in S_A} z_A \proj{\varphi_i^{(z_B)}(z_A)},\\
           Y& =   \sum_{z_A \in S_A} z_A \proj{0,z_A,z_A} + \Lambda \sum_{z_A \in S_A} \proj{1,z_A,z_A}, \\
 	\ket{\psi^{(z_B)}_{i-1}} & = \sum_{z_A\in S_A} \sqrt{p_{n-i+1}(z_A,z_B)} \ket{0,z_A,z_A} = \sum_{z_A} \sqrt{p_{n-i}(z_A,z_B)} \ket{ \varphi_i^{(z_B)}(z_A)}.
\end{align*}
 We can now define the unitary $U^{(z_B)}_{A,i}$ by its action on a subspace of $\AA\otimes \AA'$:
\begin{align*}
U_{A,i}^{(z_B)}:  \ket{\varphi_i^{(z_B)}(z_A)} \mapsto \ket{0,z_A,z_A}.
\end{align*}
We complete $U_{A,i}^{(z_B)}$ so that it is a unitary on $\AA\otimes\AA'$. Note that we have:
\begin{align*}
	U_{A,i}^{(z_B)}\sum_{z_A} \sqrt{p_{n-i+1}(z_A,z_B)} \ket{0,z_A,z_A} 
		& = U_{A,i}^{(z_B)}\sum_{z_A} \sqrt{p_{n-i}(z_A,z_B)} \ket{ \varphi_i^{(z_B)}(z_A)}, \\
		&= \sum_{z_A} \sqrt{p_{n-i}(z_A,z_B)} \ket{0, z_A,z_A},
\end{align*}
and thus have $U_i\ket{\psi_{i-1}} = \ket{\psi_{i}}$. Moreover, by the definition of the unitary and the cheating detection projection, we can see that indeed Bob is forced to use the first coordinate only as control. 
Note that we also have $X_i^{(z_B)} = U_{A,i}^\dagger E_{A} Y E_{A}U_{A,i} $.

We now need to define $Z_A$ acting only on the space $\AA$. We take $Z_A=\sum_{z_A \in S_A} z_A \proj{0,z_A} + \Lambda \sum_{z_A\in S_A} \proj{1,z_A}$ and note that the support of $\prob[Y,{\psi^{(z_B)}_{i-1}}]$ is equal to the support of $\prob[Z_A \otimes \I_{\AA'} , {\psi^{(z_B)}_{i-1}}]$. 

In other words, by the definition of the honest states, the projections, the unitaries and the dual feasible point, we have shown that any EBM line transition can be expressed as
\begin{align*}
p_{n-i+1}(\cdot,z_{B}) & = \prob[Y,{\psi^{(z_B)}_{i-1}}] =  \prob[  Z_A \otimes \I_{\AA'} ,{\psi^{(z_B)}_{i-1}}] = \prob[  Z_{A} \otimes \I_\MM \otimes \Pi^{[z_{B}]} ,{\psi_{i-1}}],  \\
p_{n-i}(\cdot,z_{B}) & =  \prob[X,{\psi^{(z_B)}_{i-1}}] =  \prob[U_{A,i}^{(z_B)\dagger} E_{A} (Z_{A} \otimes \I_{\AA'} ) E_{A}U_{A,i}^{z_B},{\psi^{(z_B)}_{i-1}}  ] \\
&= \prob[U_{A,i}^\dagger E_{A} (Z_{A} \otimes \I_\MM ) E_{A}U_{A,i} \otimes \Pi^{[z_{B}]},{\psi_{i-1}}].
\end{align*}
This is precisely the type of EBM line transitions that arose when we started from a protocol and a dual feasible point.

We need to verify that the $Z_A$'s we defined are a dual feasible point. According to the constraints of the dual, we pick $Z_{A,n} = \Pi_A^{(1)} = \proj{0,1}$. Since the initial points of the point game are $[0,1]$ and $[1,0]$, then 1 is an eigenvalue of $Z_A$, so we have $\Pi_A^{(1)} \preceq Z_A$, i.e. $Z_{A,n} \preceq Z_{A,n-1}$. For $i=\{ 0,\ldots, n-1\}$, we have
\begin{align*}
U^\dagger_{A,i} E_{A} (Z_A\otimes \I_\MM) E_{A} U_{A,i} 
	&= U^\dagger_{A,i} \left(\sum_{z_A}z_A  \proj{0,z_A,z_A} \otimes \sum_{z_B}\proj{z_B} \right) U_{A,i}, \\	
	&= \sum_{z_B} \sum_{z_A} z_A U^{(z_B)\dagger}_{A,i} \proj{0,z_A,z_A}U^{(z_B)} _{A,i}  \otimes \proj{z_B}, \\
	&= \sum_{z_B} X_i^{(z_B)} \otimes \proj{z_B},\\
	&\preceq Y \otimes \I_{\BB'}, \\
	&= \left( \sum_{z_A\in S_A} z_A \proj{0,z_A,z_A} + \Lambda \proj{1,z_A,z_A} \right) \otimes \I_{\BB'}, \\
	&\preceq  \left( \sum_{z_A,z'_A\in S_A} z_A \proj{0,z_A,z'_A} + \Lambda \proj{1,z_A,z'_A} \right) \otimes \I_{\BB'}, \\
	& = Z_A\otimes \I_\MM.
\end{align*}
To see that the first inequality is correct, consider a state $\ket{\zeta}$ in $\AA\otimes\MM$, $\ket\zeta= \sum_{z_B}\ket{\zeta_{z_B}}\ket{z_B}$. We get $\bra{\zeta} \left( \sum_{z_B} X_i^{(z_B)} \otimes \proj{z_B} \right) \ket{\zeta} = \sum_{z_B} \bra{\zeta_{z_B}} X_i^{(z_B)}\ket{\zeta_{z_B}} \leq \sum_{z_B} \bra{\zeta_{z_B}} Y \ket{\zeta_{z_B}} = \bra{\zeta}Y\otimes\I_{\BB'}\ket{\zeta}$ by \autoref{thm:SpectraXY}, hence $\sum_{z_B} X_i^{(z_B)} \otimes \proj{z_B} \preceq  Y \otimes \I_{\BB'}$. 
\end{proof}

\begin{proof}[Proof of \autoref{thm:SpectraXY}]
Let $l \to r$ be an EBM line transition, so by definition there exist two positive semidefinite matrices $X_0 \preceq  Y_0$ and a vector $\ket{\psi_0}$ such that $l = \prob[X_0, \psi_0]$ and $r = \prob[ Y_0, \psi_0]$. We will now make a succession of transformations to $X_0,Y_0$, and $\ket{\psi_0}$ in order to show that they can satisfy the properties of the Lemma. 

Notice that the size  of the matrices $X_0$ and $Y_0$ is unknown. We first see that we can decrease their size to at most $\abs{S}$. We start by diagonalizing $X_0$ and $Y_0$:
\begin{align*}
	X_0 = \sum_{x} x \Pi_{X_0}^{[x]}  \quad\quad\text{and}\quad\quad Y_0 = \sum_{y} y \Pi_{Y_0}^{[y]}.
\end{align*}
To remove the multiplicities of the eigenvalues, we go into the Hilbert space $\HH$, spanned by $\{\Pi_{X_0}^{[x]}\ket{\psi_0},  \Pi_{Y_0}^{[y]}\ket{\psi_0} \}$. This space has dimension at most $\abs{\supp(p) \cup \supp(q)} \leq \abs{S}$. We define the new $\ket{\psi} = \Pi_\HH \ket{\psi_0}$ as the projection of $\ket{\psi_0}$ on $\HH$ and the matrices $X$ and $Y$ by
\begin{align*}
	X = \sum_x x \Pi_X^{[x]} \quad\text{and}\quad\quad Y = \sum_y y \Pi_Y^{[y]},
\end{align*}
where $\Pi^{[x]}_X$ is the projector onto the one-dimensional space spanned by $\Pi_{X_0}^{[x]}\ket{\psi_0}$ and $\Pi^{[y]}_Y$ is the projector onto the one-dimensional space spanned by $\Pi_{Y_0}^{[y]}\ket{\psi_0}$.
These matrices have size at most $\abs{S}$.  By construction, the matrices $X$, $Y$ and the vector $\ket\psi$ satisfy the four properties 
\begin{itemize}
	\item $X \preceq Y$;
	\item $l = \prob[X,\psi]$ and $r = \prob[Y, \psi]$;
	\item The eigenvalues of $X$ are in $\supp(l)$ with multiplicity 1;
	\item The eigenvalues of $Y$ are in $\supp(r)$ with multiplicity 1.
\end{itemize}
Then, we will append the values in $S$ that are not yet into the spectra of $X$ and $Y$. This is done by increasing the dimension of the matrices and the vector $\ket\psi$ by the following algorithm: \\
For each $z$ in $S$ do:
\begin{itemize}
	\item if $z$ is in the spectrum of $X$ but not $Y$, $X \leftarrow X \oplus [0]$ and $Y \leftarrow  Y \oplus [z]$;
	\item if $z$ is in the spectrum of $Y$ but not $X$, $X \leftarrow X \oplus [z]$ and $Y \leftarrow Y \oplus  [\Lambda]$;
	\item if $z$ is neither in the spectrum of $X$ nor $Y$, $X \leftarrow X \oplus [z]$ and $Y \leftarrow Y  \oplus [z]$.
\end{itemize}
The output of this algorithm are matrices of size less or equal to $2\abs{S}$. We append extra 0 to $X$ and extra $\Lambda$ to $Y$ until they have exactly size $2\abs{S}$. We also increase the dimension of $\ket{\psi}$ by appending 0's. 

We have constructed two matrices $0 \preceq X \preceq Y$ and a vector $\ket\psi$ of dimension $2\abs{S}$ such that $l = \prob[X,\psi]$ and $r = \prob[Y,\psi]$. Moreover the spectrum of $X$ is exactly $\{0\}\cup S$ and all non zero eigenvalues have multiplicity one; the spectrum of $Y$ is exactly $S \cup \{\Lambda\}$ and all the eigenvalues in $S$ have multiplicity one. Thus, they can be decomposed as:
\begin{align*}
	X = \sum_{z\in S} z \proj{u_z} \quad\text{and}\quad\quad Y = \sum_{z\in S} z \proj{v_z} + \Lambda P,
\end{align*}
where the $\{ \ket{u_z} \}$ and $\{ \ket{v_z} \}$ are orthonormal families of vectors and $P$ is the projector onto the complement of $\mathrm{span}\{ \ket{v_z}\}$.

We now increase the size of $X$, $Y$, and $\ket\psi$ by appending 0's to all of them until they reach size $2|S|^2$. In particular, we can write  $X = \sum_{z\in S} z \proj{u'_z}$  and $Y = \sum_{z\in S} z \proj{v'_z} + \Lambda P'$ where $\ket{u'_z} = \ket{u_z} \otimes \ket{0^S}$, $\ket{v'_z} = \ket{v_z}\otimes\ket{0^S}$, and $P' = P \otimes \proj{0^S}$. As a consequence, $P'$ is a projector on a $|S|$-dimensional subspace of the $2|S|^2$-dimensional space. Then, let $U$ be a unitary that maps $\ket{v'_z}$ to $\ket{0,z,z}$ and sends $P'$ to $\sum_z \proj{1,z,z}$ (Such unitary exists since $P'$ is a projector onto a space of size $|S|$ orthogonal to the space spanned by the vectors $\{\ket{u'_z},\ z\in S\}$). We define $\ket{\varphi(z)} = Ue^{i\theta_z}\ket{u'_z}$ so that applying $U$ to $X$, $Y$, and $\ket\psi$ leads to:
\begin{align*}
	X = \sum_{z\in S} z \proj{\varphi(z)} \ ; \ Y = \sum_{z\in S}   z \proj{0,z,z} + \Lambda \proj{1,z,z} \ \text{and} \ \ket{\psi} = \sum_z \sqrt{l(z)}\ket{\varphi(z)}.
\end{align*}
\end{proof}
 
 \COMMENT{
\lnew{
\begin{lem}
This construction gives a protocol with $n$ messages and $2\lceil\log(S_A)\rceil + 2\lceil\log(S_B)\rceil+2$ qubits.
\end{lem}
}
 }


\section{Point games with valid transitions}
\label{sec:pointgames-valid}

Here is where we stand now: we have defined points games with EBM transitions, starting at points $1/2 [0,1] + 1/2[1,0]$ and ending at some point $[\beta,\alpha]$. We have also shown that for each point game with final point $[\beta,\alpha]$,  we can construct a weak coin flipping protocol and dual feasible points proving that the cheating probabilities are $P_A^* \leq \alpha$ and $P_B^* \leq \beta$. The final goal will thus be to find an EBM point game with final point $[1/2+\eps, 1/2+\eps]$ for any $\eps>0$. 

This task is quite challenging. We do not currently know of a direct way to handle EBM transitions ($\prob$ functions, matrices, vectors). The aim of this Section is to find an alternative characterization of EBM transitions that is easier to manipulate; those would be called valid transitions. 
For that matter, we turn to the help of convex geometry. If the space which we are considering were finite dimensional, a rather simple argument would suffice; we explain it below, but unfortunately we do not know how to make this simple argument work in the infinite dimensional case and hence more elaborate work is needed. 

\subsection{Moving between transitions and functions}
The first idea is to shift our view from transitions to functions, instead of considering transitions $p\to q$, we will look at functions $(q-p)$. The reasons behind this change are the following: it is easier to have one function $(q-p)$ than a pair $(p,q)$; the set of functions arising from EBM transitions has ``good'' geometry; and for all intents and purposes functions and transitions behave the same.


We start by defining the set $K$ of EBM functions and explain how point games with functions are equivalent to point games with transitions. This will allow us in the following subsections to look at the geometric properties of $K$.  We remark that in anticipation of the difficulties arising in infinite dimensions, we will also provide these definitions with a parameter $\Lambda$. For $\Lambda=\infty$ we have precisely the set of EBM functions, however later on we will need to consider sets with a finite $\Lambda>0$. 
For now, it might be instructive to just think of $\Lambda$ as $\infty$ and consider only the first items of the following definitions and lemmata.

\begin{definition}[$K$, EBM functions]
A function $h:[0,\infty)\to\R$ with finite support is an \dt{EBM function} if the line transition $h^- \to h^+$ is EBM, where $h^+:[0,\infty)\to[0,\infty)$ and $h^-:[0,\infty)\to[0,\infty)$ denote respectively the positive and the negative part of $h$ ($h=h^+-h^-$). We denote by $K$ the set of EBM functions.

For any finite $\Lambda \in (0,\infty)$, a function $h:[0,\Lambda)\to\R$ with finite support is an \dt{EBM function with support on $[0,\Lambda]$} if the line transition $h^- \to h^+$ is expressible by matrices with spectrum in $[0,\Lambda]$, where $h^+$ and $h^-$ denote respectively the positive and the negative part of $h$. We denote by $K_\Lambda$ the set of EBM functions with support on  $[0,\Lambda]$.
\end{definition}
As with transitions, we can also extend the notions of horizontal and vertical transitions to functions. 
\begin{definition}
A $\PP$-function $h:[0,\infty)\to\R$ is a function with finite support that has the property $\PP$. A function $t:[0,\infty)\times [0,\infty) \to \R$ is a 
\begin{itemize}
	\item \dt{horizontal $\PP$-function} if for all $y\geq0,\ t(\cdot,y)$ is a $\PP$-function;
	\item \dt{vertical $\PP$-function} if for all $x\geq0,\ t(x,\cdot)$ is a $\PP$-function.
\end{itemize}
\end{definition}
For now, we have only have seen $\PP$ being EBM, but we later see other properties, namely, valid and strictly valid. These  definitions are useful for defining point games with $\PP$-functions. To see how, consider a point game $1/2[0,1]+1/2[1,0] \to p_1 \to p_2 \to \cdots \to p_n$ and define the functions $t_0 = 1/2[0,1]+1/2[1,0]$ and $t_i = p_i - p_{i-1}$. It is easy to see that we have $p_i = \sum_{j=0}^i t_i$. Hence:
\begin{definition}[Point game with $\PP$- functions]
\label{def:pointgamewithfunctions}
A \textit{point game with $\PP$-functions} is a set $\{t_1,\dots,t_n\}$ of $n$ $\PP$-functions alternatively horizontal and vertical such that:
\begin{itemize}
	\item $1/2[0,1] +  1/2[1,0] + \sum_{i=1}^n t_i= [\beta,\alpha]$;
	\item $\forall j \in \{1,n\}, 1/2[0,1] + 1/2[1,0] +\sum_{i=1}^j t_i  \geq 0$.
\end{itemize}
We call $[\beta,\alpha]$ the final point of the game.
\end{definition}
The first condition simply rewrites the initial and final points; the second one expresses the fact that the $p_i$'s are non-negative.

We have seen how a point game with EBM transitions can be translated into a point game with EBM functions. The reverse also holds:
\begin{lem}
Given a point game with $n$ EBM functions and final point $[\beta,\alpha]$, we can construct a point game with $n$ EBM transitions and final point $[\beta,\alpha]$.
\end{lem}
\begin{proof}
Let us define $p_i= \sum_{j=1}^i f_j + 1/2[0,1] + 1/2[0,1]$. As a consequence, we have $p_0 = 1/2[0,1] + 1/2[0,1]$, $p_n = [\beta,\alpha]$, and for all $i\in\{0,n\}, p_i \geq 0$. Moreover, we have $p_{i+1} = p_i + f_{i+1}^+ - f_{i+1}^- \geq 0$. Since $f_{i+1}^+$ and $f_{i+1}^-$ have disjoint support, $\zeta = p_i - f_{i+1}^-  \geq 0$. We can rewrite the transition $p_i \to p_{i+1}$ by $\zeta + f_{i+1}^- \to \zeta + f_{i+1}^+$. This is an EBM transition since $f_{i+1}$ is an EBM function.
\end{proof}

As a consequence, from now on we can equivalently use functions or transitions, depending on what is most handy.

\subsection{Operator monotone functions and valid functions}
\label{sec:players}


We have defined the set $K$ of EBM functions. First, we show that the set $K$ is a convex cone.

\begin{definition}[Convex cone]
A set $C$ in a vector space $V$ is a cone if for all $x \in C$ and all $\lambda > 0$, $\lambda x \in C$. It is convex if for all $x,y \in C$, $x + y \in C$. 
\end{definition}

\COMMENT{
\begin{definition}[Closed set]
A set $C$ in a topological space $V$ is closed if for any sequence $\{t_i\}_i$ of points in $C$ that converges to a point $t$, we have $t \in C$.
\end{definition}
\begin{definition}[Dual cone]
We denote by $V'$ the space of continuous linear functionals from $V$ to $\R$. The dual cone of a set $C \subseteq V$ is \begin{align*}
C^* = \{\Phi \in V' \ |\  \forall h \in C,\ \Phi(h) \ge 0\}.
\end{align*}
\end{definition}
A closed convex cone has the property that it can be reconstructed from its dual cone
\begin{proposition}[\cite{Don09}]
Define $C^{**}$ as $C^{**}  = \{h \in V \ |\ \forall \Phi \in C^*, \Phi(h) \ge 0\}$. If $C$ is a convex cone, then $C^{**} = \cl(C)$, hence in particular if $C$ is a closed convex cone, then $C^{**} = C$.
\end{proposition}
}

Let us first describe the normed vector space we will be working in. This is the set  $V$ of functions from $[0,\infty)$ to $\R$ with finite support. $V$ is an infinite dimensional vector space spanned by the canonical basis $\{[x]\}_{x \in [0,\infty)}$ where $[x](y) = \delta_{x,y}$ is the Kronecker delta function. Each element $v$ of $V$ can be written as $v = \sum_x v(x)[x]$. The usual norm on this space is the 1-norm, which is defined for any $v = \sum_x v(x)[x]$ as $\norm{v}_1 = \sum_x |v(x)|$.


\begin{lem}\label{lem:ConvexCone}
$K$ is a convex cone. Also, for any $\Lambda \in (0,\infty)$, $K_\Lambda$ is a convex cone. 
\end{lem}
\begin{proof}
Fix $\Lambda >0$. Let $g,h\in K_{\Lambda}$, so $g^- \to g^+$ and $h^- \to h^+$ are two EBM line transitions, i.e. we can write them as: $g^- = \prob[X_g,\psi_g]$, $g^+ = \prob[Y_g,\psi_g]$, $h^-  = \prob[X_h,\psi_h]$ and $h^+ = \prob[Y_h,\psi_h]$.
(note that the dimensions of $X_g$ and $Y_g$ are not necessarily the same as the ones of $X_h$ and $Y_h$)

$K_{\Lambda}$ is a cone since for all $\lambda \geq 0,\ \lambda g = \lambda g^+ - \lambda g^- = \prob[Y_g,\sqrt\lambda\ket{\psi_g}]-\prob[X_g,\sqrt\lambda\ket{\psi_g}]$ and hence $\lambda g^- \to \lambda g^+$ is expressible by matrices with spectra in in $[0,\Lambda]$.

Let us finally show $K_{\Lambda}$ is convex. It is enough to prove that $g+h\in K_{\Lambda}$. Construct $X = X_g \oplus X_h = \left[ \begin{matrix} X_g & 0 \\ 0 & X_h \end{matrix} \right]$,  $Y = Y_g \oplus Y_h = \left[ \begin{matrix} Y_g & 0 \\ 0 & Y_h \end{matrix} \right]$ and $\ket{\psi} = \ket{\psi_g}\oplus\ket{\psi_h} = \left[\begin{matrix}\psi_g \\ \psi_h \end{matrix} \right].$ We now have
\begin{align*}
	g^- + h^- = \prob[X,\psi] \quad \text{and} \quad g^+ + h^+ = \prob[Y,\psi].
\end{align*}
Since we also have that $\sp(X),\sp(Y) \subset [0,\Lambda]$, we can conclude that $K_{\Lambda}$ is convex.
Notice that this proof holds for $\Lambda=\infty$, hence $K$ is also convex.
\end{proof} 

Let us now consider the dual cone of the set of EBM functions, denoted by $K^*$. 
\begin{definition}[Dual cone]
Let $C$ be a cone in a normed vector space $V$. We denote by $V'$ the space of continuous linear functionals from $V$ to $\R$. The dual cone of a set $C \subseteq V$ is \begin{align*}
C^* = \{\Phi \in V' \ |\  \forall h \in C,\ \Phi(h) \ge 0\}.
\end{align*}
\end{definition} 
It turns out that the dual cone of the set $K$ of EBM functions is exactly the set of operator monotone functions and a similar characterization holds for any $K_\Lambda$.

\begin{definition}[Operator monotone functions]
	A function $f:[0,\infty) \to\R$ is \emph{operator monotone} if for all positive semidefinite matrices $X\preceq Y$, we have $f(X)\preceq f(Y)$. 
	
	A function $f:[0,\Lambda] \to\R$ is \emph{operator monotone on $[0,\Lambda]$} if for all positive semidefinite matrices $X\preceq Y$ with spectrum in $[0,\Lambda]$, we have $f(X)\preceq f(Y)$. 
\end{definition}

We now show that the set $K^*$ is indeed equal to the set of operator monotone functions, up to an isomorphism. Indeed, note that there is a bijective mapping between $\Phi\in V'$ and $f_{\Phi}$ where $f_{\Phi}$ is a function on reals defined by $f_\Phi(x) = \Phi([x])$. This gives us by linearity of $\Phi$ that for a function $h=\sum_x h(x)[x]$ we have $\Phi(\sum_x h(x)[x]) = \sum_x h(x) f_\Phi(x)$. With this mapping, we can see elements of $K^*_{\Lambda}$ as functions on reals. Up to this mapping, the set $K^*_{\Lambda}$ is the set of operator monotone functions on $[0,\Lambda]$.

\begin{lem}\label{prop:OperatorMonotone}
$\Phi \in K^*$ if and only if $f_\Phi$ is operator monotone on $[0,\infty]$. 
Also, for any $\Lambda \in (0,\infty)$, $\Phi \in K^*_{\Lambda}$ if and only if  $f_\Phi$ is operator monotone on $[0,\Lambda]$.
\end{lem} 
\begin{proof}
Fix $\Lambda > 0$. Forward implication.  We first notice that $\Phi \in K^*_\Lambda$ implies
\begin{align}
	\forall h \in K_{\Lambda}, \ \sum_x f_\Phi(x)h(x) \geq 0. \label{eqn:OMFequiv}
\end{align}
This is immediate from the definition of $f_\Phi$. We now prove that a function $f$ with finite support on $[0,\Lambda]$ satisfies \autoref{eqn:OMFequiv} if and only if $f$ is operator monotone on $[0,\Lambda]$. The proof of this equivalence is based on the following observation:
\begin{align*}
	\sum_{x\in\sp(X)} f(x)\prob[X,\psi](x)  = \sum_{x \in \sp(X)} f(x) \bra{\psi} \Pi^{[x]} \ket{\psi}  =  \bra{\psi} \sum_{x \in \sp(X)}f(x) \Pi^{[x]} \ket{\psi}  = \bra{\psi} f(X) \ket{\psi}. 
\end{align*}
Then,
\begin{align*}
\forall h \in K_{\Lambda}, &\ \sum_x f(x)h(x) \geq 0 \\
& \Leftrightarrow 
\forall \ket{\psi}, \ \forall \ 0 \preceq X \preceq Y \mbox{ with } \sp(X),\sp(Y) \subset [0,\Lambda], \\&\quad\quad\quad
\sum_x f(x)\left(\prob[Y,\psi](x)  - \prob[X,\psi](x) \right) \geq 0 \\
& \Leftrightarrow 
\forall \ket{\psi}, \ \forall \ 0 \preceq X \preceq Y \mbox{ with } \sp(X),\sp(Y) \subset [0,\Lambda], \ \bra{\psi} f(X) \ket{\psi} \le \bra{\psi} f(Y) \ket{\psi} \\
& \Leftrightarrow \forall \ 0 \preceq X \preceq Y \mbox{ with } \sp(X),\sp(Y) \subset [0,\Lambda], \ f(X) \preceq f(Y) \\
& \Leftrightarrow f \textrm{ is operator monotone on } [0,\Lambda].
\end{align*}

For the reverse implication, consider a pair $(f_\Phi,\Phi)$ where $f_\Phi$ is a function with finite support on $[0,\Lambda]$ and $\Phi$ is its associated function in $V'$.  Hence by the previous series of equivalence, we have $\forall h \in K_{\Lambda},\ \Phi(h) \geq 0$. In order to prove that $\Phi\in K_\Lambda^*$ we need to show that $\Phi$ is continuous. Since $f_\Phi$ is operator monotone on $[0,\Lambda]$, $f_\Phi$ is increasing and $\forall x\in[0,\Lambda],\ f_\Phi(x) \in [f_\Phi(0),f_\Phi(\Lambda)]$, which means that $\norm{f_\Phi}_{\infty} < + \infty$. Thus, for any $h = \sum_x h(x)[x]$, we have $\Phi_f(h) = \sum_x h(x) f_\Phi(x) \le \norm{h}_1\norm{f_\Phi}_{\infty}$, and hence $\Phi$ is continuous.

Note that the proof holds for $K^*$, the set of operator monotone functions, up to the same mapping.
\end{proof}

Operator monotone functions are very well studied objects. 
In particular, we have the following analytic characterization:
\begin{lem}[\cite{Bha97}]
	\label{lem:DecompositionOMF}
	Any operator monotone function $f:[0,\infty)\to\R$ can be written as
	\begin{align*}
		f(t) = c_0 + c_1 t + \int_{0}^{+ \infty} \frac{\lambda t}{\lambda + t} \d w(\lambda),
	\end{align*}
	for a measure $w$ satisfying $\int_0^{+ \infty} \frac{\lambda}{1 + \lambda} \d w(\lambda) < + \infty$. \\
	Any operator monotone function $f:[0,\Lambda]\to\R$ can be written as
	\begin{align*}
		f(t) = c_0 + c_1 t + \int \frac{\lambda t}{\lambda + t} \d w(\lambda),
	\end{align*}
	with the integral ranging over $\lambda \in (-\infty, - \Lambda)\ \cup \ (0,+\infty)$.
\end{lem}

This characterization is what makes the usage of duality in our context helpful. In particular, we can now consider the dual of the set of operator monotone functions, $K^*$. 
This dual is denoted $K^{**}$, and is called the set of {\it valid functions}. 

\begin{definition}[Valid function]
	\label{def:validfunction}
	A function $h:[0,\infty) \to \R$ with finite support  is \dt{valid} if for every operator monotone function $f:[0,\infty)\to\R$, we have $\sum_{x\in\supp(h)}f(x)h(x) \geq 0$.
\end{definition}
The above is just restating the definition of the dual, but we provide it here for easier readability. Valid functions are strongly related to EBM functions since the bidual of a convex cone is the closure of the original cone:
\begin{lem}[\cite{BV04}]
Let $C \subseteq V$ be a convex cone, then $C^{**} = \cl(C)$.
\end{lem}
Unfortunately $K$ is not closed, so the valid functions are a superset of the EBM functions. In the next subsection, we see how we could circumvent this problem easily if we were in finite dimensions.

\subsection{The remaining of an ``easy'' argument}

Our goal remains to find an easy characterization of the set of EBM functions. So far, we have defined the set of EBM functions $K$, looked at its dual $K^*$, which is the set of operator monotone functions, and at the dual of the set of operator monotone functions, $K^{**}$, which we called the set of valid functions. 
However $K^{**}$ is larger than $K$, but intuitively not much larger since it is the closure of $K$, and $K$ and $K^{**}$ have the same interior since $K$ is convex:
\begin{fact}
	Let $C$ be a convex set, then $\interior(C) = \interior(\cl(C))$.
\end{fact}
Moreover, in finite dimensions, the interior of a dual cone can be expressed directly as a function of the primal cone:
\begin{lem}[\cite{BV04}]
Let $C$ be a cone in finite dimensional vector space $V$, then $\interior(C^*) = \{ \Phi\in V'\, | \, \forall h\in C-\{0\},\, \Phi(h) > 0\}$.
\end{lem}
This motivates the following definition.
\begin{definition}[Strictly valid function]
	\label{def:validfunction}
	A function $h:[0,\infty) \to \R$ with finite support  is \dt{strictly valid} if for every non-constant operator monotone function $f:[0,\infty)\to\R$, we have $\sum_{x\in\supp(h)}f(x)h(x) > 0$.
\end{definition}

At this point, if we were in finite dimensions, we could directly conclude that strictly valid functions are also EBM functions.
However, we are in infinite dimensions and this line of reasoning fails. The main difficulty lies in the fact that, as most cones in infinite dimensions, $K$ and $K^{**}$ have empty interiors whereas the set of strictly valid functions is not empty, and thus the set of strictly valid functions is not the interior of the set of valid functions. Nonetheless strictly valid functions play a similar role as we will prove the same statement, that strictly valid functions are EBM functions, but with a more cumbersome proof.

\subsection{Strictly valid functions are EBM functions}

Here is the main idea: We will look not at the set of EBM functions but at closed subsets of this set, that we have defined as $K_{\Lambda}$, for finite $\Lambda>0$. Since these sets are closed, we can exactly characterize them by their bidual. In other words, $K_\Lambda^{**}=K_\Lambda$. Now, if we look again at a strictly valid function, we can prove that it is in $K_\Lambda$ for some $\Lambda$ and hence in $K$. Hence, the set of strictly valid functions, which has a simple characterization as the strict dual of operator monotone functions, is a subset of EBM functions. 

The core of the argument is  that for any finite $\Lambda > 0$, $K_\Lambda$ is closed. Note that $K$, however, is not closed.

\begin{definition}[Closed set]
	A set $C$ in a topological space $V$ is closed if for any sequence $\{t_i\}_i$ of points in $C$ that converges to a point $t$, we have $t \in C$.
\end{definition}

\begin{lem}
	\label{thm:KLambdaIsClosed}
	For any finite $\Lambda > 0$, $K_{\Lambda}$ is closed.
\end{lem}

\begin{proof}
Fix a finite $\Lambda > 0$.
Let $\{t_i\}_{i \in \mathbb{N}}$ be a converging sequence of functions in $K_\Lambda$, and denote the limit of this sequence $t =  \lim_{i \to \infty} t_i$. The rest of the proof is devoted to show that $t \in K_\Lambda$. Denote $t= \sum_x t(x) [x]$ and $S$ the support of $t$, that is the set $S = \{x : t(x) \neq 0\}$. Note that $t$ is an element of $V$ so $t$ has finite support. Since the $t_i$ are EBM, we write $
t_i = \prob[Y_i, \psi_i] - \prob[X_i, \psi_i]$, with $0 \preceq X_i \preceq Y_i$. Each of the $X_i$'s and $Y_i$'s can be diagonalized:
\begin{align*}
	X_i = \sum_{x^{(i)}} x^{(i)} \Pi^{[x^{(i)}]} \quad\text{ and }\quad Y_i = \sum_{y^{(i)}} y^{(i)}  \Pi^{[y^{(i)}]},
\end{align*}
where $\Pi^{[x^{(i)}]}$ is the projector onto the eigenspace of $X_i$ with eigenvalue $x^{(i)}$. Since there will be no confusion, we drop the exponent $(i)$ from now on. Let us define the matrices:
\begin{align*}
	A_i = \sum_{x \in S} x \Pi^{[x]} + \sum_{x \notin S} 0 \cdot \Pi^{[x]}
	\quad\text{ and }\quad
	B_i = \sum_{y \in S} y \Pi^{[y]} + \sum_{y \notin S} \Lambda \cdot \Pi^{[y]}.
\end{align*}
First note that we immediately have $0 \preceq A_i \preceq X_i  \preceq Y_i \preceq B_i$ so we can define an EBM function $t'_i =  \prob[B_i, \psi_i] - \prob[A_i, \psi_i]$. The dimension of the matrices $A_i$ are not necessarily  identical, but this is not a problem. As done in the proof of \autoref{thm:SpectraXY} (``getting rid of the multiplicities'' and ``appending the missing eigenvalues''), we construct the positive semidefinite matrices $A'_i$, $B'_i$ of size $s=2\abs{S}$ and the vectors $\ket{\psi'_i}$ also of dimension $s$ such that $t'_i = \prob[B'_i, \psi'_i] - \prob[A'_i, \psi'_i]$.  Notice also that the spectra of the $A'_i$ and the $B'_i$ are in the interval $[0,\Lambda]$.

We show that $\lim_{i \to \infty} t'_i = t$. We write each $t_i$ as $ t_i = u_i + v_i, $ where $u_i = \sum_{x \in S} t_i(x) [x]$ and $v_i = \sum_{x \notin S} t_i(x) [x]$. Let $\eps_i = \sum_{x \notin S} t_i(x)$. Since $\lim_{i \to \infty} t_i = t$, we have $\lim_{i \to \infty} \eps_i = 0$.
Our construction of $t'_i$ implies that $t'_i = u_i + \eps^+_i [\Lambda] - \eps^-_i [0]$ with $\eps^+_i +\eps^-_i = \eps_i$. This means in particular that $\norm{t'_i - t_i}_1 \le \eps_i$. Since $\lim_{i \to \infty} \eps_i = 0$ and $\lim_{i \to \infty} t_i = t$, we conclude that $\lim_{i \to \infty} t'_i = t$.

We will now show that the limit of the sequence $\{t'_i\}_{i \in \mathbb{N}}$ is an element $t' \in K_{\Lambda}$ which will conclude the proof. We consider the sequence of triplets $\{(A'_i,B'_i,\ket{\psi'_i})\}_{i \in \mathbb{N}}$.  
Let $X_{\Lambda}^s$ the set of positive semidefinite matrices with spectrum in $[0,\Lambda]$ and $Y^s$ the set of quantum states of dimension $s$. An element of the sequence is an element of $X_{\Lambda}^s \times X_{\Lambda}^s \times Y^s$. Since $X_{\Lambda}^s$ and $Y^s$ are two compact sets, $X_{\Lambda}^s \times X_{\Lambda}^s \times Y^s$ is also a compact set. This means that our sequence of triplets has an accumulation point ($A'$, $B'$, $\ket{\psi'}$) even if this sequence does not necessarily converge.

Let us now define $t' = \prob[B', \psi'] - \prob[A', \psi']$. Since $0\preceq A'\preceq B'$, we have $t' \in K_{\Lambda}$. We can also see that $t'$ is an accumulation point of the sequence $\{t'_i\}_i$. Since the sequence of $t'_i$'s converges to $t$, we conclude that $t = t'$ and  $t \in K_{\Lambda}$.
\end{proof}

Since $K_\Lambda$ is a closed convex cone, it is characterized by its dual cone. 
\begin{cor}
	\label{prop:EquivalenceDualCone}
For any finite $\Lambda >0$, $K_\Lambda = \{h\in V \ |\ \forall \Phi \in K_\Lambda^*,\ \Phi(h) \ge 0\}$.
\end{cor}

Using the characterization of operator monotone functions on $[0,\Lambda]$ given by  \autoref{lem:DecompositionOMF}, we can restate \autoref{prop:EquivalenceDualCone} and characterize EBM functions on $[0,\Lambda]$ by three necessary and sufficient properties:
\begin{cor}\label{prop:EBM_Lambda_Characterization}
A function $h:[0,\Lambda]\to \R$ with finite support on $[0,\Lambda]$ is EBM on $[0,\Lambda]$ if and only if 
$\sum_{x} h(x) = 0$, $\sum_x x h(x) \ge 0$, and $ \forall \lambda \in (-\infty,-\Lambda] \cup (0,\infty), \ \sum_x \frac{\lambda x }{\lambda + x}h(x) \ge 0$.
\end{cor}

%
%

Our goal is to show that every strictly valid function is an EBM function. For this, we first find a characterization of strictly valid functions that looks similar to conditions \autoref{prop:EBM_Lambda_Characterization}. 

\begin{lem}
	\label{thm:simpleconditionsonvalidity}
	Let $h:[0,\infty)\to\R$ be a function with finite support such that $\sum_x h(x)=0$. The function $h$ is a strictly valid function if and only if for all $\lambda > 0,\ \sum_x \frac{-h(x)}{\lambda+x} > 0$, and is valid if and only if this inequality is large.
\end{lem}
\begin{proof}
An immediate consequence of \autoref{lem:DecompositionOMF} and the \hyperref[def:strictlyvalidfunction]{definition of strictly valid functions} is that $h$ is a strictly valid function if and only if 
\begin{enumerate}
	\item[\ding{182}] $\sum_x h(x) = 0$;
	\item[\ding{183}] for all $\lambda > 0,\ \sum_x \frac{\lambda x}{\lambda+x} h(x) > 0$;
	\item[\ding{184}] $\sum_x x\cdot h(x) > 0$.
\end{enumerate}
Condition \ding{184} is implied by condition \ding{183} in the limit $\lambda\to\infty$. Moreover, for all $\lambda>0$ we have:
\begin{align*}
\sum_x \frac{\lambda x}{\lambda+x} h(x) > 0 
&  \Leftrightarrow  \sum_x \left(1+ \frac{-\lambda}{\lambda+x}\right) h(x) \geq 0 
  \Leftrightarrow  \sum_x \frac{-1}{\lambda+x} h(x) > 0.
\end{align*}
The last equivalence is shown by using property \ding{182}.

The proof can be easily extended to handle the case of valid functions.
\end{proof}

%

%
%


We are now ready to show the main statement of this subsection:
\begin{lem}
	\label{thm:strictlyvalidisebm}
	Any strictly valid function is an EBM function.
\end{lem}
\begin{proof}
Fix $h$, a strictly valid function. We prove that there exists a $\Lambda>0$ such that $h$ is EBM on $[0,\Lambda]$ and hence is EBM. The proof easily extends to horizontal and vertical functions. 

The conditions of \autoref{thm:simpleconditionsonvalidity} are very close to the conditions in \autoref{prop:EBM_Lambda_Characterization}. We just need to show that there exists a $\Lambda>0$ such that 
\begin{align*}
	\forall \lambda < - \Lambda, \ \ \sum_x \frac{\lambda x}{\lambda + x}h(x) \geq  0.
\end{align*}
We have $\lim_{\lambda \rightarrow -\infty}  \ \ \sum_x \frac{\lambda x}{\lambda + x} h(x) = \sum_x x \ h(x) > 0$.
Consider the quantity $\sum_x \frac{\lambda x}{\lambda + x} h(x)$ as a function of $\lambda$. It is continuous in $\lambda$, so there exists a $\Lambda > 0$ such that
\begin{align*}
	\forall \lambda < - \Lambda, \ \ \sum_x \frac{\lambda x}{\lambda + x} h(x) \geq 0.
\end{align*}
\end{proof}

\subsection{From valid functions to EBM functions}
We have seen so far that a point game with strictly valid functions implies a point game with EBM functions. In this section we extend this result to valid functions, since, in the end, it will be easier to find a point game with valid transitions than with strictly valid ones.
Note that for all $\Lambda >0,\ K_\Lambda \subset K \subset K^{**}$, hence every EBM function is a valid function and thus a point game with EBM functions is also a point game with valid functions. More importantly, we prove that the converse is also ``approximately'' true:

\begin{thm}[Valid to EBM]
	\label{thm:ValidReverse}
	Given a point game with $2m$ valid functions and final point $[\beta,\alpha]$ and any $\eps >0$, we can construct a point game with $2m$ EBM functions and final point $[\beta+\eps,\alpha+\eps]$.
\end{thm}

The rest of this subsection is devoted to proving this theorem. 

To prove \autoref{thm:ValidReverse},  we, first, use \autoref{thm:strictlyvalidisebm} that shows that a point game with strictly valid functions is a point game with EBM functions.
Second, we show how to transform a point game with valid functions and final point $[\beta,\alpha]$ into a point game with strictly valid functions, hence EBM functions, and final point $[\beta+\eps,\alpha+\eps]$ for any $\eps>0$.
\begin{lem}
	\label{thm:validtostrictlyvalidgame}
	Fix $\eps>0$. Given a point game with $2m$ valid functions and final point $[\beta,\alpha]$, we can construct a point game with $2m$ strictly valid functions and final point $[\beta+\eps,\alpha+\eps]$.
\end{lem}
\begin{proof}
Consider a game with valid functions $\{t_1,\cdots,t_{2m}\}$. We will construct a new game with strictly valid functions $\{t'_1,\cdots,t'_{2m}\}$. The idea to ensure strict validity, is to shift each point by an extra $\eps/m$ to the right for horizontal functions and to shift them up by $\eps/m$ for vertical functions. After $2m$ funtions ($m$ horizontal and $m$ vertical), the final point will then be $[\beta+\eps,\alpha+\eps]$ as desired. 

For all $i \in \{1,2m\}$ and $\forall (x,y) \in [0,\infty)$, we define the shifted functions as:
\begin{align*}
	t'_i(x,y) =t_i^+(x-i\eps/m,y-(i-1)\eps/m) - t_i^-(x-(i-1)\eps/m,y-(i-1)\eps/m)    & \;\;\; \text{if } i \text{ is odd,}\\
	t'_i(x,y) =t_i^+(x-(i-1)\eps/m,y-i\eps/m) - t_i^-(x-(i-1)\eps/m,y-(i-1)\eps/m)  & \;\;\; \text{if } i \text{ is even.}
\end{align*}

Fix $i$ even. We prove that the function $t'_i$ is a strictly valid horizontal function. 
Note first that $\sum_x t'_i(x,y) = \sum_x t_i(x,y) = 0$.
Then, for all $y\in[0,\infty)$ and for all non-constant operator monotones functions we have:
\begin{align*}
	\sum_{x \in \supp(t_i^{\prime+})} t_i^{\prime+} (x,y)f(x)
		&= \sum_{x \in \supp(t'_i)} t_i^+(x - i\eps/m,y-(i-1)\eps/m)f(x) \\
		&= \sum_{x \in \supp(t_i^+)} t_i^+(x,y-(i-1)\eps/m)f(x+i\eps/m) \\
		&\geq \sum_{x \in \supp(t_i^-)} t_i^-(x,y-(i-i)\eps/m)f(x+i\eps/m) \\
		&= \sum_{x \in \supp(t_i^-)} t_i^{\prime-}(x + (i-1)\eps/m,y)f(x+i\eps/m) \\
		&= \sum_{x \in \supp(t_i^{\prime-})} t_i^{\prime-}(x,y)f(x + \eps/m), \\
      		& > \sum_{x \in \supp(t_i^{\prime-})}t_i^{\prime-}(x,y)f(x).
\end{align*}
The first inequality follows from the validity of $t_i$ and by noticing that if $f(x)$ is an operator monotone function in $x$ then $f(x+(i+1)\eps/m)$ is also an operator monotone function in $x$. The second strict inequality follows from the fact that every no constant operator monotone function is strictly increasing. 
A similar proof holds for vertical functions.
\end{proof}

\subsection{Examples of valid line transitions}
\label{sec:importantvalidtransitions}
As we said earlier, we can go back-and-forth between functions and transitions. Functions are more suited for proving equivalences between different types of point games, transitions are more suited for describing the point games and mapping them back to protocols. 

For the sake of completeness, we define now valid and strictly valid transitions and provide some examples of valid or strictly valid transitions.

\begin{definition}
Let $l,r:[0,\infty)\to[0,\infty)$ be two functions with finite support. The transition $l\to r$ is valid (resp. strictly valid) if the function $r-l$ is valid (resp. strictly valid).
\end{definition}

Let us now have a quick look at three transitions. Despite their simple expressions, these transitions play an important role in creating point games with arbitrarily small bias, and we will use them in the next two sections. Moreover, by using only these three transitions, Mochon gives a family of point games converging to final point $[2/3,2/3]$~\cite{Moc07}, i.e. a protocol with bias arbitrarily close to $1/6$.
\paragraph{Point raise} $w[x] \to w[x']$ with $x'\geq x$. \\
It is easy to see that for every operator monotone function $f$, $wf(x) \le wf(x\rq{})$ if and only if $x\rq{} \ge x$. By taking $f(x) = x$, we see that the condition is necessary. It is also sufficient since every operator monotone function is increasing.

\paragraph{Point merge} $w_1[x_1] + w_2[x_2] \to (w_1+w_2)[x_3]$ with $x_3 \geq \frac{w_1x_1 + w_2x_2}{w_1 + w_2}$. \\
Again, for every operator monotone function $f$, $w_1f(x_1) + w_2f(x_2) \le (w_1+w_2)f(x_3)$ if and only if $x_3 \geq \frac{w_1x_1 + w_2x_2}{w_1 + w_2}$. By taking $f(x) = x$, the above condition is necessary. This condition is also sufficient because operator monotone functions are concave.

\paragraph{Point split} $w[x] \to w_1 [x_1] + w_2[x_2]$ with $w = w_1 + w_2$ and $\frac{w}{x} \ge \frac{w_1}{x_1} + \frac{w_2}{x_2}$.\\
For every operator monotone function $f$, $wf(x) \le w_1 f(x_1) + w_2f(x_2)$ if and only if $\frac{w}{x} \ge \frac{w_1}{x_1} + \frac{w_2}{x_2}$. By considering the function $f(x) = -\frac{1}{\lambda + x}$ and the case where $\lambda \rightarrow 0$, we have $-\frac{w}{x} \le -\frac{w_1}{x_1} - \frac{w_2}{x_2}$ which shows that the above condition is necessary. 
We now show that the above condition is also sufficient. Assume that $\frac{w}{x} \ge \frac{w_1}{x_1} + \frac{w_2}{x_2}$. We want to verify that $wf(x) \le w_1 f(x_1) + w_2f(x_2)$ for $f(x) = -\frac{1}{\lambda + x}$. Let $q = \frac{1}{x}, q\rq{}_i = \frac{1}{x_i}$. Let a function $g(t) = -\frac{t}{1 + \lambda t}$. We have $-\frac{1}{\lambda + x} = g(q)$ and $-\frac{1}{\lambda + x_i} = g(q_i)$. This gives us
\begin{align*}
w f(x) = g(q) \le g\!\left(\frac{w_1q_1 + w_2q_2}{w_1 + w_2}\right) \le w_1g(q_1) + w_2 g(q_2) = w_1f(x_1) + w_2f(x_2).
\end{align*}
The first inequality holds because $g$ is decreasing and the second inequality holds because $g$ is convex. The special case of $f(x) = x$ follows by considering the limit $\lambda \to \infty$ when considering function $f(x) = \frac{\lambda x}{\lambda + x} = \lambda\left(1 + \lambda\cdot\frac{-1}{\lambda + x}\right)$.

A last property that will be useful later on is that no valid point game puts any weight on the point $[0,0]$.
\begin{lem}
	\label{thm:nopointattheorigine}
	A point game with valid transitions has no transition involving the point $[0,0]$.
\end{lem}
\begin{proof}
It is sufficient to prove that there is no valid line transition $l \to w[0] + (1-w)r$ where $l$ and $r$ are positive functions with finite support and $l(0)=r(0)=0$. By contradiction, assume there exists such transition. In that case the second condition of \autoref{thm:simpleconditionsonvalidity} implies that for all $\lambda >0$, we have $(1-w)\sum_x \frac{-\lambda}{\lambda+x} r(x) - w\geq \sum_x \frac{-\lambda}{\lambda+x} l(x)$. The contradiction is obtained by taking the limit $\lambda\to0$.
\end{proof}


\section{Time independent point games}
\label{sec:pointgames-tipg}

In the previous section, we showed that if  there exists a point game with valid transitions and final point $[\beta, \alpha]$, then for any $\eps>0$, there exists a weak coin flipping protocol with $P_A^* \leq \alpha + \eps$ and $P_B^* \leq \beta + \eps$. Moreover, given a line transition, it is easy to verify whether it is valid or not. Nevertheless, it is still not straightforward how to find a valid point game with arbitrarily small bias. 
We now introduce the last model, namely \dt{time independent point games} (TIPG).

As its name suggests, the idea behind a time independent point game is to remove the time ordering of the transitions. 
 This is done by dropping the second condition in \autoref{def:pointgamewithfunctions}, imposing that all the points should exist before being ``transitioned'', and by summing together all the horizontal functions on one hand, and all the vertical ones on the other hand.
\begin{definition}[Time independent point game]
A \emph{time independent point game} is a valid horizontal function $h$ and a valid vertical function $v$ such that
\begin{align*}
	h + v = 1[\beta,\alpha] - \ahalf [0,1] - \ahalf [1,0],
\end{align*}
for some $\alpha,\beta > 1/2$. We call the point $[\beta,\alpha]$ the final point of the game.
\end{definition}

The interest of this model in comparison to point games with valid transitions is obvious: we only need to find two valid functions, instead of a sequence with an appropriate order. Even simpler, since $h+v = 0$ almost everywhere (except in $[0,1],\ [1,0]$, and $[\beta, \alpha]$), our task basically boils down to finding a single valid function.

It is easy to construct a time independent point game with final point $[\beta,\alpha]$ from a point game with valid horizontal functions $(h_1,h_2,\ldots,h_n)$, valid vertical functions $(v_1,v_2,\ldots, v_n)$ and final point $[\beta,\alpha]$. As a matter of fact, we take $h=\sum_{i=1}^n h_i$ and $v = \sum_{i=1}^n v_i$. More interestingly, the reverse also holds:

\begin{thm}[TIPG to valid point games]
\label{prop:TipgToTdpg}
Given a time independent point game with a valid horizontal function $h$ and a valid vertical function $v$ such that
$h + v = 1[\beta,\alpha] - \ahalf [0,1] - \ahalf [1,0]$, we can construct, for all $\eps>0$, a valid point game with final point $[\beta+\eps,\alpha+\eps]$ and a number of transitions that depends on $\epsilon$.
\end{thm}

Before we prove the above theorem let us define \dt{transitively valid transitions}.
\begin{definition}[Transitively valid transition]
	Let $p,q:[0,\infty)\times[0,\infty) \to [0,\infty)$ be two functions with finite support. The transition $p\to q$ is \dt{transitively valid} if there exists a sequence of valid transitions $p_0 \to p_1, p_1\to p_2, \cdots, p_{m-1} \to p_m$ such that $p=p_0$ and $q=p_m$. 
\end{definition}
Our goal is to show that for every $\eps>0$ the transition $ \ahalf [0,1] + \ahalf [1,0] \to 1[\beta+\eps,\alpha+\eps]$ is transitively valid, which implies the theorem.

We start by the following technical lemma:

\begin{lem}\label{transitive}
If $p'\to q'$ is a transitively valid transition and $\zeta : [0,\infty) \times [0,\infty) \rightarrow [0,\infty)$ is a non-negative function with finite support, then $\delta p' + \zeta \to \delta q' + \zeta$ is also a transitively valid transition for all $\delta > 0$.
\end{lem}

\begin{proof}
It suffices to prove the statement for a line transition. Consider a valid line transition $l \to r$ and a non-negative function with finite support $\xi$, then for all $\delta>0$, $\delta l + \xi \to \delta r +\xi$ is also a valid transition since $\delta( r- l)$ is a valid function.
\end{proof}

In the following construction of a valid point game from a time indepedent point game, we also keep track of the number of valid transitions we need, which will correspond to the number of rounds of the protocol. We make this more precise at the end of the section. 

\begin{proof}[Proof of \autoref{prop:TipgToTdpg}]
  The proof consists of three main parts.\\

\noindent
{\em Part 1}: First, we show that the transition from $\ahalf [0,1] + \ahalf [1,0]$ to $[\beta,\alpha]$ is transitively valid in the presence of an extra set of points that we refer to as a ``catalyst'', since these points remain unchanged through the transition and their weight can be made arbitrarily small. 
Let us write $v = v^+ -v^-$, where $v^+$ and $ v^-$ are positive functions with disjoint supports and  $h = h^+ -h^-$, where  $h^+ $ and $h^-$ are again  positive functions with disjoint supports. Then, for any $\gamma>0$, we show that the following transition is transitively valid:
\begin{align} 
	\label{eqn:SmallCatalyst}
	\frac{1}{2} [0,1] + \frac{1}{2} [1,0] + \gamma v^-  \to  [\beta,\alpha]  + \gamma v^-. 
\end{align}
More precisely, we decompose this transitively valid transition into a sequence of $2\lceil 1/\gamma \rceil$ valid transitions.

By definition of $v$ and $h$,  $v^- \to v^+$ is a valid vertical transition and $h^- \to h^+$ is a valid horizontal transition. Hence, since $\frac{1}{2} [0,1] + \frac{1}{2} [1,0]$ is non-negative, we have by \autoref{transitive}, that
\begin{align*}
	\frac{1}{2} [0,1] + \frac{1}{2} [1,0] + v^- & \rightarrow \frac{1}{2} [0,1] + \frac{1}{2} [1,0] + v^+ 
\end{align*}
is a valid vertical transition. Moreover, remark that 
$h + v = (h^+ - h^-) + (v^+ - v^-) = - (\frac{1}{2} [0,1] + \frac{1}{2} [1,0]) + [\beta,\alpha]$
implies $\frac{1}{2} [0,1] + \frac{1}{2} [1,0] + v^+ = [\beta,\alpha] + h^- - h^+ + v^-.$
Define the function with finite support $\zeta = [\beta,\alpha] - h^+ + v^-$. $\zeta$ is a positive function: the only place where $\zeta$ could be negative is on the support of $h^+$. But $\zeta + h^- = \frac{1}{2} [0,1] + \frac{1}{2} [1,0] + v^+$ is non negative and $\supp(h^+) \cap \supp(h^-) = \emptyset$ so $\zeta$ is non negative. 
By \autoref{transitive}, we get that $\zeta + h^- \to \zeta + h^+ = [\beta,\alpha]  + v^- $ is a valid horizontal transition since $h^- \to h^+$ is a valid horizontal transition and $\zeta$ is non negative. This shows that the transition
\begin{align}
	\label{eqn:TIPGtoValidCatalyst}
	\frac{1}{2} [0,1] + \frac{1}{2} [1,0] + v^-  \to  [\beta,\alpha]  + v^- 
\end{align}
is transitively valid and can be decomposed into a sequence of two valid transitions. It remains to show how to reduce the weight associated to $v^-$ in the transition.

\begin{lem}\label{lem:ReduceCatalysts}
	Suppose we have a transitively valid transition $p + \xi \rightarrow q + \xi$, then for any $\gamma> 0$, the transition $p + \gamma \xi  \to q + \gamma \xi$ is transitively valid.
\end{lem}
\begin{proof}
	Pick $\gamma'$ the largest inverse of an integer such that $\gamma \geq \gamma'$, that is $\gamma'= 1 / \lceil 1/\gamma \rceil$. By \autoref{transitive}, the following transition
	\begin{align*}
		p + \gamma' \xi &=  (1-\gamma')p + \gamma'\left(p + \xi\right) 
		\rightarrow (1-\gamma')p + \gamma'\left(q + \xi\right) = (1 - \gamma')p + \gamma'\xi + \gamma' q 
	\end{align*}
	is transitively valid, since $\gamma'(p + \xi) \rightarrow \gamma'(q + \xi)$ is transitively valid and  $ (1-\gamma')p$ is non-negative. If we repeat one more time, we can see that the following transition is again transitively valid
\begin{align*}
(1 - \gamma')p + \gamma'\xi + \gamma' q   = (1- 2\gamma')p + \gamma'\left(p + \xi\right) + \gamma' q \\
 \rightarrow (1-\gamma')p + \gamma'\left(q + \xi\right) + \gamma' q = (1 - 2\gamma')p + \gamma'\xi + 2\gamma' q.
\end{align*}
By repeating this process $1/\gamma' = \lceil 1/\gamma \rceil$ times, we end up with $q+\gamma'\xi$ and by adding on both sides $(\gamma-\gamma')\xi$, we obtain that the transition $p + \gamma\xi \to q+\gamma\xi$ is transitively valid. 
\end{proof}
From this Lemma, we conclude the proof that the transition of \autoref{eqn:SmallCatalyst} is transitively valid. Note that we also showed that the transition in \autoref{eqn:SmallCatalyst} can be decomposed into a sequence of $2\lceil1/\gamma\rceil$ valid transitions.\\

\noindent
{\em Part 2}: 
The second part of the proof consists in showing how to construct this small weight catalyst. In fact, we start from our initial distribution of points $\ahalf [0,1] + \ahalf [1,0]$ and by using a small part of their weight we construct the catalyst. Then, by performing the transitively valid transition we explained in Part 1, we can move to the final point $[\beta,\alpha]$ with weight almost one; nevertheless there is a small weight left in other points that we will deal with in Part 3.

More precisely, let us define
\begin{align*}
	m = \min_{(x,y)\in\supp(v^-)} \{\max\{x,y\}\}.
\end{align*}
By \autoref{thm:nopointattheorigine}, $v^-(0,0) = 0$ and hence $m>0$. In addition, for all $(x,y)\in\supp(v^-)$ it holds that $(x \ge m \ \mbox{or} \ y \ge m)$. This means, that there exist $a,b \ge 0$ with $\sum_{x,y} v^-(x,y) = \norm{v^-} = a + b$ such that the transition
\begin{align}
	\label{eqn:CreateCatalyst}
	a[0,m] +b [m,0] \rightarrow v^-
\end{align}
is transitively valid. All points $[x,y]$ in the support of $v^-$ can be reached through a point raise of $[0,m]$ or $[m,0]$, and thus the transition \autoref{eqn:CreateCatalyst} can be decomposed into a sequence of two valid transitions.

Let us now assume that $m<1$ (in fact, the case $m \geq 1$ is simpler and we will consider it afterwards). Let $m_x,m_y$ such that 
\begin{align}\label{mtransition}
[0,1] \rightarrow \frac{am}{a+b}[0,m] + \frac{b+a(1-m)}{a+b}[0,m_y] \; \mbox{and} \; [1,0] \rightarrow \frac{bm}{a+b}[m,0] + \frac{a+b(1-m)}{a+b}[m_x,0]
\end{align}
 are valid line transitions (such $m_x,m_y$ always exist). 

For any $\delta > 0$, we prove that the following transitions are transitively valid:

\begin{align*}
	\frac{1}{2} [0,1] + \frac{1}{2}[1,0]
	\rightarrow   
	  &\frac{1 -\delta}{2} [0,1] + \frac{\delta a m}{2(a+b)}[0,m] + \frac{\delta (b+a(1-m))}{2(a+b)}[0,m_y]\\
	 \quad &+  \frac{1-\delta}{2} [1,0] +  \frac{\delta b m}{2(a+b)}[m,0] + \frac{\delta (a+b(1-m))}{2(a+b)}[m_x,0] &\text{by \autoref{mtransition} }\\
\rightarrow & (1-\delta) \left( \frac{1}{2} [0,1] + \frac{1}{2} [1,0] + \frac{\delta m}{2(1-\delta)(a+b)} v^- \right) \\
\quad & + \quad \frac{\delta (b+a(1-m))}{2(a+b)}[0,m_y] + \frac{\delta (a+b(1-m))}{2(a+b)}[m_x,0] &\text{by \autoref{eqn:CreateCatalyst} }\\
\rightarrow & (1-\delta) [\beta,\alpha] + \frac{\delta m}{2(a+b)} v^- \\
\quad & + \quad \frac{\delta (b+a(1-m))}{2(a+b)}[0,m_y] + \frac{\delta (a+b(1-m))}{2(a+b)}[m_x,0] &\text{by \autoref{eqn:SmallCatalyst}}
\end{align*}
For $m\geq 1$, we start by considering the raises $[0,1] \rightarrow [0,m]$ and $[1,0] \rightarrow [m,0]$ and then continue as above.
Let
$
	\xi = \frac{m}{2(a+b)} v^-  + \frac{ b+a(1-m)}{2(a+b)}[0,m_y] + \frac{a+b(1-m)}{2(a+b)}[m_x,0]
$. We have shown that for any $\delta > 0$, the transition 
\begin{align}\label{xi}
\frac{1}{2} [1,0] + \frac{1}{2}[0,1] \to (1-\delta) [\beta,\alpha] + \delta\xi
\end{align}
is transitively valid. 
Note that we also showed that the transition in \autoref{xi} can be decomposed into a sequence of $2 + 2 + 2\left\lceil \frac{2(1-\delta)\norm{v^-}}{\delta m} \right\rceil$ valid transitions.\\

\noindent
{\em Part 3}:
In this part, we get rid of the $\delta\xi$ in \autoref{xi} by merging it with the final point $[\beta,\alpha]$. This has as effect that the final point moves to $[\beta+\eps,\alpha+\eps]$. To do this, we use the following Lemma. 
\begin{lem}\label{lem:FinalMerge}
Given $\eps > 0$ and a function $\xi : [0,\infty) \times [0,\infty) \rightarrow [0,\infty)$ with finite support and $\sum_{(x,y)\in\supp(\xi)} \xi(x,y) = 1$, there exists $0 < \delta < 1$ such that $
(1-\delta)[\beta,\alpha] + \delta \xi \rightarrow [\beta + \eps, \alpha + \eps]$ is transitively valid.
\end{lem}
\begin{proof}
By point raising, there exist values $n_x$ and $n_y$ such that $\xi \rightarrow [n_x,n_y]$ is transitively valid. Moreover, by further point raising, we can have $n_x > \beta + \eps$ and $n_y > \alpha + \eps$. We pick $\delta$ and $\delta'$ such that the following two point merges are valid:
\begin{align*}
\delta'[n_x,\alpha] + \delta[n_x,n_y] & \rightarrow (\delta + \delta')
[n_x,\alpha + \eps], \\
(1 - \delta - \delta')[\beta,\alpha + \eps] + (\delta' + \delta)[n_x,\alpha + \eps] & \rightarrow [\beta + \eps,\alpha + \eps]. 
\end{align*}
This is possible by taking $\delta,\delta'>0$ that satisfy $\delta' \alpha + \delta n_y = (\delta + \delta')(\alpha + \eps)$ and $(1 - \delta - \delta')\beta + (\delta' + \delta)n_x = \beta + \eps$. In other words,
\begin{align*}
	\delta = \frac{\eps^2}{(n_x-\beta)(n_y-\alpha)} \quad \text{and} \quad \delta' = \frac{\eps}{n_x-\beta}\left(1-\frac{\eps}{n_y-\alpha}\right).
\end{align*}
We conclude that
\begin{align*}
(1-\delta)[\beta,\alpha] + \delta \xi & \rightarrow (1-\delta)[\beta,\alpha] + \delta [n_x,n_y] &\text{$\xi \rightarrow [n_x,n_y]$ transitively valid}\\
& \rightarrow (1-\delta - \delta')[\beta,\alpha] + \delta'[n_x,\alpha] + \delta [n_x,n_y]  &\text{valid  point raise}\\
& \rightarrow (1-\delta - \delta')[\beta,\alpha] + (\delta' + \delta) [n_x,\alpha + \eps] &\text{valid merge}\\
& \rightarrow (1-\delta - \delta')[\beta,\alpha + \eps] + (\delta' + \delta) [n_x,\alpha + \eps] &\text{valid point raise}\\
& \rightarrow [\beta+ \eps,\alpha + \eps] &\text{valid merge.}
\end{align*}
\end{proof}
Note that the transition in \autoref{lem:FinalMerge} can be decomposed into a sequence of six valid transitions. 

This concludes the proof of the \autoref{prop:TipgToTdpg}.
\end{proof}

\paragraph{Number of rounds} The proof of \autoref{prop:TipgToTdpg} gives an explicit way of constructing a valid point game with final point $[\beta+\eps,\alpha+\eps]$ from any time independent point game with final point $[\beta,\alpha]$. This construction creates a point game with $10+2\left\lceil \frac{2(1-\delta)\norm{v^-}}{\delta m} \right\rceil$ valid transitions where $\delta = \frac{\eps^2}{(n_x-\beta)(n_y-\alpha)}$ and $v^-, m,n_x,n_y,\alpha$ and $\beta$ are parameters of the original TIPG. In the TIPGs we will consider, we  have that $m \geq \frac{1}{2}$ and let $\Gamma=\max \{ n_x,n_y\}$. Note also that $\norm{v^-}=\norm{v}/2=\norm{h}/2$. 
Then, the number of transitions is $O\left( \frac{ \norm{h} \Gamma^2}{\eps^2} \right)$. This corresponds to the number of rounds of the protocol.  Hence, we can restate \autoref{prop:TipgToTdpg} as

\begin{cor}
	\label{cor:numberOfRounds}
Assume there exists a time independent game with a valid horizontal function $h=h^+-h^-$ and a valid vertical function $v=v^+-v^-$ such that $h + v = 1[\beta,\alpha] - \ahalf [0,1] - \ahalf [1,0]$. Let $\Gamma$ the largest coordinate of all the points that appear in the TIPG game. Then, for all $\eps>0$, we can construct a point game with $O\left( \frac{ \norm{h} \Gamma^2}{\eps^2}\right)$ valid transitions and final point $[\beta+\eps,\alpha+\eps]$.
\end{cor}


\section{Construction of a time independent point game achieving bias $\eps$}
\label{sec:construction}

In this section we construct for every $\eps>0$ a game with final point $[1/2 + \eps, 1/2 + \eps]$. Moreover, the number of qubits used in the protocol will be $O(\log \frac{1}{\eps})$ and the number of rounds $\left(\frac{1}{\eps}\right)^{O(\frac{1}{\eps})}$.

\subsection{Overview of the game}

First, to simplify the analysis, we will place all points, except the initial points $[0,1]$ and $[1,0]$, on a regular grid of step $\omega$, i.e. all points can be written as $[a\omega,  b\omega]$ for some $a,b\in\N$.

We will describe a family of protocols parametrized by $k$, whose final point $[\alpha,\alpha]$ is such that $\alpha=\zeta \omega = \frac{1}{2}+O( \frac{1}{k})$. Hence, to achieve a small bias $\eps$, one needs to use the protocol with $k = O(\frac{1}{\eps})$.

\begin{figure}
\begin{tikzpicture}[scale=.645]
\draw (-.5,-.5) rectangle (7.5,7.5);
\node at (0,7) {a)};
\draw[->] (0.5,0.5) -- (0.5,7); \draw[->] (0.5,0.5) -- (7,0.5);
\node at (0,3) {\tiny $\zeta\omega$};
\node at (0,4) {\tiny $1$};
\node at (0,6) {\tiny $\Gamma\omega$};
\node at (3,0) {\tiny $\zeta\omega$};
\node at (4,0) {\tiny $1$};
\node at (6,0) {\tiny $\Gamma\omega$};
\draw[line width=2pt,color=red!80!black] (.5,3) -- (.5,6);
\draw[line width=2pt,color=green!80!black] (3,.5) -- (6,.5);
\fill[color=black] (.5,4) circle (.1);
\fill[color=black] (4,.5) circle (.1);
\draw[->,opacity=.6] (.6,4) to [bend right=90] (.5,6);
\draw[->,opacity=.6] (.6,4) to [bend right=80] (.5,5.5);
\draw[->,opacity=.6] (.6,4) to [bend right=70] (.5,5);
\draw[->,opacity=.6] (.6,4) to [bend right=60] (.5,4.5);
\draw[->,opacity=.6] (.6,4) to [bend left=70] (.5,3);
\draw[->,opacity=.6] (.6,4) to [bend left=60] (.5,3.5);
\draw[->,opacity=.6] (4,.6) to [bend left=90] (6,.5);
\draw[->,opacity=.6] (4,.6) to [bend left=80] (5.5,.5);
\draw[->,opacity=.6] (4,.6) to [bend left=70] (5,.5);
\draw[->,opacity=.6] (4,.6) to [bend left=60] (4.5,.5);
\draw[->,opacity=.6] (4,.6) to [bend right=70] (3,.5);
\draw[->,opacity=.6] (4,.6) to [bend right=60] (3.5,.5);
\end{tikzpicture}
\begin{tikzpicture}[scale=.645]
\draw (-.5,-.5) rectangle (7.5,7.5);
\node at (0,7) {b)};
\draw[->] (0.5,0.5) -- (0.5,7); \draw[->] (0.5,0.5) -- (7,0.5);
\node at (0,3) {\tiny $\zeta\omega$};
\node at (0,4) {\tiny $1$};
\node at (0,6) {\tiny $\Gamma\omega$};
\node at (3,0) {\tiny $\zeta\omega$};
\node at (4,0) {\tiny $1$};
\node at (6,0) {\tiny $\Gamma\omega$};
\draw[line width=2pt,color=black] (.5,3) -- (.5,6);
\draw[line width=2pt,color=black] (3,.5) -- (6,.5);
\fill[color=red!80!black] (2,3) circle (.1);
\fill[color=green!80!black] (3,2) circle (.1);
\draw[rounded corners=14pt,opacity=.6] (4.5,1) -- (5.6,2.7) -- (6,4) -- (6,5.5) -- (5,5);
\draw[rounded corners=14pt,opacity=.6] (1,4.5) -- (2.7,5.6) -- (4,6) -- (5.5,6) -- (5,5);
\draw[opacity=.6] (5,5) -- (4,4);
\draw[->,opacity=.6] (4,4) to [bend right=20] (3,2.2);
\draw[->,opacity=.6] (4,4) to [bend left=20] (2.2,3);
\end{tikzpicture}
\begin{tikzpicture}[scale=.645]
\draw (-.5,-.5) rectangle (7.5,7.5);
\node at (0,7) {c)};
\draw[->] (0.5,0.5) -- (0.5,7); \draw[->] (0.5,0.5) -- (7,0.5);
\node at (0,3) {\tiny $\zeta\omega$};
\node at (0,4) {\tiny $1$};
\node at (0,6) {\tiny $\Gamma\omega$};
\node at (3,0) {\tiny $\zeta\omega$};
\node at (4,0) {\tiny $1$};
\node at (6,0) {\tiny $\Gamma\omega$};
\fill[color=black] (2,3) circle (.1);
\fill[color=black] (3,2) circle (.1);
\fill[color=red!80!black,opacity=0.7] (3,3) circle (.1);
\fill[color=green!80!black,opacity=0.7] (3,3) circle (.1);
\draw[->,opacity=.6] (2.2,3) -- (2.8,3);
\draw[->,opacity=.6] (3,2.2) -- (3,2.8);
\end{tikzpicture}
\caption{Schematic representation of the game.  The initial points are in black, the final points are colored in red if they are part of the horizontal ladder and in green of the vertical ladder. The arrows represents the idea of the movements of the points. a) Each point is split into many points (represented by a line) on their axes. b) The ladder combines the points on the axes into 2 points. c) The raises create the final point of the game.}
\label{fig:overview}
\end{figure}
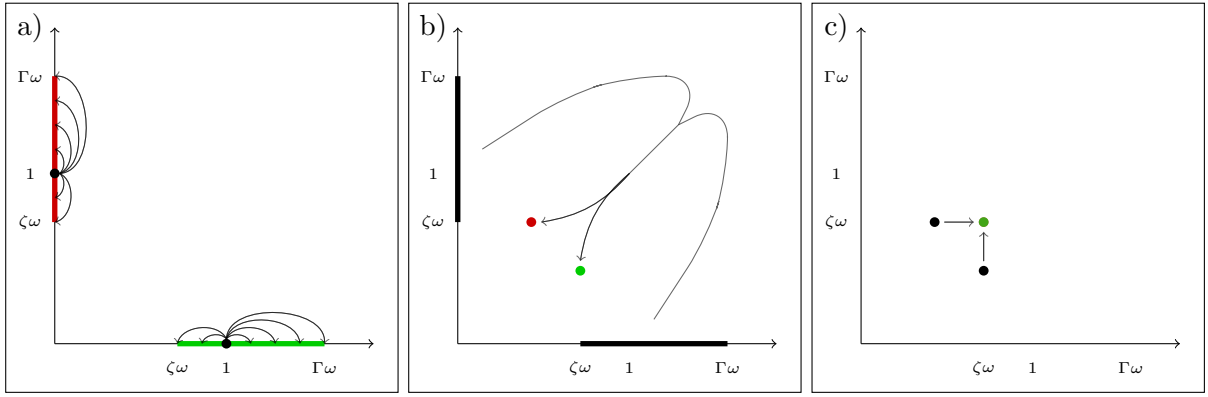

The point game with parameter $k$ consists of the three following steps
\begin{description}
	\item[Split:] The point $[0,1]$ is split into points on the vertical axis between positions $\alpha=\zeta \omega$ and $\Gamma\omega$, and same for the point $[1,0]$. See a) in \autoref{fig:overview}. The weight of the points $[0,j\omega]$ and $[j\omega,0]$ for $\zeta \leq j \leq \Gamma$ is given by a function $\mathrm{split}(j)$, so that the transitions are valid (in fact, they will be strictly valid). 
	
	\item [Ladder of width $k$:] 
	The points on the axes will be transitioned to two final points $[\alpha-k\omega,\alpha]$ and $[\alpha, \alpha-k\omega]$. See b) in \autoref{fig:overview}. This transition is transitively valid, meaning that there exists a sequence of valid transitions starting from the points on the axes and ending at the two final points. These transitions use more points on the grid, in fact, points whose $x$ and $y$ coordinates are between $\alpha-k\omega$ and $\Gamma \omega$. Moreover, on every line, except the two axes, there are at most $2k+1$ points.
		
	\item [Raise:] The two points are raised into a final point $[\alpha,\alpha]$. See c) in \autoref{fig:overview}. This raise is valid. 
\end{description}
More formally:
\begin{align}
	\frac{1}{2}[0,1] + \frac{1}{2}[1,0]
	& \quad \xrightarrow{\mathrm{\,\,split\,\,\,}}  \quad \sum_{j=\zeta}^{\Gamma} \mathrm{split}(j) [0,j\omega] + \sum_{j=\zeta}^{\Gamma} \mathrm{split}(j) [j\omega,0] \label{eqn:splittransition}\\
	& \quad\xrightarrow{\mathrm{ladder}} \quad \frac{1}{2}[\alpha-k\omega,\alpha] + \frac{1}{2} [\alpha,\alpha-k\omega] \label{eqn:LadderTransition} \\
	& \quad\xrightarrow{\mathrm{\,\,\,raise\,\,}} \quad 1 [\alpha,\alpha] \nonumber
\end{align}

The next two sections are devoted to proving that for any $k$ there exist values for the parameters $\omega$ and $\Gamma$, such that the two initial splits are valid,  the ladder is a transitively valid transition, and $\alpha = \frac{1}{2}+O(\frac{1}{k})$.

\subsection{The ladder}

\subsubsection{Description}

We define a {\em ladder} as a time independent point game, described by a valid horizontal function $h_{\mathrm{lad}}$  and a valid vertical function $v_{\mathrm{lad}}$ such that
\begin{align*}
h_{\mathrm{lad}} + v_{\mathrm{lad}} = \frac{1}{2} [ \alpha-k\omega,\alpha] + \frac{1}{2} [ \alpha,\alpha-k\omega] -\sum_{j=\zeta}^\Gamma\mathrm{split}(j)\left([0,j\omega]+[j\omega,0]\right),
\end{align*}
where for each axis, $\mathrm{split}(j)$ is a distribution on points on the axis that arises from a split of the initial point $[0,1]$ or $[1,0]$ (see \autoref{sec:importantvalidtransitions} for the definition of a split); its exact parameters will be defined shortly. 

Our goal is to find functions $h_{\mathrm{lad}}$  and $v_{\mathrm{lad}}$ that may put weight also on other points on the grid, such that both $h_{\mathrm{lad}}$  and $v_{\mathrm{lad}}$ are valid functions. Moreover, we need to do this while finding a function $\mathrm{split}(j)$ such that the initial split is a valid transition. 

We restrict ourselves to symmetric point games, i.e. games where the horizontal function $h$ and the vertical function $v$ satisfy:
\begin{align}
	\label{eqn:symmetric}
	v(x, y)= - h(x, y)\ \mbox{ and also }\ h(x,y) = -h(y, x),
\end{align}
except for the final and initial points.
This implies that there are no points on the diagonal, i.e.  $\forall z,\ h(z, z) = v(z, z)=0$ and that if $h_{\mathrm{lad}}$ is a valid function, then $v_{\mathrm{lad}}$ will be valid too. 

In fact, we do not know of a simple way of transitioning the points on the axes to the two final points. To do so, we must add new points and perform a sequence of transitions that will gradually transition the points on the axes to the final two points.  We now describe these extra points. 

We call a \emph{rung} in the ladder, the function corresponding to the points at a fixed height, i.e. $h_{\mathrm{rung}} (\cdot, y) = \sum_x h(x,y)$ for some $y$.  Our ladder will have rungs from height $\alpha=\zeta\omega$ to $\Gamma\omega$.
A \emph{ladder of width $k$} has rungs that have $2k$ points centered on the diagonal and one point on the $y$-axis. More formally for $\zeta \leq  j \leq \Gamma$, the $x$-coordinate of the points of the rung at height $j\omega$ are
\begin{align}
	\label{eqn:BluePoints}
	\left\{ 0,\,\left(j-k\right)\omega,\, \left(j-k+1\right)\omega, \dots,\,(j-1)\omega,\,(j+1)\omega,\dots,\left(j+k-1\right)\omega,\,\left(j+k\right)\omega\right\}.
\end{align}

Note that we will have $k\omega \ll \zeta\omega$, meaning that all points are at least a constant away from the axes.  
Some principles governing the ladder are shown in \autoref{fig:ladder}. First in a), we give a schematic representation of all the points defined in \autoref{eqn:BluePoints} that are the horizontal part of the ladder. For every point of the vertical axis, we consider $2k$ new points on the same height and centered on the diagonal.The vertical part is constructed using the symmetry \hyperref[eqn:symmetric]{relation \ref*{eqn:symmetric}}  we imposed. Both, the horizontal part and the vertical part of the ladder are represented in~b). By symmetry, we also have that if a point is in the horizontal and in the vertical part of the ladder, the sum of its weights is null. All these points are located on the overlap of the two parts of the ladder.  There are then only a few remaining points: the initial points on the axes, the final points in $[\alpha-k\omega,\alpha] $ and  $[\alpha,\alpha-k\omega] $ and 4 ``triangles''. We get rid of these triangles by considering ladders where the weight on these triangles is 0. Hence, all the points we consider are shown in c). 

\subsubsection{Finding weights for the points in the ladder}

Our goal, now, is to find weights for the points on the axes and the remaining points on the plane such that the function $h_{\mathrm{lad}}$ is valid. By symmetry, $v_{\mathrm{lad}}$ will also be valid. As we have said, $h_{\mathrm{lad}} = \sum_{j=\zeta}^{\Gamma} h_{\mathrm{rung}}^j$, where $h_{\mathrm{rung}}^j$ is the rung function on height $j\omega$.

Note that when finding the functions $h_{\mathrm{rung}}^j$, we also fix the weights on the points of the original split (this gives an explicit definition of the split function in \autoref{eqn:splittransition}). In \autoref{sec:Split} we show that for any $k$, we can choose parameters $\Gamma$ and $\omega$ such that this split is valid, and such that  $\alpha = \frac{1}{2}+O(\frac{1}{k})$.

\begin{figure}
\begin{tikzpicture}[scale=.645]
\draw (-.5,-.5) rectangle (7.5,7.5);
\node at (0,7) {a)};
\draw[->] (0.5,0.5) -- (0.5,7); \draw[->] (0.5,0.5) -- (7,0.5);
\node at (0,3) {\tiny $\zeta\omega$};
\node at (0,4) {\tiny $1$};
\node at (0,6) {\tiny $\Gamma\omega$};
\node at (3,0) {\tiny $\zeta\omega$};
\node at (4,0) {\tiny $1$};
\node at (6,0) {\tiny $\Gamma\omega$};
\filldraw[color=red!80!black,opacity=0.5] (2,3) -- (4,3) -- (7,6) -- (5,6) -- (2,3);
\draw[line width=2pt,color=red!80!black,opacity=0.5] (.5,3) -- (.5,6);
\end{tikzpicture}
\begin{tikzpicture}[scale=.645]
\draw (-.5,-.5) rectangle (7.5,7.5);
\node at (0,7) {b)};
\draw[->] (0.5,0.5) -- (0.5,7); \draw[->] (0.5,0.5) -- (7,0.5);
\node at (0,3) {\tiny $\zeta\omega$};
\node at (0,4) {\tiny $1$};
\node at (0,6) {\tiny $\Gamma\omega$};
\node at (3,0) {\tiny $\zeta\omega$};
\node at (4,0) {\tiny $1$};
\node at (6,0) {\tiny $\Gamma\omega$};
\filldraw[color=red!80!black,opacity=0.5] (2,3) -- (4,3) -- (7,6) -- (5,6) -- (2,3);
\filldraw[color=green!80!black,opacity=0.5] (3,2) -- (6,5) -- (6,7) -- (3,4) -- (3,2);
\draw[line width=2pt,color=red!80!black,opacity=0.5] (.5,3) -- (.5,6);
\draw[line width=2pt,color=green!80!black,opacity=0.5] (3,.5) -- (6,.5);
\end{tikzpicture}
\begin{tikzpicture}[scale=.645]
\draw (-.5,-.5) rectangle (7.5,7.5);
\node at (0,7) {c)};
\draw[->] (0.5,0.5) -- (0.5,7); \draw[->] (0.5,0.5) -- (7,0.5);
\node at (0,3) {\tiny $\zeta\omega$};
\node at (0,4) {\tiny $1$};
\node at (0,6) {\tiny $\Gamma\omega$};
\node at (3,0) {\tiny $\zeta\omega$};
\node at (4,0) {\tiny $1$};
\node at (6,0) {\tiny $\Gamma\omega$};
\filldraw[color=red!80!black,opacity=0.5] (3,3) -- (4,3) -- (6,5)  -- (6,6) -- (5,6) -- (3,4) -- (3,3);
\filldraw[color=green!80!black,opacity=0.5] (3,3) -- (4,3) -- (6,5)  -- (6,6) -- (5,6) -- (3,4) -- (3,3);
\draw[line width=2pt,color=red!80!black,opacity=0.5] (.5,3) -- (.5,6);
\draw[line width=2pt,color=green!80!black,opacity=0.5] (3,.5) -- (6,.5);
\fill[color=red!80!black,opacity=0.5] (2,3) circle (.05);
\fill[color=green!80!black,opacity=0.5] (3,2) circle (.05);
\end{tikzpicture}
\caption{Schematic construction of the ladder. a) The horizontal part of the ladder. b) Superposition of the horizontal part and the vertical part of the ladder. By symmetry, the sum of the weights of the point in the overlap is 0. Except the final points, the weights of the points in the 4 ``triangles'' with no overlap will be set to 0 by truncation. c) All the points actually involved in the ladder transition.}
\label{fig:ladder}
\end{figure}
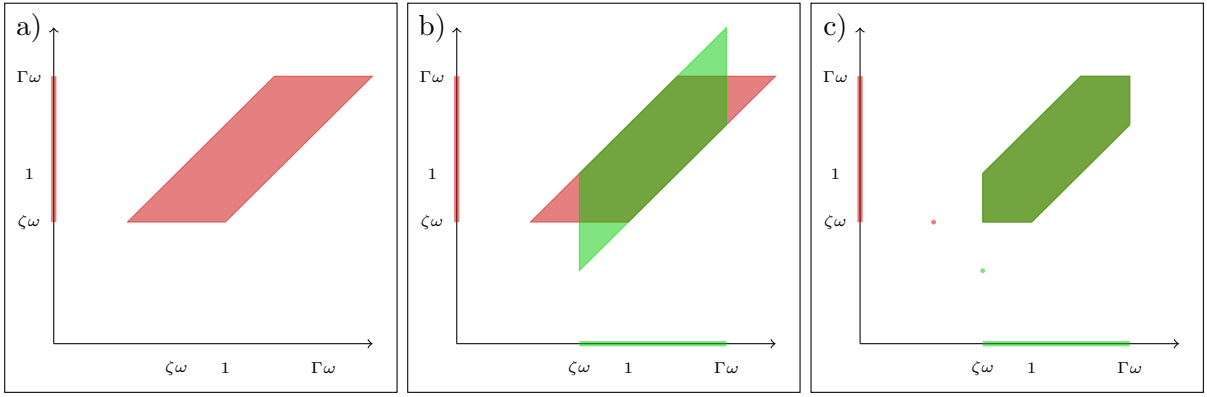

Finding valid functions $h_{\mathrm{rung}}^j$ is not an easy task. Let us assume that $h_{\mathrm{rung}}^j = \sum_i w_i [x_i]$ and that we would like to verify that the function is valid. According to \autoref{thm:simpleconditionsonvalidity}, it would be necessary to prove that for all $\lambda>0$, $\sum_i \frac{-w_i}{\lambda+x_i} = \frac{-\sum_i w_i \prod_{k\neq i}(\lambda+x_k)}{\prod_k (\lambda+x_k)} \geq 0$. In practice, this means checking that the polynomial $f(-\lambda) = -\sum_i w_i \prod_{k\neq i}(\lambda+x_k) \geq 0$ for all $\lambda<0$. In other words, from a given function we can construct a polynomial $f$ such that the validity of the function holds when $f(\lambda) \geq 0$ for all $\lambda<0$. The following Lemma basically does the reverse. Given a low degree polynomial $f$ with $f(\lambda) \geq 0$ for all $\lambda<0$, we construct a valid function $h_{\mathrm{rung}}$ by assigning weights $\frac{-f(x_i)}{\prod_{k\neq i}(x_k-x_i)}$ to $[x_i]$.

\begin{lem}
\label{thm:WeightPolynomial}
Let $x_{1},\ldots,x_{2k+1}\in\mathbb{R}_{+}$ be different points, $f\in\mathbb{R}\left[X\right]$
be a real polynomial such that:
\begin{itemize}
	\item the absolute value of its leading coefficient is 1,
	\item $\deg(f) \leq 2k-1$,
	\item $\forall x<0,\ f(x) \geq 0$,
\end{itemize}
then $\forall C>0$ the following function $h_{\mathrm{rung}}$ is a valid function:
\begin{align}
	h_{\mathrm{rung}} &= \sum_{i=1}^{2k+1} \frac{-C\cdot f(x_{i})}{\underset{j \neq i}{\prod}(x_{j}-x_{i})}\left[x_i\right].
	\label{h as polynom}
\end{align}
\end{lem}

This lemma gives us some freedom in the choice of the polynomials.  
At height $j\omega$ we are looking for a function $h_{\mathrm{rung}}(\cdot,j\omega)$ of the form
\begin{align*}
	h_{\mathrm{rung}}^j =
	\frac{-C_j\cdot f(0,j\omega)}{\prod_{\substack{l=-k \\ l \neq 0}}^k ((j+l)\omega) }[0,j\omega]
	+
	\sum_{\substack{i=-k \\ i \neq 0} }^{k}
		\frac	{-C_j \cdot f((j+i)\omega, j\omega)}
			{-(j+i)\omega\underset{\substack{l\neq i \\ l \neq 0} }{\prod}(\left(l-i\right)\omega)}
		\left[(j+i)\omega,j\omega \right].
\end{align*}
Note that the first term corresponds to the point on the axis and the remaining $2k$ points are centered on the diagonal (see \autoref{eqn:BluePoints} for the $x$-coordinate of the points of the rung at height $j\omega$). 

We start by taking $C_j = C/j\omega$, so that the weights are symmetric and we can add the rung functions more easily. We have

\begin{align*}
	h_{\mathrm{rung}}^j =
	\frac{- C\cdot f(0,j\omega)}{\prod_{l=-k}^k ((j+l)\omega) }[0,j\omega]
	+
	\sum_{\substack{i=-k \\ i \neq 0} }^{k}
		\frac	{C \cdot f((j+i)\omega, j\omega)}
			{((j+i)\omega)(j\omega)\underset{\substack{l\neq i \\ l \neq 0}}{\prod}((l-i)\omega)}
		\left[(j+i)\omega,j\omega \right].
\end{align*}
Adding all the rungs of different heights, we get:
\begin{align}\label{hlad}
h_{\mathrm{lad}} =
\sum_{j=\zeta}^{\Gamma}\left(
	\frac{- C\cdot f(0,j\omega)}{\prod_{l=-k}^k ((j+l)\omega) }  \left[0,j\omega\right]
	+
	\sum_{\substack{i=-k \\ i\neq 0}}^{k}
		\frac	{C\cdot f(\left(j+i\right)\omega,\, j\omega)}
			{((j+i)\omega)(j\omega) \underset{\substack{l \neq i \\ l \neq 0}}{\prod}((l-i)\omega)}
	[(j+i)\omega,\ j\omega]
	\right).
\end{align}

Now, we need to ensure that when we add the the functions $h_{\mathrm{lad}}$ and $v_{\mathrm{lad}}$ we are only left with the points on the axes and the two final points $[ \alpha-k\omega,\alpha]$ and $[ \alpha,\alpha-k\omega]$. Since they are symmetric functions, this means that it is necessary and sufficient to put zero weight on the points that appear only in $h_{\mathrm{lad}}$ or only in $v_{\mathrm{lad}}$. This corresponds to the ladder truncation we described in \autoref{fig:ladder}.
To do so, we choose the symmetric polynomial $f$ so that the weight on the points with $x$-coordinate in $\{\alpha-(k-1)\omega, \alpha-(k-2)\omega, \dots, \alpha-\omega\}$ as well as in $\{(\Gamma+1)\omega,\dots,(\Gamma+k)\omega\}$ is zero (both for the horizontal and vertical function). Since we want a polynomial of degree at most $2k-1$ on each variable, we have only one possibility:
\begin{align}\label{f}
	f(x,y) =
	(-1)^{k+1}\prod_{i=1}^{k-1}\left(\alpha-i\omega-x\right)\left(\alpha - i\omega-y\right) \prod_{i=1}^{k}\left(\Gamma\omega+i\omega-x\right)\left(\Gamma\omega+i\omega-y\right).
\end{align}

When we look at the polynomial $f(x,y)$ as a polynomial on one variable, then it needs to satisfy the hypothesis of \autoref{thm:WeightPolynomial}. Indeed, the absolute value of its leading coefficient is 1, its degree is $2k-1$ and $\forall\,\zeta\leq j\leq\Gamma$ and $\forall x<0$ we have $f(x,j\omega)\geq 0$. This is true, since the only negative terms in the product are the $(k-1)$ negative values $(\alpha-i\omega-j\omega)$ and another $(k+1)$ from the $(-1)^{k+1}$ factor.

Hence, we conclude that $h_{\mathrm{lad}}$ is a valid horizontal function (and similarly $v_{\mathrm{lad}}$ is a valid vertical function). Since $v_{\mathrm{lad}}(x,y)=-h_{\mathrm{lad}}(x,y)$ everywhere except on the points on the axes and $[ \alpha,\alpha-k\omega]$ and $[ \alpha-k\omega, \alpha]$, adding the two function leaves us with the desired outcome. We have proved that
\begin{lem}\label{lem:ladder}
The function $h_{\mathrm{lad}}$ as defined in \autoref{hlad} and \autoref{f} and, by symmetry, the function $v_{\mathrm{lad}}$, are valid functions that satisfy
	\begin{align*}
	h_{\mathrm{lad}} + v_{\mathrm{lad}} = \frac{1}{2} [ \alpha-k\omega,\alpha] + \frac{1}{2} [ \alpha,\alpha-k\omega] -\sum_{j=\zeta}^\Gamma\mathrm{split}(j)\left([0,j\omega]+[j\omega,0]\right)
	\end{align*}
with
\begin{align}\label{split}
	\mathrm{split}(j) = \frac{ C\cdot f(0,j\omega)}{\prod_{l=-k}^k ((j+l)\omega)}.
\end{align}
\end{lem}

In \autoref{sec:Split}, we will see that for every $k$, we can find values for the parameters $C, \omega$ and $\Gamma$  and take $\alpha = \frac{1}{2}+ \frac{c}{k}$ for some constant $c$, such that the two initial splits defined by the function $\mathrm{split}(j)$ are valid. Before that, we provide the proof of \autoref{thm:WeightPolynomial}.

\subsubsection{Proof of \autoref{thm:WeightPolynomial}}

We start with the following technical Lemma

\begin{lem}
	\label{thm:MochonLemma30}
	Let $x_{1},\dots,x_{m}$ be $m\ (\geq2)$ distinct values in $\mathbb R$ and $f$ a polynomial such that $\mathrm{deg}(f) \leq m-2$, then
	\begin{align*}
		\sum_{i=1}^{m}\frac{f(x_{i})}{\prod_{j \neq i} (x_{j} - x_{i})} = 0.
	\end{align*}
\end{lem}
\begin{proof}
By induction on $\deg(f)$. For $\deg(f)=0$ we need to prove that $\sum_{i=1}^{m}\prod_{j \neq i} \frac{1}{x_{j} - x_{i}} = 0$. This proof is also done by induction, this time on the number of points. The initialization is trivial. For $m>2$ and $1<i<m$ we have the identity:
\begin{align*}
\frac{1}{(x_{1}-x_{i})(x_{m}-x_{i})} &= \frac{1}{x_{m}-x_{1}}\left(\frac{1}{x_{1}-x_{i}} - \frac{1}{x_{m}-x_{i}} \right).
\end{align*}
which gives us:
\begin{align*}
	\sum_{i=1}^{m} \prod_{\substack{j=1 \\ j\neq i}}^{m} \frac{1}{x_{j}-x_{i}}
	&= \prod_{j=2}^m\frac{1}{x_{j}-x_{1}} + \frac{1}{x_{m}-x_{1}} \sum_{i=2}^{m-1} \prod_{\substack{j=2 \\ j\neq i}}^{m-1}\frac{1}{x_{j}-x_{i}}\left(\frac{1}{x_{1}-x_{i}} - \frac{1}{x_{m}-x_{i}}\right) + \prod_{j=1}^{m-1}\frac{1}{x_{j}-x_{m}} \\
	&= \frac{1}{x_{m}-x_{1}} \left[ \sum_{i=2}^{m-1} \left( \prod_{\substack{j=1 \\ j\neq i}}^{m-1}\frac{1}{x_{j}-x_{i}} - \prod_{\substack{j=2 \\ j\neq i}}^{m}\frac{1}{x_{j}-x_{i}} \right) +\prod_{j=2}^{m-1} \frac{1}{x_{j}-x_{1}} - \prod_{j=2}^{m-1} \frac{1}{x_{j}-x_{m}} \right] \\
	&= \frac{1}{x_{m}-x_{1}} \left( \sum_{i=1}^{m-1} \prod_{\substack{j=1 \\ j\neq i}}^{m-1}\frac{1}{x_{j}-x_{i}} -  \sum_{i=2}^{m} \prod_{\substack{j=2 \\ j\neq i}}^{m}\frac{1}{x_{j}-x_{i}}   \right)
\end{align*}
by induction, each of the terms in the parenthesis is 0. That concludes the proof for $\deg(f) = 0$.
If $\deg(f)\leq k$, there exists a constant $\kappa\neq0$ and a polynomial $g$ with $\deg(g) < k$ such that $f(x)=\kappa\prod_{j=1}^{k}(x_{j}-x) + g(x)$. We then have:
\begin{align*}
\sum_{i=1}^{m}\frac{f(x_{i})}{\prod_{\substack{j=1 \\ j\neq i}}^{m} (x_{j} - x_{i})} &=
\underbrace{\kappa \sum_{i=k+1}^{m}\prod_{\substack{j=k+1 \\ j\neq i}}^{m} \frac{1}{x_{j}-x_{i}}}_{=0\ \text{initialization case}} + \underbrace{\sum_{i=1}^{m}\frac{g(x_{i})}{\prod_{\substack{j=1 \\ j\neq i}}^{m} (x_{j} - x_{i})}}_{=0\ \text{by induction}}.
\end{align*}
\end{proof}

\begin{proof}[Proof of \autoref{thm:WeightPolynomial}] Fix $C>0$. We have two statements to show: first that $\sum_{x} h_{\mathrm{rung}}(x) = 0$. This is immediate using the previous lemma. Secondly, let us fix $\lambda > 0$, we need to show that $Q = \sum_{x}\left(\frac{-1}{\lambda + x}\right)h_{\mathrm{rung}}(x) \geq 0$.
\begin{align*}
Q \geq0  \iff & \sum_{i=1}^{2k+1} \frac{1}{\lambda + x_{i}} \cdot \frac{f(x_{i})}{\prod_{\substack{j=1 \\ j\neq i}}^{2k+1} x_{j}-x_{i}} \geq 0 &\text{ since } C > 0.
\end{align*}
Using \autoref{thm:MochonLemma30}, with the points $\{-\lambda,x_{1},\dots,x_{2k+1}\}$ we have:
\begin{align*}
\frac{f(-\lambda)}{\prod_{j=1}^{2k+1}\left(x_{j}-(-\lambda)\right)} + \sum_{i=1}^{2k+1}\frac{f(x_{i})}{\left((-\lambda) - x_{i}\right) \prod_{\substack{j=1 \\ j \neq i}}^{2k+1}(x_{j}-x_{i})} = 0.
\end{align*}
Combining the two previous equations, we get:
\begin{align*}
	Q \geq 0 \iff \frac{f(-\lambda)}{\prod_{j=1}^{2k+1}\left(x_{j}-(-\lambda)\right)}  \geq 0.
\end{align*}
By assumption, $f(-\lambda) \geq 0$ and all the terms $(x_j+\lambda)$ are positive.

\end{proof}


\subsection{Validity of initial splits}
\label{sec:Split}


Here we show that for every $k$, we can find values for the parameters $C, \omega$, and $\Gamma$ as functions of $k$,  and take $\alpha = \frac{1}{2}+ \frac{c}{k}$ for some constant $c$, such that the 
two initial splits defined in \autoref{split} by the function $\mathrm{split}(j)$ are valid.

\begin{lem}[The splits are valid]
	\label{thm:validsplits}
For any $k$, we can find $\omega$ and $\Gamma$, such that by taking $\alpha = \frac{1}{2}+ \frac{c}{k}$ for some constant $c$, the functions
	\begin{align*}
		h_{\mathrm{split}} = \sum_{j=\zeta}^{\Gamma} \mathrm{split}(j) [j\omega,0] - \frac 1 2 [1,0]
		\quad \mathrm{and} \quad
		v_\mathrm{split} = \sum_{j=\zeta}^{\Gamma} \mathrm{split}(j) [0,j\omega] - \frac 1 2 [0,1]
	\end{align*}
	are valid functions, where $\mathrm{split}(j) = \frac{C\cdot f(0,j\omega)}{\prod_{l=-k}^k ((j+l)\omega)}$
	and  $C = \frac{1}{2} \cdot \left(\sum_{j=\zeta}^{\Gamma}\frac{f(0,j\omega)}{\underset{l=-k}{\overset{k}{\prod}}\omega\left(j+l\right)}\right)^{-1}$.
\end{lem}

\begin{proof}
We consider the vertical split (similarly for the horizontal one). By the analysis of point splits in \autoref{sec:importantvalidtransitions}, it suffices to verify two conditions: first, that  $\sum_{j=\zeta}^\Gamma\mathrm{split}(j)=\frac{1}{2}$, which holds for our choice of $C$. Second, we need to show that
\begin{align*}
\frac{1}{2} > \sum_{j=\zeta}^{\Gamma}\frac{C \cdot f(0,j\omega)}{j\omega\underset{l=-k}{\overset{k}{\prod}}\omega\left(j+l\right)}.
\end{align*}
By replacing the value of $C$ we get:
\begin{align*}
\sum_{j=\zeta}^{\Gamma}\frac{f(0,j\omega)}{\underset{l=-k}{\overset{k}{\prod}}\omega\left(j+l\right)} >
\sum_{j=\zeta}^{\Gamma}\frac{f(0,j\omega)}{j\omega\underset{l=-k}{\overset{k}{\prod}}\omega\left(j+l\right)}.
\end{align*}
We know that
\begin{align*}
f(0,j\omega)
	&= (-1)^{k+1}\prod_{i=1}^{k-1}(\alpha - i\omega) (\alpha - i\omega - j\omega)   \prod_{i=1}^{k}  (\Gamma\omega + i\omega) (\Gamma\omega + i\omega - j\omega)  \\
	&= \left[ \prod_{i=1}^{k-1} (j\omega - (\alpha- i\omega)) \prod_{i=1}^{k} (\Gamma\omega + i\omega - j\omega) \right] \left[\prod_{i=1}^{k-1} (\alpha - i\omega) \prod_{i=1}^{k} (\Gamma\omega + i\omega) \right].
\end{align*}
We can divide each side of the inequality with the square of the second bracket, which is independent of $j$ and non-zero, and have
\begin{align}
	\label{eqn:Mochon176}
	\sum_{j=\zeta}^\Gamma p(j\omega) > \sum_{j=\zeta}^\Gamma \frac{p(j\omega)}{j\omega}
	\,\,\,\,\, \text{with} \,\,\,\,\,
	 p(j\omega) = \prod_{i=1}^{k-1} \frac{j\omega-(\alpha-i\omega)}{\alpha-i\omega}
	                    \prod_{i=1}^k \frac{\Gamma\omega+i\omega-j\omega}{\Gamma\omega+i\omega}
	                    \prod_{i=-k}^k \frac{1}{j\omega + i\omega}.
\end{align}

To conclude the proof we will need to show that \hyperref[eqn:Mochon176]{Inequality~\ref*{eqn:Mochon176}} holds for $\alpha = \frac{1}{2}+ \frac{c}{k}$ for some constant $c$ and for some chosen $\omega$ and $\Gamma$. This is done in \autoref{conclude}; an alternative proof is given in \autoref{sec:resources}.
\end{proof}

\subsubsection{Concluding the proof of existence}\label{conclude}

To conclude the proof of \autoref{thm:validsplits}, and with it the proof of existence of a WCF protocol with an arbitrarily small bias, we will need to show that \hyperref[eqn:Mochon176]{Inequality~\ref*{eqn:Mochon176}} holds for $\alpha = \frac{1}{2}+ \frac{c}{k}$ for some constant $c$ and for some chosen $\omega$ and $\Gamma$. 
In this subsection we do exactly this. To this end, we provide the following technical lemma. 
This lemma shows that \hyperref[eqn:Mochon176]{Inequality~\ref*{eqn:Mochon176}} can be simplified by approximating $p$ by the function $f(x) = \left(\frac{x-\alpha}{\alpha}\right)^{k-1}x^{-2k-1}$ and the sum by an integral. The approximation of $p$ by $f$ is quantified by the function $E$ and the approximation of the sum by the integrals by the functions $\epsilon_l$ and $\epsilon_r$.

\begin{lem}
	\label{lem:FinalSufficientCondition}
	Fix any values for the parameters $\omega,\Gamma,\alpha$, such that $\omega < 1$, $\Gamma \omega^2 > 1$, and $\alpha > 1/2$.
	Define the functions 
		$f(x) = \left(\frac{x-\alpha}{\alpha}\right)^{k-1}x^{-2k-1}$,
		$\tilde{f}(x) = f(x)/x$,
		$\epsilon_l(\omega,\Gamma) = \Gamma\omega^4 \abs{\fpp}_\infty$,
		$\epsilon_r(\omega,\Gamma) = \Gamma\omega^4 \abs{\tilde{\fpp}}_\infty$, and
		$E(\omega) = \frac{\left(1+2k\omega\right)^{2k+1}}{\left(1-4k\omega \right)^{4k+1}}$.
If the following inequality holds
\begin{align}
	\label{eqn:FinalSufficientCondition}
	\int_\alpha^{\Gamma\omega^2} f(x) \d x -\epsilon_l(\omega,\Gamma) > E(\omega) \left[\int_\alpha^\infty \tilde{f}(x) \d x + \epsilon_r(\omega,\Gamma) \right],
\end{align}
then \hyperref[eqn:Mochon176]{Inequality~\ref*{eqn:Mochon176}} also holds for the same $\omega,\Gamma$, and $\alpha$.
\end{lem}
The proof of this lemma is delayed to the next subsection. This lemma implies that in order to prove Lemma \ref{thm:validsplits}, we just need to find $\Gamma,\omega$ satisfying $\omega < 1,\Gamma\omega^2 > 1$ and $\alpha = \frac{1}{2} + \frac{c}{k}$ for some constant $c$. 

We introduce the Beta function $B(a,b)$. In \autoref{sec:resources} we will also need a related function, the  Regularized Incomplete Beta function $I(z;a,b)$. These two functions are defined by:
\begin{align*}
	B(a,b) &= \int_0^1 t^{a-1} (1-t)^{b-1}\d t = \alpha^a \int_\alpha^\infty (t-\alpha)^{b-1}t^{-a-b} \d t,\\
	I(z;a,b) &= \frac{1}{B(a,b)} \int_0^z t^{a-1} (1-t)^{b-1}\d t  = \frac{\alpha^a}{B(a,b)}  \int_{\alpha/z}^\infty (t-\alpha)^{b-1}t^{-a-b} \d t.
\end{align*}
Both functions have very nice form when $a$ and $b$ are integers 
\begin{align}
	\label{eqn:BetaWithInteger}
	I(z;a,b) = \sum_{j=a}^{a+b-1} {a+b-1 \choose j} z^j (1-z)^{a+b-1-j} \text{\quad and \quad} B(a,b) = \frac{(a-1)!(b-1)!}{(a+b-1)!}.
\end{align}
We now prove our final lemma which concludes the proof of \autoref{thm:validsplits}.

\begin{lem}\label{easy}
For any $k$, we can find $\Gamma,\omega$ such that \hyperref[eqn:Mochon176]{Inequality~\ref*{eqn:FinalSufficientCondition}} is satisfied for any $\alpha > \frac{k+1}{2k+1}$
\end{lem}
\begin{proof}
Let us take $\Gamma=\omega^{-3}$, so that when we look at the limit $\omega \to 0$, we have $\Gamma \to \infty$,  $\Gamma \omega^2 \to \infty$ and $\Gamma\omega^4 \to 0$.  We have 
\begin{align*}
\lim_{\omega \rightarrow 0} \int_\alpha^{\Gamma\omega^2} f(x) \d x -\epsilon_l(\omega,\Gamma) = \int_\alpha^{\infty} f(x) \d x
\end{align*}
and 
\begin{align*}
\lim_{\omega \rightarrow 0}E(\omega) \left[\int_\alpha^\infty \tilde{f}(x) \d x + \epsilon_r(\omega,\Gamma) \right] = \int_\alpha^\infty \tilde{f}(x) \d x.
\end{align*}
Using the Beta function, we have $\int_\alpha^\infty f(x) \d x = \frac{B(k+1,k)}{\alpha^{2k+2}}$ and $\int_\alpha^\infty \tilde{f}(x) \d x = \frac{B(k+2,k)}{\alpha^{2k+3}}$ which implies  
\begin{align*}
\int_\alpha^\infty f(x) \d x > \int_\alpha^\infty \tilde{f}(x) \d x \quad \mbox{ for } \alpha > \frac{B(k+2,k)}{B(k+1,k)} = \frac{k+1}{2k+1}.
\end{align*}
From there, we have
\begin{align*}
\lim_{\omega \rightarrow 0} \int_\alpha^{\Gamma\omega^2} f(x) \d x -\epsilon_l(\omega,\Gamma) > \lim_{\omega \rightarrow 0}E(\omega) \left[\int_\alpha^\infty \tilde{f}(x) \d x + \epsilon_r(\omega,\Gamma) \right].
\end{align*}
This means that we can find a small $\omega$ and $\Gamma = \omega^{-3}$ such that \hyperref[eqn:Mochon176]{Inequality~\ref*{eqn:FinalSufficientCondition}} holds for any $\alpha > \frac{k+1}{2k+1}$.
\end{proof}

Together with the proof of \autoref{lem:FinalSufficientCondition} (see \autoref{fsc} below) this concludes the proof of existence of a WCF protocol.

In \autoref{sec:resources}, we analyze the resources needed for the protocol, in terms of qubits and number of rounds. To do this, we need to have explicit expressions of $\omega,\Gamma$ as a function of $k$. Unfortunately, \autoref{easy} only shows the existence of $\omega,\Gamma$ without expliciting what those terms are. In \autoref{sec:resources}, we provide a  more detailed version of \autoref{easy} that allows us to keep track of the resources. Before that, we provide the proof of \autoref{lem:FinalSufficientCondition}.

\subsubsection{Proof of \autoref{lem:FinalSufficientCondition}} \label{fsc}

\begin{proof}

\noindent
\textit{Step 1: Approximation of $p$ by $f$.}
First note that, since $\Gamma \omega^2 > 1$, \hyperref[eqn:Mochon176]{Inequality~\ref*{eqn:Mochon176}} holds if the following inequality holds
\begin{align*}
	\sum_{j=\zeta}^{\Gamma\omega} p(j\omega) > \sum_{j=\zeta}^{\Gamma\omega} \frac{p(j\omega)}{j\omega}.
\end{align*}
Notice than for $j> \Gamma \omega$ we have $j\omega >1$. 

The first order of business is to show that the term $\prod_{i=1}^k \frac{\Gamma\omega+i\omega-j\omega}{\Gamma\omega+i\omega}\approx1$. We have
\begin{align*}
	1 \geq \prod_{i=1}^k \frac{\Gamma\omega+i\omega-j\omega}{\Gamma\omega+i\omega} \geq \left(1-\frac{j}{\Gamma}\right)^{k} \geq (1 - \omega)^k.
\end{align*}
Hence for all $j\in[\zeta,\Gamma\omega]$ we have
\begin{align*}
	(1-\omega)^k
	 \prod_{i=1}^{k-1} \frac{j\omega-(\alpha-i\omega)}{\alpha-i\omega}
	\prod_{i=-k}^k \frac{1}{j\omega + i\omega}
		< p(j\omega) <
	 \prod_{i=1}^{k-1} \frac{j\omega-(\alpha-i\omega)}{\alpha-i\omega}
	\prod_{i=-k}^k \frac{1}{j\omega + i\omega}.	
\end{align*}
We now bound all the term in the products individually,
\begin{align*}
	(1-\omega)^k
	\left(\frac{j\omega-\alpha}{\alpha}\right)^{k-1}
	\left(\frac{1}{j\omega+k\omega}\right)^{2k+1}
		< p(j\omega) <
	\left(\frac{j\omega-\alpha+k\omega}{\alpha-k\omega}\right)^{k-1}
	\left(\frac{1}{j\omega-k\omega}\right)^{2k+1}.	
\end{align*}
This means that the following inequality is a sufficient condition for \hyperref[eqn:Mochon176]{Inequality~\ref*{eqn:Mochon176}}:
\begin{align}
	\label{eqn:NoProduct}
	(1-\omega)^k
	\sum_{j=\zeta}^{\Gamma\omega} 
		\!\left(\frac{j\omega-\alpha}{\alpha}\right)^{\! k-1}
		\!\left(\frac{1}{j\omega+k\omega}\right)^{\! 2k+1} 
	> \sum_{j=\zeta}^{\Gamma\omega} 
		\!\left(\frac{j\omega-\alpha+k\omega}{\alpha-k\omega}\right)^{\! k-1}
		\!\left(\frac{1}{j\omega-k\omega}\right)^{\! 2k+1}
		\frac{1}{j\omega}.
\end{align}

We want to take the terms in $k\omega$ out of the sum in order to keep only terms in $j\omega$ and hence be able to replace the sums by integrals. For the LHS, we have:
\begin{align*}
	\frac{1}{j\omega+k\omega} 
		&= \frac{1}{j\omega} \cdot \frac{1}{(1+\frac{k\omega}{j\omega})}  > \frac{1}{j\omega} \cdot \frac{1}{1 + 2k\omega}
\end{align*}
since $j\omega \geq \alpha > 1/2$.

For the RHS, we first shift the sum in order to remove the term $k\omega$ in the numerator.
\begin{align*}
	\sum_{j=\zeta}^{\Gamma\omega} 
		\left(\frac{j\omega-\alpha+k\omega}{\alpha-k\omega}\right)^{k-1}
		\left(\frac{1}{j\omega-k\omega}\right)^{2k+1}
		\frac{1}{j\omega}
	&=
	\sum_{j=\zeta+k}^{\Gamma\omega+k} 
		\left(\frac{j\omega-\alpha}{\alpha-k\omega}\right)^{k-1}
		\left(\frac{1}{j\omega-2k\omega}\right)^{2k+1}
		\frac{1}{j\omega-k\omega} \\
	& < \sum_{j=\zeta+k}^{\Gamma\omega+k} 
		\left(\frac{j\omega-\alpha}{\alpha-k\omega}\right)^{k-1}
		\left(\frac{1}{j\omega-2k\omega}\right)^{2k+2}.
\end{align*}

We now bound the terms in the new RHS. We have
\begin{align}
	\label{eqn:Prefactor2}
	\frac{1}{\alpha-k\omega} < \frac{1}{\alpha} \cdot \frac{1}{1-2k\omega} < \frac{1}{\alpha} \cdot \frac{1}{1-4k\omega} \text{\quad and \quad}
	\frac{1}{j\omega-2k\omega} < \frac{1}{j\omega} \cdot  \frac{1}{1-4k\omega}.
\end{align}

Plugging these inequalities in the LHS and RHS of \autoref{eqn:NoProduct}, we get that the following inequality is a sufficient condition for \hyperref[eqn:Mochon176]{Inequality~\ref*{eqn:Mochon176}}. 
\begin{align*}
	\frac{(1-\omega)^k}{(1+2k\omega)^{2k + 1}}
	\sum_{j=\zeta}^{\Gamma\omega} 
		\left(\frac{j\omega-\alpha}{\alpha}\right)^{k-1}
		\left(\frac{1}{j\omega}\right)^{2k+1} 
	> \frac{1}{(1 - 4k\omega)^{(3k+1)}} \sum_{j=\zeta + k}^{\Gamma\omega + k} 
		\left(\frac{j\omega-\alpha}{\alpha}\right)^{k-1}
		\left(\frac{1}{j\omega}\right)^{2k+2}
\end{align*}
Consequently, using $(1 - \omega) \geq (1 - 4k\omega)$, we have that if the following inequality holds
 \begin{align*}
 	\sum_{j=\zeta}^{\Gamma\omega}
		\left(\frac{j\omega-\alpha}{\alpha}\right)^{k-1}
		\left(\frac{1}{j\omega}\right)^{2k+1}
	>
	E(\omega)
	\sum_{j=\zeta+k}^{\Gamma\omega+k}
		\left(\frac{j\omega-\alpha}{\alpha}\right)^{k-1}
		\left(\frac{1}{j\omega}\right)^{2k+2}
\end{align*}
with $E(\omega) = \frac{\left(1+2k\omega\right)^{2k+1}}{\left(1-4k\omega \right)^{4k+1}}$,
then \hyperref[eqn:Mochon176]{Inequality~\ref*{eqn:Mochon176}} also holds. 

{~}

\noindent
\textit{Step 2: Approximating sums by integrals.} The two sums can be approximately computed by replacing them by an integral, using the so-called rectangle method:
\begin{align*}
 \omega \sum_{j=\zeta}^{\Gamma\omega}
		\left(\frac{j\omega-\alpha}{\alpha}\right)^{k-1}
		\left(\frac{1}{j\omega}\right)^{2k+1}
	&>
	\int_\alpha^{\Gamma\omega^2} f(x)\d x - \frac{\Gamma\omega^2-\alpha}{24}\omega^2\abs{\fpp}_\infty, \\
	 \omega \sum_{j=\zeta+k}^{\Gamma\omega+k}
		\left(\frac{j\omega-\alpha}{\alpha}\right)^{k-1}
		\left(\frac{1}{j\omega}\right)^{2k+2}
	&<
	\int_{\alpha+k\omega}^{\Gamma\omega^2+k\omega} \tilde{f}(x)\d x  + \frac{\Gamma\omega^2-\alpha}{24}\omega^2\abs{\tilde\fpp}_\infty,
\end{align*}
where $\tilde{f}(x) = f(x)/x$. Notice that the error terms are upper-bounded by $\epsilon_l(\omega,\Gamma) = \Gamma\omega^4\abs{\fpp}_\infty$ and $\epsilon_r(\omega,\Gamma) = \Gamma\omega^4\abs{\tilde{\fpp}}_\infty$ respectively.

Last, observe that  $\int_{\alpha+k\omega}^{\Gamma\omega^2+k\omega} \tilde{f}(x)\d x  < \int_{\alpha}^{\infty} \tilde{f}(x)\d x$ since $\tilde{f} \geq 0$ on $[\alpha,+\infty)$.
\end{proof}

\section{Resources of the protocol}\label{sec:resources}

We now provide an alternative way to finishing the proof of existence provided in \autoref{conclude}, which allows us to analyze the resources needed for the protocol.
In \autoref{subsec:split}, we provide a  more detailed version of \autoref{easy} that allows us to keep track of the resources. In \autoref{subsec:resources}, we show that the number of qubits used for a protocol with bias $\eps$ is $O(\log \frac{1}{\eps})$, while the number of rounds is $\left(\frac{1}{\eps}\right)^{O(\frac{1}{\eps})}$.

\subsection{Refined analysis of strict validity of initial splits}\label{subsec:split}

Given \autoref{lem:FinalSufficientCondition}, we will find specific values for the parameters $\alpha, \omega$, and $\Gamma$ for which \autoref{eqn:FinalSufficientCondition} is satisfied. These parameters will allow us to calculate the resources necessary for the protocol. 
\begin{lem}\label{LastLemma}
For any $k$, by taking $\alpha = \frac{1}{2}+ \frac{c}{k}$ for some constant $c$, 
	$\omega = k^{-4}$ and $\Gamma = 2k^8$, \hyperref[eqn:FinalSufficientCondition]{Inequality~\ref*{eqn:FinalSufficientCondition}}  holds.
\end{lem}

Note first, that the values  $\alpha = \frac{1}{2}+\frac{c}{k}$ for some constant $c$, $\omega = k^{-4}$ and $\Gamma = 2k^8$ satisfy the assumptions of the \autoref{lem:FinalSufficientCondition}.

\begin{proof}
We are going to prove that for the chosen values of $\omega$ and $\Gamma$, it is sufficient to chose $\alpha > \frac{k+1}{2k+1} + O\!\left(\frac{1}{k^2}\right) = \frac{1}{2} + \frac{1}{4k+2}+O\!\left(\frac{1}{k^2}\right)$ for  \hyperref[eqn:FinalSufficientCondition]{Inequality~\ref*{eqn:FinalSufficientCondition}}  to hold.

We have
\begin{align*}
\int_\alpha^{\Gamma\omega^2} f(x) \d x 
	&= \int_\alpha^{\infty} f(x) \d x - \int_{\Gamma\omega^2}^\infty f(x) \d x 
	= \frac{B(k+1,k)}{\alpha^{2k+2}} - B(k+1,k)\frac{I(\alpha/\Gamma\omega^2;k+1,k)}{\alpha^{2k+2}} \\
	& = \frac{B(k+1,k)}{\alpha^{2k+2}}(1-I(\alpha/\Gamma\omega^2;k+1,k)).
\end{align*}
Recall that we have $\Gamma\omega^2 = 2$, and use \autoref{eqn:BetaWithInteger} to get
\begin{align*}
	1- I(\alpha/2;k+1,k) 
		= \sum_{j=0}^k  {2k \choose j} \left(\frac{\alpha}{2}\right)^j \left(1-\frac{\alpha}{2}\right)^{2k-j}.
\end{align*}
Since we have $\alpha/2 < 1/2$, we can use the Chernoff bound~\cite{Che52}:
\begin{align}
	1-I(\alpha/2;k+1,k)  > 1-\alpha^k e^{(k-k\alpha)} = 1 - e^{k(1-\alpha+\log\alpha)} > 1-e^{-\Omega(k)}. \label{eqn:integralBigO}
\end{align}

Let us now bound the error term $\epsilon_l(\omega,\Gamma)=\epsilon_l(k^{-4},2k^8) = 2k^{-8}\abs{\fpp}_\infty$. We have
\begin{align*}
	\fpp(x) 
		&= \frac{1}{\alpha^{k+1}}
			\left[ 
				(k-1)(k-2)\frac{(x-\alpha)^{k-3}}{x^{2k+1}}
				-2(k-1)(2k+1)\frac{(x-\alpha)^{k-2}}{x^{2k+2}}
				+(2k+1)(2k+2)\frac{(x-\alpha)^{k-1}}{x^{2k+3}}
			\right]	
\end{align*}
hence, $\abs{\fpp(x)} \leq O\!\left(\frac{k^2}{\alpha^k}\right) \frac{(x-\alpha)^{k-3}}{x^{2k+1}}$. We define $g(x) = \frac{(x-\alpha)^{k-3}}{x^{2k+1}}$ and we find that $g'(x) = 0 \Leftrightarrow x = \alpha\frac{2k+1}{k+4}$, thus 
\begin{align*}
	\abs{g}_\infty  = \frac{1}{\alpha^{k+2}}\left(\frac{k-3}{k+4}\right)^{k-3}\left(\frac{k+4}{2k+1}\right)^{2k+1}.
\end{align*}	
Recall that $\epsilon_l(\omega,\Gamma) = \Gamma\omega^4\abs{\fpp}_\infty$.
Using the fact that $\lim_{k\to\infty} \left(1+\Theta(1/k)\right)^k = \Theta(1)$, we get 
\begin{align*}
\abs{g}_\infty & = \frac{1}{\alpha^{k+2}}\left(\frac{k-3}{k+4}\right)^{k-3}\left(\frac{k+4}{2k+1}\right)^{2k+1} \\
& \le \frac{1}{\alpha^{k+2}} \frac{1}{2^{2k+1}} \left(1 + \frac{7}{2k+1}\right)^k \le O\!\left((4\alpha)^{-k}\right).
\end{align*}
From there, we have $\abs{f''}_\infty \leq O\!\left(\frac{k^2}{(2\alpha)^{2k}}\right)$ and in turn
\begin{align}
	\label{eqn:epsilonLBigO}
	\epsilon_l(k^{-4},2k^8) \leq 2k^{-8} \abs{f''}_\infty \leq O\!\left(k^{-6}(2\alpha)^{-2k}\right).
\end{align}
Using a similar technique, we also get
\begin{align}
	\label{eqn:epsilonRBigO}
	\epsilon_r(k^{-4},2k^8) \leq O\!\left(k^{-6}(2\alpha)^{-2k}\right).
\end{align}

The last term to bound is $E(\omega)= E(k^{-4}) = \frac{\left(1+2k^{-3}\right)^{2k+1}}{\left(1-4k^{-3} \right)^{4k+1}}$. Since $\left(1+O(k^{-3})\right)^k = 1 + O(k^{-2})$, we get
\begin{align}
	\label{eqn:CBigO}
	E(k^{-4}) = 1 + O\!\left(k^{-2}\right).
\end{align}
To conclude, in \hyperref[eqn:FinalSufficientCondition]{Inequality~\ref*{eqn:FinalSufficientCondition}}, we replace all the terms by their asymptotic equivalent computed in Equations~\ref{eqn:integralBigO}, \ref{eqn:epsilonLBigO}, \ref{eqn:epsilonRBigO}, and \ref{eqn:CBigO}: 
\begin{align*}
	\frac{B(k+1,k)}{\alpha^{2k+2}} > \frac{B(k+2,k)}{\alpha^{2k+3}} + O\!\left(\frac{1}{k^2}\right)\frac{B(k+2,k)}{\alpha^{2k+3}} + O\!\left(\frac{k^{-6}}{(2\alpha)^{2k}}\right).
\end{align*}
This means that \hyperref[eqn:FinalSufficientCondition]{Inequality~\ref*{eqn:FinalSufficientCondition}} holds for
\begin{align*}
	\alpha > \frac{k+1}{2k+1} + O\!\left(\frac{1}{k^2}\right) + O\!\left(\frac{1}{k^6 2^{2k} B(k+1,k)}\right) = \frac{1}{2} + O\!\left(\frac{1}{k}\right).
\end{align*}
We used Sterling's formula to show that the last term is $o(\frac{1}{k})$, since we have $\frac{1}{B(k+1,k)} = O\!\left(\frac{\sqrt{k}(2k/e)^{2k}}{(k/e)^{2k}}\right) = O(\sqrt{k}2^{2k})$. 
\end{proof}

\subsection{Bounds on the resources of the protocol}
\label{subsec:resources}

We are now ready to quantify the necessary resources for the protocol.
\begin{thm}\label{thm:resources}
For all $\eps >0$, there exists a quantum weak coin flipping protocol with cheating probabilities $P_A^* = P_B^* < 1/2 + \eps$ with $O(\log \frac{1}{\eps})$ qubits and $\left( \frac{1}{\eps} \right) ^{O(\frac{1}{\eps})}$ messages. 
\end{thm}

\begin{proof}
We consider the above TIPG with $k = \lceil\frac{c}{\eps}\rceil$,  $\omega = k^{-4}$, $\Gamma = 2k^8$ and $\alpha = \frac{1}{2}+ \frac{c}{k} = \frac{1}{2}+\eps$ where $c$ is a constant such that \autoref{thm:validsplits} holds with this $\alpha$ for each $k$.

{~} 

\textit{Step 1: Number of qubits.}
It is clear from the different steps of the construction that the number of qubits needed is equal to the logarithm of the number of different points in the point game. The number of points are no more than $O(\Gamma k) = O(k^9)$, which implies that the number of qubits is $O(\log \frac{1}{\eps})$. 

{~} 

\textit{Step 2: Number of rounds.} 
By \autoref{cor:numberOfRounds}, we need to to upper-bound $\norm{h}$ where \
\begin{align*}
	\norm{h}=1/2 + \sum_{j=\zeta}^{\Gamma} \sum_{\substack{i=-k \\ i\neq 0}}^{k}
		\frac	{C\cdot \abs{f(\left(j+i\right)\omega,\, j\omega)}}
			{((j+i)\omega)(j\omega) \underset{\substack{l \neq i \\ l \neq 0}}{\prod}\omega\abs{l-i}}.
\end{align*}
The other terms in the Corollary are all polynomial in $k$. We will do some crude approximations in the terms below, but the asymptotic behaviour is not going to change. The main property is that there are points in the ladder, for example the ones at the edges, whose weight is of the order $k^{O(k)}$.

Let us first bound $C$ defined by $\frac{1}{2C} = \sum_{\zeta}^\Gamma \frac{f(0,j\omega)}{\prod \omega(j+l)}$. Since all the terms in the sum are positive, we can lower bound the sum by the term for $j=\zeta$.
We have
\begin{align*}
	f(0, \zeta \omega) 
		&= \prod_{i=1}^{k-1}(\alpha-i\omega)(i\omega)
	              \prod_{i=1}^{k} (\Gamma\omega+i\omega)(\Gamma \omega + i\omega-\alpha) 
	        \geq (\alpha-k\omega)^{k-1} ( \omega)^{k-1} (\Gamma\omega)^k (\Gamma \omega -\alpha)^k \\
	        & \geq k^{\Omega(k)},
\end{align*}
and the denominator by,
\begin{align*}
	\prod_{l\neq0}(\alpha+ l \omega ) \leq (\alpha+k\omega)^{2k} = 2^{-O(k)}.
\end{align*}
This gives us $C \leq k^{-O(k)}$.
\COMMENT{
\begin{align*}
C & \leq \frac{(\Gamma \omega+k\omega)^{2k}}{(\Gamma \omega-\alpha)^{k} } = k^{O(k)}
\end{align*}
\begin{align*}
	\frac{1}{2C} 
		&\geq (\alpha-k\omega)^{k-1} (\Gamma\omega)^k \omega^k  
			\sum_{\zeta}^\Gamma \frac{ (j\omega-\alpha)^{k-1} }{(j\omega+k\omega)^{2k}}.
\end{align*}
Moreover, we know from previous calculations that 
\begin{align*}
	\omega \sum_{\zeta}^\Gamma \frac{ (j\omega-\alpha)^{k-1} }{(j\omega+k\omega)^{2k}} 
		= \Theta\!\left(  \int_\alpha^\infty \left(x-\alpha\right)^{k-1}x^{-2k} \d x\right) 
		= \Theta\!\left(\frac{B(k,k)}{\alpha^k}\right).
\end{align*}
Hence,
\begin{align*}
C &\leq O\!\left( \Gamma^{-k} \omega^{1-2k} B(k,k)^{-1} \right)
	\leq O\!\left( 2^{-k} k^{-16k-4+8k} \sqrt{k} 2^{2k} \right) \\
    &\leq O\!\left( 2^k k^{-7/2} k^{-8k}\right)
\end{align*}
}

Now, note that the norm is given by a sum with a number of terms polynomial in $k$, hence it suffices to give an upper bound on the ratio 
\begin{align*}
S  = 		\frac	{\abs{f(\left(j+i\right)\omega,\, j\omega)}}
			{((j+i)\omega)(j\omega) \underset{\substack{l \neq i \\ l \neq 0}}{\prod}\omega\abs{l-i}}.
\end{align*}
The numerator can be bounded by:
\begin{align*}
\abs{f((j+i)\omega,j\omega)}  
	&= \prod_{l=1}^{k-1}((j+i+l)\omega-\alpha)((j+l)\omega-\alpha)
	              \prod_{l=1}^{k} (\Gamma\omega+(l-i-j)\omega)(\Gamma\omega+(l-j)\omega)\\
	& \leq (\Gamma\omega + 2k\omega -\alpha)^{4k-2}  = k^{O(k)},
\end{align*}
and the denominator by,
\begin{align*}
	((j+i)\omega)(j\omega) \underset{\substack{l \neq i \\ l \neq 0}}{\prod}\omega\abs{l-i}
		\geq (\alpha-k\omega)\alpha\omega^{2k-1} (k!)^2
		\geq k^{-\Omega(k)}.
\end{align*}
Thus 
$S \leq k^{O(k)}$. 

Hence, we have proved that
$\norm{h} \leq k^{O(k)}$.
Using \autoref{cor:numberOfRounds}, we finally get that the number of rounds is upper-bounded by
$\left( \frac{1}{\eps}\right)^{O(\frac{1}{\eps})} $.
\COMMENT{
\begin{align*}
O\!\left( 2^k (4e)^{2k} k^{3} k^{30k} \right) = k^{O(k)} = \left( \frac{1}{\eps}\right)^{O(\frac{1}{\eps})} .
\end{align*}
}
\end{proof}

We see that the protocol is very efficient in the number of qubits that it uses. However, our analysis shows that the number of rounds is exponential. It could be the case that by choosing different values for $\omega$ and $\Gamma$, one can reduce the number of rounds. Nevertheless, our intuition, backed with numerical evidence, is that for this protocol one would always need an exponential number of rounds. Another way would be to try to find a more efficient way to turn a TIPG into a point game with valid transitions. Last, it remains open to find a simpler and more efficient point game that can be easily transformed into an easy-to-describe protocol.

\section*{Acknowledgments}
We would like to thank Peter H\o{}yer for useful comments on a preliminary version of the paper. 

\bibliographystyle{alphaurl}
\bibliography{wcf}

\end{document}